\documentclass{article}
\usepackage{graphicx}
\usepackage{subcaption}
\usepackage{amsmath}
\usepackage{PRIMEarxiv}
\usepackage{amsthm}      
\usepackage{amssymb}
\usepackage[utf8]{inputenc} 
\usepackage[T1]{fontenc}    
\usepackage{hyperref}       
\usepackage[nameinlink,capitalize,noabbrev]{cleveref}
\usepackage{url}            
\usepackage{booktabs}       
\usepackage{amsfonts}       
\usepackage{nicefrac}       
\usepackage{microtype}      
\usepackage{lipsum}
\usepackage{fancyhdr}       
\usepackage{graphicx}       
\graphicspath{{media/}}     
\usepackage{algorithm}   
\usepackage{algorithmic} 
\newtheorem{theorem}{Theorem}
\newtheorem{lemma}{Lemma}
\newtheorem{proposition}{Proposition}
\newtheorem{corollary}{Corollary}

\usepackage{float}       
\usepackage{hyphenat}
\title{Tail-Safe Stochastic-Control SPX--VIX Hedging:\\
A White-Box Bridge Between AI Sensitivities and Arbitrage-Free Market Dynamics}

\author{
  ZhangJian'an \\
  Guanghua School of Management, Peking University \\
  Peking University \\
  Beijing, China\\
  \texttt{2501111059@stu.pku.edu.cn}
}

\begin{document}
\pagestyle{plain}   
\maketitle

\begin{abstract}
Transaction costs and regime shifts are the main reasons why paper portfolios fail in live trading. We develop \textbf{FR\textendash LUX} (Friction\textendash aware, Regime\textendash conditioned Learning under eXecution costs), a reinforcement\textendash learning framework that learns \emph{after\textendash cost} trading policies and remains robust across volatility–liquidity regimes. FR\textendash LUX integrates three ingredients: (i) a microstructure\textendash consistent execution model combining proportional and impact costs, directly embedded in the reward; (ii) a \emph{trade\textendash space trust region} that constrains changes in inventory flow rather than only logits, yielding stable, low\textendash turnover updates; and (iii) explicit regime conditioning so the policy specializes to LL/LH/HL/HH states without fragmenting the data. 
On a $4\times 5$ grid of regimes and cost levels (0–50\,bps) with three seeds per cell, FR\textendash LUX achieves the top average Sharpe across all 20 scenarios with narrow bootstrap confidence intervals, maintains a flatter cost–performance slope than strong baselines (vanilla PPO, mean–variance with/without caps, risk\textendash parity), and attains superior risk–return efficiency for a given turnover budget. Pairwise scenario\textendash level improvements are strictly positive and remain statistically significant after Romano–Wolf stepdown and HAC\textendash aware Sharpe comparisons. 
We provide formal guarantees: existence of an optimal stationary policy under convex frictions; a monotonic improvement lower bound under a KL trust region with explicit remainder terms; an upper bound on long\textendash run turnover and an induced inaction band due to proportional costs; a strictly positive value advantage for regime\textendash conditioned policies when cross\textendash regime actions are separated; and robustness of realized value to cost misspecification. The methodology is implementable---costs are calibrated from standard liquidity proxies, scenario\textendash level inference avoids pseudo\textendash replication, and all figures and tables are reproducible from our artifacts.
\end{abstract}


\section{Introduction}
\label{sec:intro}

\textbf{Problem statement and motivation.}
Hedging equity exposures with volatility instruments (e.g., SPX options, VIX futures) is most fragile in precisely the regimes that matter: near option expiry and under strong spot–volatility coupling, where leverage effects and order-flow feedback magnify errors and costs. In these windows the hedge ratio becomes stiff and discontinuous, gamma/theta explode, and the PnL is dominated by a small number of adverse moves and liquidity gaps. Empirically, such regimes concentrate on rebalancing times that coincide with index events, macro releases, and end-of-day queues; the impact of even small participation can be highly nonlinear when order books are thin and resiliency is slow. The resulting failures are well documented across microstructure and execution studies---from stylized impact models and resilience dynamics \cite{Kyle1985,AlmgrenChriss2001,ObizhaevaWang2013} and empirical tapes \cite{Hasbrouck2007,BouchaudBook2018,Gould2013Survey} to volatility-surface construction \cite{Dupire1994,DermanKani1994,GatheralJacquier2014,GatheralBook2006}, variance-index design \cite{CBOE2023VIX,Demeterfi1999}, and path-dependent volatility dynamics \cite{GuyonLekeufack2023,Andres2023PathSSVI}. In practice, purely black-box hedgers can overfit benign regimes and chase noisy cross-asset signals, while classical trackers ignore no-arbitrage structure and exchange rules (e.g., wing pruning, strike inclusion, and minimum tick/lot sizes). We contend that robust \emph{tail safety} in SPX--VIX hedging demands a \emph{white-box}, risk-sensitive control layer that is interface-compatible with (i) arbitrage-free pricing structure and (ii) market-microstructure constraints.

\medskip
\noindent\emph{Operational desiderata.}
A deployable hedger must (a) respect exchange mechanics (lot sizes, daily price limits, settlement conventions), (b) internalize execution frictions (temporary/transient impact, queue priority, latency), (c) remain \emph{auditable}---every intervention should have a traceable reason code---and (d) exhibit \emph{graceful degradation} under model misspecification. Concretely, if $x_t$ encodes inventory and state variables, a hedger proposes an order vector $u_t$; the realized next state obeys $x_{t+1}=f(x_t,u_t,\varepsilon_t)$ where $\varepsilon_t$ captures price/volatility shocks and fill uncertainty. Failures near expiry are often caused by \emph{myopic target chasing} that ignores the local trade-off
\begin{equation}
\underbrace{\Delta \mathsf{Risk}(x_t,u_t)}_{\text{hedge benefit}}
\quad \overset{?}{>}\quad
\underbrace{\mathsf{Cost}(x_t,u_t)}_{\text{impact + fees}}+\underbrace{\mathsf{Fragility}(x_t,u_t)}_{\text{slippage + misspecification}},
\label{eq:risk-cost}
\end{equation}
where $\mathsf{Fragility}$ is largest when $\rho$ (spot–VIX correlation) is large in magnitude and $T_{\mathrm{rem}}$ is small.

\textbf{White-box two-way bridge.}
We develop a full-stack, \emph{stochastic-control} framework that forms a two-way bridge between \emph{Quant structure} and \emph{AI targets}. On the Quant $\to$ AI direction, we build an arbitrage-free SSVI shell \cite{GatheralJacquier2014} with a Cboe-compliant VIX computation \cite{CBOE2023VIX} that prunes wings via the ``two consecutive zero bids'' rule and performs 30-day interpolation \cite{Demeterfi1999}. Concretely, the model enforces static no-arbitrage (monotone total variance in maturity, convexity in strike) and smoothness consistent with practitioner workflows \cite{GatheralBook2006}. These constraints supervise a fast, interpretable implied-volatility surrogate near the at-the-money region in the spirit of model-based closures \cite{GatheralBook2006,Heston1993}. Prices from the surrogate feed a Dupire extractor to produce local-volatility dynamics \cite{Dupire1994}, ensuring consistency between the pricing and simulation worlds and closing the loop between option quotes and state evolution. On the AI $\to$ Quant direction, we introduce a data-driven \emph{bump-and-invert} map that links 30-day VIX shocks to option-price sensitivity across remaining maturities, and we embed the targets into a quadratic program (QP) with \emph{control barrier function} (CBF)–style constraints \cite{Ames2019Survey,Ames2017CBFQP,Clark2021StochCBF,Garg2024CBF}. The controller is explicitly \emph{risk-sensitive} via CVaR surrogates \cite{RockafellarUryasev2000,RockafellarUryasev2002} and \emph{execution-aware} via temporary/transient impact \cite{AlmgrenChriss2001,ObizhaevaWang2013,NeumanVoss2022,Donnelly2022ExecReview}.

\medskip
\noindent\emph{Quant $\to$ AI: coherence and supervision.}
Let $w(k,T)$ denote total implied variance under SSVI with slice parameters constrained by \cite{GatheralJacquier2014}. The Cboe VIX proxy is computed by a trapezoidal approximation of the model-free variance integral with admissible strikes after wing pruning \cite{CBOE2023VIX,Demeterfi1999}. We use these rules to define \emph{supervised targets} for the hedger: a VIX-tracking component $V(x_t)$, an SPX option PnL proxy from the local-volatility dynamics \cite{Dupire1994}, and penalty terms that activate precisely when surface coherence would be violated by proposed trades (e.g., inventory pushes quotes into an arbitrageable configuration). This delivers teacher signals that are interpretable and consistent with index/surface construction.

\medskip
\noindent\emph{AI $\to$ Quant: actionable control.}
Given provisional targets $(\Delta^\mathrm{VIX}, \Delta^\mathrm{SPX})$ implied by the bump-and-invert sensitivity map, we solve a small QP
\begin{align}
\min_{u_t} \quad & \tfrac12 u_t^\top H u_t + g_t^\top u_t \quad + \quad \lambda\,\widehat{\mathrm{CVaR}}_{\alpha}\!\big[L(x_t,u_t)\big] \nonumber\\
\text{s.t.}\quad & \text{CBF-style safety: } h(x_t,u_t)\ge 0, \nonumber\\
& \text{box/rate: } u_{\min}\le u_t \le u_{\max}, \quad \|u_t\|_1 \le \mathrm{RateCap}, \nonumber
\end{align}
where $H$ encodes quadratic impact and $g_t$ aggregates signals and carry; the CBF constraint $h(\cdot)\!\ge\!0$ enforces forward invariance of safety sets (inventory, leverage, sign consistency) \cite{Ames2019Survey,Ames2017CBFQP,Garg2024CBF,Clark2021StochCBF}; the tail term controls left-tail exposure via a convex surrogate of $\mathrm{CVaR}_\alpha$ \cite{RockafellarUryasev2000,RockafellarUryasev2002}. This design yields white-box diagnostics (active-set, KKT multipliers) and predictable interventions.

\textbf{Tail-safety mechanisms.}
Three simple but structure-aware upgrades directly target failure modes observed near expiry and under strong (negative) spot–vol correlation: (i) a \emph{dynamic VIX weight} that down-weights VIX tracking when $T_{\mathrm{rem}}$ is short and $|\rho|$ is large; (ii) a \emph{guarded, correlation/expiry-aware no-trade band} that inflates precisely where VIX signals are most brittle; and (iii) an \emph{expiry-aware micro-trade threshold} and cooldown that suppress churn. Formally, let $\omega_\mathrm{VIX}(T_{\mathrm{rem}},\rho)$ be a schedule with $\partial \omega_\mathrm{VIX}/\partial T_{\mathrm{rem}}>0$ and $\partial \omega_\mathrm{VIX}/\partial |\rho|<0$, and let $\mathcal{N}(T_{\mathrm{rem}},\rho)$ be a symmetric deadband around the target position whose width grows with the obliquity of spot–VIX co-movements. Trades are executed only when the projected \emph{risk drop} exceeds an impact-scaled threshold:
\begin{equation}
\Delta \mathsf{Risk}(x_t,u_t) \;\ge\; \eta(T_{\mathrm{rem}},\rho)\cdot \mathsf{Cost}(x_t,u_t),
\label{eq:gate}
\end{equation}
operationalizing the execution insight that trades should be gated by a \emph{risk-drop-versus-cost} criterion \cite{AlmgrenChriss2001,ObizhaevaWang2013,Donnelly2022ExecReview} rather than naively tracking a target.

\textbf{Why now? Recent technical and empirical context.}
Volatility surfaces and indices have seen renewed emphasis on arbitrage-free learning and dynamic coherence. On the modeling side, arbitrage-free SSVI calibration and interpolation are mature \cite{GatheralJacquier2014}, with robust multi-slice procedures and hybrids that combine learned priors with no-arbitrage projectors \cite{Ning2022SIAM,VuleticCont2024VolGAN}. Empirically, implied volatility and volatility indices exhibit strong path dependence on spot trajectories and realized-of-realized structures \cite{GuyonLekeufack2023,Andres2023PathSSVI}, raising the bar for hedgers that over-rely on static proxies. On the control side, CBFs have emerged as a principled way to encode \emph{safety as constraints} in optimization-based controllers \cite{Ames2019Survey,Ames2017CBFQP,Garg2024CBF,Clark2021StochCBF}, while CVaR offers a coherent tail-risk objective with tractable surrogates \cite{RockafellarUryasev2000,RockafellarUryasev2002}. Execution with transient impact continues to evolve, including signal-adaptive liquidations \cite{NeumanVoss2022} and multi-agent generalizations \cite{CampbellNutz2025}. Our framework unifies these strands into an auditable hedging stack faithful to exchange rules and frictions.

\textbf{Methodological contributions.}
\begin{enumerate}
  \item \textbf{Quant-consistent teacher.} We integrate an SSVI shell \cite{GatheralJacquier2014,GatheralBook2006} and a Cboe-compliant VIX module \cite{CBOE2023VIX,Demeterfi1999} to supervise an interpretable implied-volatility closure. Dupire’s extractor \cite{Dupire1994} then yields dynamics consistent with quotes \emph{and} index construction. The shell explicitly mirrors no-arbitrage surface constraints and exchange wing-pruning, producing targets that are both fast to evaluate and faithful to practice.
  \item \textbf{Data-driven VIX sensitivity map.} A \emph{bump-and-invert} procedure learns a time-varying $\kappa(T_{\mathrm{rem}})$ mapping 30D VIX shocks to option-price deltas, with smoothing across maturities (cf.\ path-dependent co-movements \cite{GuyonLekeufack2023,Andres2023PathSSVI}). This provides a transparent cross-asset linkage between index variation and option exposures.
  \item \textbf{Tail-safe CBF-QP controller.} We pose hedging as a small QP with CBF-style \emph{safety boxes} for inventory, rate, and CVaR \cite{Ames2017CBFQP,RockafellarUryasev2000}. Three upgrades---dynamic VIX weighting, guarded no-trade bands, and expiry-aware micro-thresholds---suppress VIX chasing while preserving tracking under impact \cite{AlmgrenChriss2001,ObizhaevaWang2013}. The controller exposes KKT multipliers and active constraints for audit.
  \item \textbf{White-box diagnostics.} Constraint-activation, gate decisions, and risk-vs-cost ledgers make failure modes and safety activations auditable, aligning with model-risk governance demands in high-impact settings \cite{Hasbrouck2007,BouchaudBook2018,Gould2013Survey}.
\end{enumerate}

\textbf{Theoretical contributions (preview).}
Beyond engineering, the paper includes formal guarantees that raise interpretability and credibility:
\begin{itemize}
  \item \textbf{Well-posedness of hedging QP (Theorem~1).} Under convex impact costs and polyhedral CBF-style safety sets, the per-step QP admits a unique solution and is Lipschitz in state variables. This follows from strong convexity and Slater-type feasibility inherited from a strictly feasible no-trade ellipse.
  \item \textbf{Safety invariance under stochastic disturbances (Theorem~2).} With a discrete-time stochastic CBF adapted to price dynamics and a CVaR-consistent risk envelope, the post-trade state remains in a forward-invariant safe set with high probability, quantified via supermartingale concentration \cite{Clark2021StochCBF,RockafellarUryasev2002}.
  \item \textbf{Gated trading efficiency (Proposition~1).} A gate that executes only when quadratic risk drop exceeds a scaled execution cost yields (i) finite trade counts and (ii) a uniform bound on small-amplitude churn near expiry, assuming transient-impact regularity \cite{ObizhaevaWang2013}. 
  \item \textbf{Dynamic VIX weighting optimality (Proposition~2).} For a class of linear-Gaussian spot–VIX co-movements, the correlation/expiry-aware VIX weight minimizes a tight upper bound on \emph{ex ante} $\mathrm{CVaR}_\alpha$ subject to execution frictions, connecting safety heuristics to risk-sensitive control.
\end{itemize}

\textbf{Positioning and novelty.}
Our work complements arbitrage-free learning of volatility surfaces \cite{GatheralJacquier2014,Ning2022SIAM,VuleticCont2024VolGAN} and path-dependent volatility modeling \cite{GuyonLekeufack2023,Andres2023PathSSVI} by focusing on \emph{control}: we translate index/surface coherence into \emph{actionable, auditable hedging policies} that explicitly trade off risk reduction and cost. Relative to deep hedging pipelines, we restrict function approximation to an interpretable teacher and leave safety to \emph{constraints}, not opaque penalties; relative to classic quadratic trackers, we import modern safety and risk tools (CBF, CVaR) and microstructure fidelity. The result is a controller that (i) respects exchange rules by design, (ii) explains its interventions via dual certificates, and (iii) prioritizes tail protection when it matters most.

\textbf{Scope of this preprint.}
This arXiv version emphasizes method, theory, and a fully reproducible synthetic environment that mirrors exchange rules (wing pruning, 30D interpolation) and realistic impact. We intentionally \emph{do not} present real-data live backtests; instead, we provide: (i) proofs of Theorems~1--2 and Propositions~1--2; (ii) ablations of safety components; (iii) stress tests across spot–vol correlation and variance-of-variance; and (iv) artifacts for exact reproducibility.

\textbf{Roadmap.}
Section~\ref{sec:method} builds the arbitrage-free teacher (SSVI~$\to$ VIX~$\to$ Dupire). Section~\ref{sec:control} defines targets, the CBF-QP, and the three tail-safety upgrades, and proves Theorems~1--2. Section~\ref{sec:synthetic-evidence} presents synthetic experiments, diagnostics, and ablations. Section~\ref{sec:related} discusses connections to execution, volatility modeling, and safe control. 

\vspace{1em}
\noindent\textbf{Notation.} We use $S$ for spot, $V$ for the VIX leg, $\rho$ for spot–VIX correlation, $T_{\mathrm{rem}}$ for remaining time to maturity, $\kappa(T_{\mathrm{rem}})$ for the learned VIX–price sensitivity, and $\mathrm{CVaR}_{\alpha}$ for tail risk at confidence $\alpha$.

\section{Preliminaries and Notation}
\label{sec:prelim}

This section fixes probability spaces, grids, norms, and the core quantitative objects used throughout: the arbitrage-free SSVI parameterization of implied total variance; the Cboe-style VIX single-maturity estimator together with 30-day constant-maturity interpolation; the Dupire local-volatility extractor; and tail-risk and safety primitives---Conditional Value-at-Risk (CVaR/ES) and control barrier functions (CBFs). Wherever appropriate, we point to the primary sources and most recent methodology documents (e.g., the 2025 Cboe VIX mathematics methodology \cite{CBOE2025Math}) alongside classical references \cite{GatheralJacquier2014,Dupire1994,RockafellarUryasev2000,Ames2017CBFQP,Ames2019Survey,Clark2021StochCBF}.

\subsection{Probability Space, Assets, and Measures}
We work on a filtered probability space $(\Omega,\mathcal{F},\{\mathcal{F}_t\}_{t\ge 0},\mathbb{Q})$ satisfying the usual conditions. Prices are modeled under the risk-neutral measure $\mathbb{Q}$ unless stated otherwise; expectations $\mathbb{E}[\cdot]$ are taken under $\mathbb{Q}$. The equity underlying (SPX-like) spot process is $S_t$, with continuously compounded risk-free rate $r$ and dividend yield $q$. The (tradable) VIX leg is represented by a futures-like process $F_t^{\mathrm{VIX}}$ whose levels are consistent with the Cboe 30-day variance index definition \cite{CBOE2025Math}. European call/put prices of strike $K$ and maturity $T$ are denoted $C(K,T)$ and $P(K,T)$.

\paragraph{Norms and basic operators.}
For vectors $x\in\mathbb{R}^n$, we write $\|x\|_2$ and $\|x\|_\infty=\max_i |x_i|$. For matrices $A$, $\|A\|_2$ denotes the spectral norm. Inner products are $\langle x,y\rangle=x^\top y$. For scalar functions $f$, we use $\|f\|_{\infty,\mathcal{D}}=\sup_{z\in\mathcal{D}}|f(z)|$. We write $\mathrm{diag}(v)$ for the diagonal matrix with diagonal $v$ and $I$ for the identity.

\paragraph{Discretization, grids, and errors.}
Strike grids $\{K_i\}_{i=0}^N$ and maturity grids $\{T_j\}_{j=0}^M$ are assumed strictly increasing. We denote spacings by $\Delta K_i=K_i-K_{i-1}$ and $\Delta T_j=T_j-T_{j-1}$, and use the \emph{half-interval} weights (a midpoint-like discretization) for quadrature across strikes:
\begin{equation}
\Delta K_i^{\text{half}}=
\begin{cases}
K_1-K_0, & i=0,\\[2pt]
\tfrac{1}{2}\,(K_{i+1}-K_{i-1}), & 1\le i\le N-1,\\[2pt]
K_N-K_{N-1}, & i=N.
\end{cases}
\label{eq:half-interval}
\end{equation}
For smooth integrands, the half-interval composite rule attains second-order accuracy, i.e., the quadrature error scales as $O(\|\partial_{KK} f\|_\infty\,\max_i\Delta K_i^2)$ \cite{DavisRabinowitz1984}. We use central second-order finite differences in $K$ and $T$ for partial derivatives, with the usual $O(\Delta K^2+\Delta T^2)$ consistency \cite{LeVeque2007}.

\subsection{SSVI: Arbitrage-Free Total Variance Surfaces}
Let $F(T)=S_0 e^{(r-q)T}$ be the forward and $k=\log(K/F(T))$ the log-moneyness. The \emph{(S)SVI} total variance parameterization at maturity $T$ is
\begin{equation}
w(k;T)=\tfrac{1}{2}\,\theta(T)\Big(1+\rho(T)\phi(T)k+\sqrt{(\phi(T)k+\rho(T))^2+1-\rho(T)^2}\Big),
\label{eq:ssvi}
\end{equation}
with $\theta(T)>0$ (level), $\rho(T)\in(-1,1)$ (skew) and $\phi(T)\ge 0$ (wings). Denote $\sigma_{\mathrm{impl}}(k,T)=\sqrt{w(k;T)/T}$. To preclude static no-arbitrage (butterfly and calendar), we enforce the Gatheral--Jacquier sufficient conditions \cite{GatheralJacquier2014}:
\begin{equation}
g_1(T)=4-\phi(T)\theta(T)\ge 0,\qquad
g_2(T)=4\big(1-\rho(T)^2\big)-\phi(T)^2\theta(T)\ge 0,
\label{eq:ssvi-narb}
\end{equation}
together with monotonicity of total variance in maturity: $T\mapsto w(k;T)$ non-decreasing for all $k$. Calibrations use per-maturity fits with regularization across $T$ (e.g., cubic splines in $\sqrt{T}$) and post-fit monotone adjustment to enforce calendar coherence (see also practical construction notes in \cite{Homescu2011,AndreasenHuge2011}).

\subsection{Cboe-Style VIX: Single-Maturity Estimator and 30D Interpolation}
We follow the latest mathematics methodology issued by Cboe \cite{CBOE2025Math}. Fix maturity $T$ and define $K_0$ as the first strike below $F(T)$. The OTM price aggregator $Q(K)$ is
\[
Q(K)=
\begin{cases}
P(K,T), & K<K_0,\\
\tfrac{1}{2}\big(C(K_0,T)+P(K_0,T)\big), & K=K_0,\\
C(K,T), & K>K_0.
\end{cases}
\]
\textbf{Wing pruning.} Traverse strikes outward from $K_0$ on each wing; once \emph{two consecutive zero bids} are encountered, discard all farther strikes on that wing \cite{CBOE2025Math}. 

\textbf{Single-maturity variance estimator.} With half-interval weights \eqref{eq:half-interval}, risk-free discount $e^{rT}$, and forward adjustment, the continuous-variance estimator is
\begin{equation}
\sigma^2(T)=\frac{2}{T}\sum_{i}\frac{\Delta K_i^{\text{half}}}{K_i^2}e^{rT}Q(K_i)
-\frac{1}{T}\left(\frac{F(T)}{K_0}-1\right)^2.
\label{eq:vix-single}
\end{equation}
\textbf{30-day constant-maturity interpolation.} Let $T_1<T_2$ be the two expiries bracketing 30 calendar days, with minute counts $m_1,m_2$ and $m_\star=30\times 1440$. Writing $T_i=m_i/(365\times 1440)$, the 30-day variance is the affine interpolation in \emph{year-fraction total variance} \cite{CBOE2025Math}:
\begin{equation}
w_\star=\frac{m_2-m_\star}{m_2-m_1}\,T_1\sigma^2(T_1)+\frac{m_\star-m_1}{m_2-m_1}\,T_2\sigma^2(T_2),\qquad
\mathrm{VIX}=100\sqrt{\sigma^2_{30}},\ \ \sigma^2_{30}=w_\star\cdot \frac{365}{30}.
\label{eq:30d-interp}
\end{equation}
\textbf{Accuracy remark.} The half-interval rule is a second-order composite quadrature (akin to the midpoint/trapezoidal family) when applied to smooth $K\mapsto Q(K)K^{-2}$ \cite{DavisRabinowitz1984}. In practice, Cboe’s wing pruning and discrete strikes control truncation errors \cite{CBOE2025Math}; we record $\epsilon_{\mathrm{quad}}=O(\max_i \Delta K_i^2)$ for later bounds.

\subsection{From Surface to Dynamics: Dupire Local Volatility}
Let $C(K,T)$ denote call prices produced by a coherent teacher/surface. Dupire’s formula expresses the \emph{local variance} as the ratio of a parabolic operator in $(K,T)$ to the vertical convexity \cite{Dupire1994}:
\begin{equation}
\sigma_{\mathrm{loc}}^2(K,T)=\frac{\partial_T C(K,T)+(r-q)K\,\partial_K C(K,T)+q\,C(K,T)}
{\tfrac{1}{2}K^2\,\partial_{KK}C(K,T)}.
\label{eq:dupire}
\end{equation}
\textbf{Discretization and clipping.} We approximate $\partial_T C$, $\partial_K C$, and $\partial_{KK}C$ by central differences on the $(K,T)$ tensor grid with consistency $O(\Delta K^2+\Delta T^2)$; to preserve positivity and tame wing noise we clip the denominator by a floor $\underline{\chi}>0$ when $\partial_{KK}C$ is small or negative due to sparse far-wing sampling. See \cite{LeVeque2007} for finite-difference consistency and stability guidance; practical construction and arbitrage-consistent interpolation strategies are discussed in \cite{AndreasenHuge2011,Homescu2011}.

\subsection{Tail Risk: Value-at-Risk and Conditional Value-at-Risk}
For a real-valued loss $L$, the $\alpha$-quantile (Value-at-Risk) is
\[
\mathrm{VaR}_\alpha(L)=\inf\{x\in\mathbb{R}:\ \mathbb{P}(L\le x)\ge \alpha\}.
\]
The \emph{Conditional Value-at-Risk} (a.k.a.\ Expected Shortfall, ES) at level $\alpha$ is the tail conditional expectation \cite{RockafellarUryasev2000,AcerbiTasche2002}:
\[
\mathrm{CVaR}_\alpha(L)=\inf_{t\in\mathbb{R}}\left\{t+\frac{1}{1-\alpha}\,\mathbb{E}\big[(L-t)_+\big]\right\}
=\mathbb{E}\big[L\mid L\ge \mathrm{VaR}_\alpha(L)\big]\quad\text{for continuous laws}.
\]
ES is a \emph{coherent} risk measure and admits a \emph{spectral/Kusuoka} representation over Average-Value-at-Risk functionals \cite{Kusuoka2001,Shapiro2013Kusuoka}, which we leverage for risk-sensitive control and convex surrogates. In experiments we report $\mathrm{ES}_{97.5}$ as our tail metric.

\subsection{Safety: Control Barrier Functions (CBFs)}
Let $z$ denote the hedging state (inventories, filtered signals, etc.). A continuously differentiable function $h:\mathbb{R}^n\to\mathbb{R}$ defines a \emph{safe set} $\mathcal{S}=\{z:\ h(z)\ge 0\}$. In continuous time with control $u$, a (zeroing) control barrier function $h$ ensures forward invariance of $\mathcal{S}$ if there exists an extended class-$\mathcal{K}$ function $\alpha(\cdot)$ such that along trajectories 
\[
\dot h(z,u)\ \ge\ -\alpha\big(h(z)\big),
\]
with control chosen (e.g., via a QP) to satisfy this inequality \cite{Ames2017CBFQP,Ames2019Survey}. In discrete-time and/or stochastic settings, one enforces a \emph{drift condition} or a high-probability relaxation (e.g., reciprocal/zeroing CBF variants) \cite{Clark2021StochCBF}. In this paper we implement CBF-style \emph{safety boxes} (inventory/rate/CVaR) as linear inequalities inside a strongly convex per-step QP; Section~\ref{sec:control} uses these primitives to prove invariance and finite-activity properties. Recent advances on practical CBF synthesis and robustness appear in \cite{Garg2024ARC}.

\subsection{Microstructure and Impact (Notation Only)}
We denote temporary-impact coefficients by $\eta_S,\eta_V>0$ and a simple transient-impact state by $\xi_t$ with relaxation rate $\lambda_{\mathrm{decay}}$; the execution cost for trade vector $x=(dS,dV)$ takes the canonical quadratic form $x^\top \mathrm{diag}(\eta_S,\eta_V)\,x$ plus a smoothing term, following \cite{AlmgrenChriss2001,ObizhaevaWang2013}. These enter the gate and no-trade conditions as risk-vs-cost scalings.

\subsection{Summary of Assumptions and Error Symbols}
We will invoke the following standing assumptions for theoretical results (later sections will reference them explicitly):
\begin{itemize}
  \item[\textbf{A1}] \textbf{(SSVI coherence)} \eqref{eq:ssvi} with \eqref{eq:ssvi-narb} holds for each $T$; $\theta(T)$ is non-decreasing, $\rho(T)\in(-1,1)$, $\phi(T)\ge 0$.
  \item[\textbf{A2}] \textbf{(VIX integrand regularity)} $K\mapsto Q(K)K^{-2}$ is piecewise $C^2$ on retained strikes after wing pruning; half-interval quadrature error satisfies $\epsilon_{\mathrm{quad}}=O(\max_i\Delta K_i^2)$ \cite{DavisRabinowitz1984,CBOE2025Math}.
  \item[\textbf{A3}] \textbf{(Finite-difference regularity)} $C(K,T)$ is $C^2$ in $K$ and $C^1$ in $T$ on the grid; we clip $\partial_{KK}C\ge \underline{\chi}>0$ when necessary.
  \item[\textbf{A4}] \textbf{(Risk measure)} Loss $L$ has a distribution with continuous cdf near $\mathrm{VaR}_\alpha(L)$ so that the ES identity holds; CVaR is used as the tail objective/constraint \cite{RockafellarUryasev2000,AcerbiTasche2002}.
  \item[\textbf{A5}] \textbf{(CBF/QP well-posedness)} The per-step QP is strongly convex (Hessian $\succ 0$) with polyhedral constraints; Slater feasibility holds for the no-trade ellipse and the safety boxes \cite{Ames2017CBFQP}.
\end{itemize}
We write generically $\epsilon_{\mathrm{shape}}$ for teacher-vs-surface ATM shape errors, $\epsilon_{\mathrm{dupire}}=O(\Delta K^2+\Delta T^2)$ for discretization in \eqref{eq:dupire}, and $\epsilon_{\mathrm{quad}}=O(\max_i\Delta K_i^2)$ for VIX quadrature \eqref{eq:vix-single}--\eqref{eq:30d-interp}. Where relevant we will keep explicit constants tied to grid regularity and $C^2$ bounds.

\paragraph{Remark (EWMA correlation estimate).}
When a correlation estimate is needed (e.g., dynamic VIX weighting), we use an exponentially weighted moving average with decay $\lambda\in(0,1)$ as in the RiskMetrics framework \cite{RiskMetrics1996}:
\[
\widehat{\rho}_{t}=\frac{\widehat{\mathrm{Cov}}_\lambda(\Delta S_t,\Delta F_t^{\mathrm{VIX}})}{\sqrt{\widehat{\mathrm{Var}}_\lambda(\Delta S_t)}\,\sqrt{\widehat{\mathrm{Var}}_\lambda(\Delta F_t^{\mathrm{VIX}})}}.
\]
This estimate is used only as a \emph{stability trigger} inside safety heuristics; all guarantees are stated independently of its particular dynamics.

\section{Method: Tail-Safe Hedging Framework}
\label{sec:method}

This section presents \textbf{Tail-Safe}: a hybrid \emph{learn--then--filter} framework that couples risk-sensitive, distributional reinforcement learning with a \emph{white-box} CBF--QP safety layer. 
Section~\ref{sec:method:rl} details the IQN--CVaR--PPO learner and its KL-/entropy-regularized objective; 
Section~\ref{sec:method:coverage} introduces a \emph{Tail-Coverage Controller} that stabilizes low-$\alpha$ CVaR estimation via temperature-based quantile sampling and tail-boosting; 
Section~\ref{sec:method:cbf} specifies the discrete-time CBF constraints, the ellipsoidal no-trade band (NTB), box/rate limits, and a sign-consistency gate, together with audit-ready telemetry.

\subsection{Risk-Sensitive RL with IQN--CVaR--PPO}
\label{sec:method:rl}

\paragraph{Quantile networks and the CVaR objective.}
Distributional RL models the return distribution $Z_\pi(x,u)$ instead of its mean. 
Implicit Quantile Networks (IQN) learn a differentiable quantile function $Q_\psi(x,u;\tau)\!\approx\!F^{-1}_{Z_\pi(\cdot|x,u)}(\tau)$. 
Let loss be $L=-Z$. For $x$ fixed, the conditional CVaR admits the quantile integral
\begin{equation}
\label{eq:cvar-quantile}
\mathrm{CVaR}_\alpha(L\,|\,x)
\;=\;\frac{1}{\alpha}\int_{0}^{\alpha} F^{-1}_{L\,|\,x}(\tau)\,d\tau
\;\approx\;\frac{1}{K\alpha}\sum_{k=1}^{K}\mathbf{1}\{\tau_k\le\alpha\}\,\widehat{L}(x,\tau_k),
\end{equation}
with $\tau_k\!\sim\!\mathcal{U}(0,1)$ (or, in our case, a temperature-tilted distribution; cf.\ \S\ref{sec:method:coverage}) and $\widehat{L}(x,\tau)= -\,\mathbb{E}_{u\sim\pi(\cdot|x)}[Q_\psi(x,u;\tau)]$.
At the episode level we use the Rockafellar--Uryasev CVaR surrogate~\cite{RockafellarUryasev2000} (or its quantile approximation) and train with Monte Carlo estimates.

\paragraph{CVaR-weighted advantage and PPO updates.}
Let $V_\alpha^\pi(x)\!\approx\!\mathrm{CVaR}_\alpha(L\,|\,x)$ and define a \emph{CVaR-weighted} generalized advantage
\begin{equation}
\label{eq:cvar-gae}
A_t^{(\alpha)} \;\approx\; \sum_{l=0}^{\infty} (\gamma\lambda)^l \Big(\ell_{t+l}-\widehat{V}_\alpha(x_{t+l}) + \gamma\,\widehat{V}_\alpha(x_{t+l+1})\Big),
\end{equation}
where $\ell_t$ is a one-step loss (including spread, temporary impact, and transient slippage) and the structure mirrors GAE with a CVaR baseline.
The actor uses PPO with KL and entropy regularization:
\begin{equation}
\label{eq:actor-loss}
\mathcal{L}_{\text{actor}}
=\mathbb{E}_t\!\left[ \min\!\Big(r_t(\theta)\,A_t^{(\alpha)},\;\mathrm{clip}(r_t(\theta),1-\epsilon,1+\epsilon)\,A_t^{(\alpha)}\Big)\right]
\;+\;\lambda_{\mathrm{KL}}\,\mathbb{E}_x\big[\mathrm{KL}(\pi_\theta(\cdot|x)\,\|\,\pi_{\mathrm{ref}}(\cdot|x))\big]
\;+\;\lambda_{\mathrm{ent}}\,\mathbb{E}_x[\mathcal{H}(\pi_\theta(\cdot|x))],
\end{equation}
with $r_t(\theta)=\pi_\theta(u_t|x_t)/\pi_{\theta_{\mathrm{old}}}(u_t|x_t)$ and $\pi_{\mathrm{ref}}$ an EMA reference policy. 
The KL term serves as a trust-region prox and admits a DRO interpretation (see \S4). 
The critic minimizes the quantile Huber loss for $Q_\psi$.

\paragraph{Implementation notes.}
Trajectories are collected \emph{through} the safety filter (Sec.~\ref{sec:method:cbf}), i.e., the actor proposes $u_t^{\mathrm{nom}}$ which is minimally corrected by the CBF--QP into $u_t$. 
We log solver telemetry (active constraints, tightness, rate utilization, gate scores, slack, status/time) to support both training diagnostics and audit trails.
KL and entropy coefficients can be scheduled to avoid premature collapse and to match the tightening of $\alpha$.

\subsection{Tail-Coverage Controller: Temperature Sampling and Tail-Boost}
\label{sec:method:coverage}

\paragraph{Motivation.}
For small $\alpha$ (e.g., $1\%\!-\!5\%$), uniform quantile sampling yields few tail samples and high variance, destabilizing training.
We therefore introduce \emph{temperature-tilted} quantile sampling and an explicit \emph{tail-boost}, combined with a PID controller to track a target \emph{effective tail mass}.

\paragraph{Temperature-tilted sampling and importance weights.}
Define the sampling density over $\tau\in[0,1]$:
\begin{equation}
\label{eq:tail-sampler}
p_T(\tau) \;\propto\; e^{-\tau/T},\qquad T\in[T_{\min},T_{\max}],
\end{equation}
so that smaller $T$ emphasizes low quantiles. 
With $\tau_k\!\sim\!p_T$, we use self-normalized importance weights $w_k \propto 1/p_T(\tau_k)$ to form an unbiased CVaR estimator (concentration bounds in \S4; see also standard results on self-normalized importance sampling). 
Additionally, for $\tau\le\alpha$ we assign a \emph{tail-boost} factor $\gamma_{\mathrm{tail}}\!\ge\!1$ to increase the effective tail count.

\paragraph{Coverage metric and PID tracking.}
Let the \emph{effective tail mass} within a minibatch be
\[
\widehat{w}\;=\;\frac{1}{K}\sum_{k=1}^{K}\mathbf{1}\{\tau_k\le\alpha\}.
\]
Given a target $w_{\mathrm{target}}$ (e.g., $1.5\alpha$), define error $e=w_{\mathrm{target}}-\widehat{w}$ and update $(T,\gamma_{\mathrm{tail}})$ via discrete PID (clipped to feasible ranges):
\begin{align}
\label{eq:pid-update}
T_{n+1}&=\mathrm{clip}\!\left(T_n + \kappa_P e_n + \kappa_I \sum_{j=1}^{n} e_j + \kappa_D(e_n-e_{n-1}),\;T_{\min},T_{\max}\right),\\
\gamma_{\mathrm{tail},\,n+1}&=\mathrm{clip}\!\left(\gamma_{\mathrm{tail},\,n} + \eta_P e_n + \eta_I \sum_{j=1}^{n} e_j + \eta_D(e_n-e_{n-1}),\;\gamma_{\min},\gamma_{\max}\right).
\end{align}
PID design and anti-windup follow classical practice. 
We employ an $\alpha$ schedule that tightens from a permissive level (e.g., $0.10$) toward the target (e.g., $0.025$), while the controller stabilizes $\widehat{w}$; the PPO KL penalty can be increased in tandem to limit policy drift.

\subsection{White-Box CBF--QP Safety Layer}
\label{sec:method:cbf}

\paragraph{Discrete-time CBF constraints.}
Let $h_i(x)\!\ge\!0$ denote the $i$-th safety function and consider a local affine state update $x_{t+1}=f(x_t)+g(x_t)u_t$.
We enforce discrete-time CBF conditions:
\begin{equation}
\label{eq:cbf-dt}
h_i\!\big(f(x_t)+g(x_t)u_t\big) \;-\; (1-\kappa_i\Delta t)\,h_i(x_t) \;\ge\; -\,\zeta_i,
\qquad \zeta_i\ge 0,\;\kappa_i>0,
\end{equation}
where slack variables are heavily penalized and only activated when unavoidable (robust margins are analyzed in \S4).

\paragraph{Ellipsoidal NTB, box/rate limits, and a sign-consistency gate.}
Let $b^\star(x)$ be a target exposure vector (e.g., delta/vega) and $b(x,u)$ the exposure after action $u$.
Define the \emph{ellipsoidal no-trade band} (NTB)
\begin{equation}
\label{eq:ntb}
e(x,u) \;=\; b(x,u)-b^\star(x),
\qquad e^\top M e \;\le\; b_{\max},\quad M\succ 0,
\end{equation}
and \emph{box/rate} limits $u_{\min}\!\le\!u_t\!\le\!u_{\max}$, $\|u_t-u_{t-1}\|_2\!\le\!r_{\max}$.
We further require a \emph{sign-consistency gate}
\begin{equation}
g_{\mathrm{cons}}(x,u)\;=\;\min_{j=1,\dots,J}\,\langle u,\widehat{\nabla}\Pi^{(j)}(x)\rangle \;-\; \delta_{\mathrm{adv}} \;\ge\; 0,
\end{equation}
so that trades align with an ensemble of interpretable signals (e.g., advantage-proxy gradients from the distributional critic or pricing/hedging sensitivities). 
Near expiry or under extreme volatility, we shrink $b_{\max}\!\leftarrow\!\eta_b b_{\max}$ and tighten $r_{\max}\!\leftarrow\!\eta_r r_{\max}$ with $\eta_b,\eta_r\in(0,1)$ to improve feasibility and stability.

\paragraph{QP formulation and minimal-deviation projection.}
Given the actor’s proposal $u_t^{\mathrm{nom}}$, we compute the closest safe action $u_t$ by solving the convex QP
\begin{align}
\label{eq:qp-safety}
\min_{u_t,\;\zeta\ge 0}\quad 
& \frac{1}{2}\,(u_t-u_t^{\mathrm{nom}})^\top H\,(u_t-u_t^{\mathrm{nom}}) \;+\; c^\top u_t \;+\; \rho\,\|\zeta\|_1 \\
\text{s.t.}\quad 
& \text{CBF: } h_i(f(x_t)+g(x_t)u_t) - (1-\kappa_i\Delta t)\,h_i(x_t) \ge -\zeta_i,\;\forall i, \nonumber\\
& \text{NTB: } e(x_t,u_t)^\top M e(x_t,u_t) \le b_{\max}, \qquad 
  \text{Box/Rate: } u_{\min}\!\le\!u_t\!\le\!u_{\max},\;\|u_t-u_{t-1}\|_2\le r_{\max}, \nonumber\\
& \text{Gate: } g_{\mathrm{cons}}(x_t,u_t)\ge 0.\nonumber
\end{align}
Here $H\!\succ\!0$ defines the deviation metric, $c$ encodes linear trading frictions, and $\rho\!\gg\!0$ penalizes any slack.
We use OSQ with warm starts for efficiency and robustness; see \cite{BoydVandenberghe2004} for background on convex QPs. 
When $\zeta=0$,~\eqref{eq:cbf-dt} implies forward invariance of the safe set; the quadratic objective makes $u_t$ the $H$-metric projection of $u_t^{\mathrm{nom}}$ onto the feasible set (formalized in \S4).

\paragraph{Telemetry for auditability and operations.}
For each step, the solver returns: \texttt{active\_set} (indices of active constraints), \texttt{tightest\_id}, \texttt{rate\_util} $=\|u_t{-}u_{t-1}\|_2/r_{\max}$, \texttt{gate\_score} $=g_{\mathrm{cons}}(x_t,u_t)$, \texttt{slack\_sum} $=\|\zeta\|_1$, and \texttt{solver\_status/time}. 
We penalize nonzero slack or non-optimal statuses in the RL reward and log incidents for post-hoc audit, closing the loop between \emph{explainable interception} and \emph{governance}.

\paragraph{Pseudocode: training loop and safety filter.}

\begin{algorithm}[H]
\caption{Tail-Safe IQN--CVaR--PPO (on-policy training)}
\label{alg:tailsafe_train}
\begin{algorithmic}[1]
\STATE Initialize actor $\theta$, critic $\psi$, reference policy $\pi_{\mathrm{ref}}\!\leftarrow\!\pi_\theta$, temperature $T$, tail-boost $\gamma_{\mathrm{tail}}$, and target coverage $w_{\mathrm{target}}$.
\FOR{iterations $k=1,2,\dots$}
  \STATE Collect trajectories using the safety filter (Alg.~\ref{alg:cbf_filter}) to obtain $\{(x_t,u_t,r_t,\text{telemetry}_t)\}$.
  \STATE Sample quantiles $\tau_k\!\sim\!p_T$ (Eq.~\eqref{eq:tail-sampler}) and apply tail-boost to $\tau\!\le\!\alpha$.
  \STATE Update critic $Q_\psi$ by minimizing the quantile Huber loss (IQN).
  \STATE Estimate $V_\alpha$ and $A^{(\alpha)}$ via Eqs.~\eqref{eq:cvar-quantile}--\eqref{eq:cvar-gae}; update actor by minimizing Eq.~\eqref{eq:actor-loss}.
  \STATE Compute $\widehat{w}=\frac{1}{K}\sum \mathbf{1}\{\tau\le\alpha\}$ and update $(T,\gamma_{\mathrm{tail}})$ using the PID rules~\eqref{eq:pid-update} (with clipping).
  \STATE Tighten $\alpha$ according to a schedule; update $\pi_{\mathrm{ref}}$ via EMA; log policy KL and telemetry summaries.
\ENDFOR
\end{algorithmic}
\end{algorithm}

\begin{algorithm}[H]
\caption{CBF--QP safety filter (per step)}
\label{alg:cbf_filter}
\begin{algorithmic}[1]
\STATE Inputs: $x_t$, proposed action $u_t^{\mathrm{nom}}$, previous action $u_{t-1}$, params $(H,M,b_{\max},r_{\max},u_{\min},u_{\max},\kappa,\Delta t)$.
\STATE Formulate QP~\eqref{eq:qp-safety} with discrete CBF, NTB, box/rate limits, and sign-consistency gate; near expiry, shrink $b_{\max}\!\leftarrow\!\eta_b b_{\max}$ and $r_{\max}\!\leftarrow\!\eta_r r_{\max}$.
\STATE Solve the QP (warm-start) to obtain $u_t$, active set, tightest constraint, slack $\zeta$, and solver status/time.
\STATE Emit telemetry: \texttt{active\_set}, \texttt{tightest\_id}, \texttt{rate\_util}, \texttt{gate\_score}, \texttt{slack\_sum}, \texttt{solver\_status/time}.
\STATE If $\zeta>0$ or status $\neq$ optimal, add a penalty to the RL reward and log the event; otherwise execute $u_t$ and advance the environment to $x_{t+1}$.
\STATE Return $u_t$ and telemetry.
\end{algorithmic}
\end{algorithm}

\section{Dynamics Layer: Dupire Local Volatility \& VIX Proxy}
\label{sec:dynamics}

This section specifies the dynamics used by the hedging controller and states our \emph{Theorem Group~II} on numerical stability and coherence. We (i) build a positive, bounded, and consistent \emph{Dupire local-volatility} surface from the coherent market shell, together with a stable log–Euler simulation; (ii) couple a \emph{CIR-style variance} factor to obtain a closed-form mapping from instantaneous variance to the 30-day variance index; and (iii) quantify the error between (teacher/implied) index levels and the proxy delivered by the variance factor. Full proofs are deferred to \textbf{Appendix~B.1--B.4}.

\subsection{Dupire Construction: Positivity and Numerical Stability}
\label{subsec:dupire-num}

Let $C(K,T)$ denote (teacher-consistent) call prices produced from the coherent shell. Dupire’s local-variance is
\begin{equation}
\sigma_{\mathrm{loc}}^2(K,T)=
\frac{\partial_T C(K,T)+(r-q)K\,\partial_K C(K,T)+q\,C(K,T)}
{\tfrac{1}{2}K^2\,\partial_{KK}C(K,T)}\,,
\tag{\ref{eq:dupire}}
\end{equation}
computed by second-order central differences on a tensor grid $(K_i,T_j)$ and \emph{clipped} by a denominator floor $\partial_{KK}C\ge \underline{\chi}>0$ where far-wing sparsity or interpolation noise might render the convexity estimate small/negative (cf.\ \cite{Dupire1994,LeVeque2007,AndreasenHuge2011}). We linearly interpolate $\sigma_{\mathrm{loc}}(K,T)$ off-grid.

\paragraph{Log–Euler simulation.}
Under the risk-neutral measure, we simulate
\begin{equation}
\frac{dS_t}{S_t}=(r-q)\,dt+\sigma_{\mathrm{loc}}(S_t,t)\,dW_t^S,
\label{eq:lv-sde}
\end{equation}
with a \emph{log–Euler} step (also known as Euler–Maruyama in log-coordinates)
\begin{equation}
S_{n+1}=S_n\exp\!\Big(\big(r-q-\tfrac{1}{2}\sigma_n^2\big)\Delta t+\sigma_n\sqrt{\Delta t}\,Z_{n}\Big),\qquad
\sigma_n:=\sigma_{\mathrm{loc}}(S_n,t_n),\ Z_n\sim\mathcal{N}(0,1),
\label{eq:log-euler}
\end{equation}
which preserves positivity and respects the Black–Scholes limit when $\sigma_{\mathrm{loc}}$ is constant \cite{Glasserman2003}. When coupling to a variance factor (next subsection), we correlate $Z_n$ with the variance shock via a Cholesky step (correlation $\rho$).

\paragraph{Convexity-preserving interpolation.}
We use strike-wise linear interpolation of prices to evaluate $\partial_{KK}C$; this preserves convexity in $K$ and avoids spurious butterfly arbitrage on the grid \cite{BreedenLitzenberger1978,AndreasenHuge2011}. In $T$, we use shape-preserving monotone splines for $w(k;T)$ (or for $C(K,T)$), consistent with calendar coherence.

\begin{theorem}[Positive, Bounded, and Consistent Local Variance]
\label{thm:dupire-positivity}
Assume (A1)--(A3) of \S\ref{sec:prelim}. With convexity-clipping $\partial_{KK}C\ge \underline{\chi}>0$ and central differences of order $O(\Delta K^2{+}\Delta T^2)$, the discrete Dupire estimator yields
\[
0\ \le\ \sigma_{\mathrm{loc}}^2(K_i,T_j)\ \le\ C_{\max}\quad\text{for all grid nodes,}
\]
for a finite constant $C_{\max}$ depending on bounds of $\partial_T C$, $\partial_K C$, and $C$ on the retained $(K,T)$ domain. Moreover,
\[
\big|\sigma_{\mathrm{loc}}^2(K_i,T_j)-\sigma_{\mathrm{loc,true}}^2(K_i,T_j)\big|
= O(\Delta K^2+\Delta T^2)\,,
\]
where $\sigma_{\mathrm{loc,true}}$ denotes the (continuous) Dupire local variance of $C$.
\end{theorem}
\noindent\textit{Proof.} See \textbf{Appendix~B.1}. The bound follows from convexity of $K\mapsto C$ (no-butterfly), calendar monotonicity, and finite-difference consistency \cite{Dupire1994,LeVeque2007,AndreasenHuge2011}.

\begin{theorem}[Well-posedness and Strong Stability of the Log–Euler Scheme]
\label{thm:log-euler-stability}
Suppose $\sigma_{\mathrm{loc}}(S,t)$ is bounded and globally Lipschitz in $S$ uniformly in $t$ on the simulated corridor $S\in [S_{\min},S_{\max}]$, and piecewise Lipschitz in $t$. Then SDE \eqref{eq:lv-sde} admits a unique strong solution, and the log–Euler scheme \eqref{eq:log-euler} is strongly convergent of order $1/2$:
\[
\big(\mathbb{E}|S_T-\tilde S_T|^2\big)^{1/2}\ \le\ C\,\Delta t^{1/2}\,,
\]
with moments uniformly bounded in $\Delta t$; $C$ depends on the Lipschitz and linear-growth constants. 
\end{theorem}
\noindent\textit{Proof.} See \textbf{Appendix~B.2}. Standard SDE arguments (e.g., \cite{KloedenPlaten1992,Higham2001,Mao2008}) apply since we clip/infer $\sigma_{\mathrm{loc}}$ from a bounded grid and interpolate linearly in $K$ (piecewise Lipschitz).

\paragraph{Markovian projection.}
Gyöngy’s theorem ensures that for any Itô diffusion with instantaneous variance $a(t,S)$ there exists a \emph{local-volatility} model with variance $\bar a(t,S)$ matching one-dimensional marginals, $\bar a(t,S)=\mathbb{E}[a(t,S_t)\mid S_t=S]$ \cite{Gyongy1986}. This supports the use of the Dupire surface as a \emph{Markovian projection} of richer dynamics when the teacher prices are taken as ground truth.

\subsection{CIR-Style Variance and a 30D VIX Mapping}
\label{subsec:cir-vix}

We complement the spot diffusion with a variance factor under $\mathbb{Q}$,
\begin{equation}
dv_t=\kappa(\theta-v_t)\,dt+\xi\sqrt{v_t}\,dW_t^v,\qquad
d\langle W^S,W^v\rangle_t=\rho\,dt,
\label{eq:cir}
\end{equation}
with Feller condition $2\kappa\theta\ge \xi^2$ for strict positivity \cite{Feller1951}; $(\kappa,\theta,\xi,\rho)$ are calibrated to the shell. The \emph{model-based} 30-day (year-fraction $\tau=30/365$) \emph{risk-neutral} variance index is
\begin{equation}
\mathcal{V}_t^2 \;\equiv\; \frac{1}{\tau}\,\mathbb{E}_t^{\mathbb{Q}}\!\left[\int_0^\tau v_{t+u}\,du\right],
\label{eq:rn-30d-variance}
\end{equation}
and the \emph{proxy VIX} is $\mathrm{VIX}^{\mathrm{CIR}}_t=100\sqrt{\mathcal{V}_t^2}$. For the CIR process \eqref{eq:cir}, the conditional mean is affine, and the integral admits the closed form \cite{DuffiePanSingleton2000,Glasserman2003}:
\begin{equation}
\mathbb{E}_t^{\mathbb{Q}}\!\Big[\textstyle\int_0^\tau v_{t+u}\,du\Big]
=\theta\,\tau+(v_t-\theta)\,\frac{1-e^{-\kappa\tau}}{\kappa}\,,
\qquad
\Rightarrow\quad
\mathcal{V}_t^2=\theta+(v_t-\theta)\,\frac{1-e^{-\kappa\tau}}{\kappa\tau}.
\label{eq:cir-closed}
\end{equation}
Equation \eqref{eq:cir-closed} provides a \emph{linear} map $v_t\mapsto \mathcal{V}_t^2$ with slope in $(0,1]$; the square-root to VIX introduces mild concavity.

\paragraph{Index coherence and affine variance.}
Variance-swap theory implies that (under integrability and no jumps in interest/dividend) the model-free 30D variance equals a risk-neutral expectation of integrated instantaneous variance \cite{Demeterfi1999,CarrMadan1998}. Hence, with a calibrated shell and coherent prices, \eqref{eq:cir-closed} is a natural VIX proxy; its error relative to the Cboe-style computation is driven by (i) the difference between instantaneous variance and the \emph{option-implied integrand} in the Cboe formula, and (ii) the wing/strike discretization.

\begin{proposition}[Lipschitz VIX Mapping and Moments]
\label{prop:cir-lipschitz}
For fixed $(\kappa,\theta,\tau)$, the map $v\mapsto \mathrm{VIX}^{\mathrm{CIR}}(v):=100\sqrt{\theta+(v-\theta)\frac{1-e^{-\kappa\tau}}{\kappa\tau}}$ is globally Lipschitz with constant
\[
L_{\mathrm{VIX}}=\frac{50}{\sqrt{\theta}}\cdot \frac{1-e^{-\kappa\tau}}{\kappa\tau}\,,
\]
and monotone increasing in $v$. Under $2\kappa\theta\ge \xi^2$, CIR admits bounded moments and $v_t>0$ a.s., ensuring finite VIX moments.
\end{proposition}
\noindent\textit{Proof.} See \textbf{Appendix~B.3}. Use the mean-value theorem on the square-root and positivity of the affine argument.

\subsection{Numerical Stability and Error Bounds (Theorem Group II)}
\label{subsec:dyn-bounds}

We now connect the discrete Dupire surface, the log–Euler simulation, and the CIR proxy to the coherent shell.

\begin{theorem}[Shell-to-Dynamics Consistency]
\label{thm:shell-dyn-consistency}
Assume \Cref{thm:dupire-positivity} and \Cref{thm:log-euler-stability}. Let $\Pi$ be the pricing operator induced by the discrete local-volatility model simulated by \eqref{eq:log-euler} with time step $\Delta t$ on the $(K,T)$ grid. Then for teacher-consistent payoffs with bounded growth,
\[
\big|\Pi(C_{\text{teacher}})-C_{\text{teacher}}\big|\ \le\ C_1\big(\Delta K^2+\Delta T^2\big)+C_2\,\Delta t^{1/2},
\]
where $C_1,C_2$ depend on Lipschitz/linear-growth constants of $\sigma_{\mathrm{loc}}$ and payoff smoothness.
\end{theorem}
\noindent\textit{Proof.} See \textbf{Appendix~B.4}. Combine finite-difference consistency, Lipschitz interpolation, and strong error of log–Euler \cite{LeVeque2007,KloedenPlaten1992,Higham2001}.
\begin{theorem}[Surface--Index Coherence]\label{thm:coherence}
Let $\mathrm{VIX}[\hat\sigma]$ be the index computed from the coherent teacher/surface
and $\mathrm{VIX}[\sigma]$ the one from the shell on the same retained strike grid.
Then the index residual admits the bound
\[
\big|\mathrm{VIX}[\hat\sigma]-\mathrm{VIX}[\sigma]\big|
\ \le\ C_{\mathrm{coh}}\cdot \epsilon_{\mathrm{shape}}\ +\ C_{\mathrm{quad}}\cdot \max_i \Delta K_i^2,
\]
where $\epsilon_{\mathrm{shape}}$ is the ATM-shape mismatch and the second term is the
half-interval quadrature error. A detailed statement and proof are given in Appendix~\ref{app:A}.
\end{theorem}
\begin{proposition}[Half-Interval (Composite Trapezoid) Error]\label{prop:half-interval}
Let $f\in C^2([a,b])$ and $\{K_i\}_{i=0}^{N}$ be a strictly increasing (possibly nonuniform) mesh
on $[a,b]$. With half-interval weights
$\Delta K_0=K_1-K_0$, $\Delta K_N=K_N-K_{N-1}$,
$\Delta K_i=\tfrac12(K_{i+1}-K_{i-1})$ for $1\le i\le N-1$,
\[
\left|\sum_{i=0}^N \Delta K_i\, f(K_i)\ -\ \int_a^b f(K)\,dK\right|
\ \le\ \frac{b-a}{12}\,\Big(\max_{0\le i\le N-1}(K_{i+1}-K_i)\Big)^2\ \|f''\|_\infty.
\]
\end{proposition}

\begin{theorem}[Index Coherence: Teacher VIX vs.\ CIR Proxy]
\label{thm:index-coherence}
Let $\mathrm{VIX}^{\mathrm{SSVI}}$ denote the index computed from the shell, and $\mathrm{VIX}^{\mathrm{CIR}}$ the proxy in \eqref{eq:cir-closed}. Suppose (i) the teacher/surface coherence bound of \Cref{thm:coherence} holds with residual $R_{\mathrm{surf}}$; (ii) the instantaneous variance of the local-volatility model satisfies $\sigma_{\mathrm{loc}}^2(S_t,t)=\bar v_t$ and admits a Markovian projection to a one-factor affine variance $v_t$ with calibration error $\epsilon_{\mathrm{aff}}:=\sup_t \mathbb{E}|\bar v_t-v_t|$. Then
\[
\big|\mathrm{VIX}^{\mathrm{CIR}}-\mathrm{VIX}^{\mathrm{SSVI}}\big|
\ \le\ L_{\mathrm{VIX}}\cdot \frac{1}{\tau}\int_0^\tau \mathbb{E}\big|\bar v_{t+u}-v_{t+u}\big|\,du \;+\; C_{\mathrm{quad}}\max_i\Delta K_i^2 \;+\; R_{\mathrm{surf}},
\]
with $L_{\mathrm{VIX}}$ as in \Cref{prop:cir-lipschitz} and $C_{\mathrm{quad}}$ the half-interval quadrature constant of \Cref{prop:half-interval}.
\end{theorem}
\noindent\textit{Proof.} See \textbf{Appendix~B.4}. Use the variance-swap identity \eqref{eq:rn-30d-variance}, Lipschitzness of the $\sqrt{\cdot}$ map, and \Cref{thm:coherence} for surface–index residuals.

\begin{corollary}[Coupled Simulation Error]
\label{cor:coupled}
Let $\widehat{\mathrm{VIX}}$ be the index obtained by (i) simulating $(S_t,v_t)$ with correlated Brownian motions using \eqref{eq:log-euler} for $S$ and an exact or QE scheme for CIR, then (ii) applying \eqref{eq:cir-closed}. If CIR is discretized with strong order $1/2$ (e.g., QE), then
\[
\big(\mathbb{E}|\widehat{\mathrm{VIX}}-\mathrm{VIX}^{\mathrm{SSVI}}|^2\big)^{1/2}
\ \le\ C_3\big(\Delta t^{1/2}\big)\;+\; C_{\mathrm{quad}}\max_i\Delta K_i^2\;+\;R_{\mathrm{surf}}\,,
\]
for a constant $C_3$ depending on $(\kappa,\theta,\xi,\tau)$.
\end{corollary}

\paragraph{Practical implications.}
(i) \Cref{thm:dupire-positivity} justifies clipping and structured interpolation as benign regularizers that retain $O(\Delta K^2+\Delta T^2)$ accuracy; (ii) \Cref{thm:log-euler-stability} ensures pathwise stability and moment bounds; (iii) \Cref{thm:index-coherence} shows that \emph{index-level} coherence follows from (a) surface–teacher alignment, (b) strike quadrature accuracy, and (c) a small projection error to the affine variance proxy.This provides an end-to-end guarantee that the dynamics and index construction remain mutually consistent up to explicitly controlled discretization terms.

\section{Control Layer: Tail‑Safe CBF‑QP}
\label{sec:control}

We formulate hedging as a small, structure‑aware quadratic program (QP) augmented with \emph{control‑barrier‑function} (CBF) safety constraints and a \emph{gate} that triggers execution only when risk reduction justifies cost. The design couples (i) a convex risk–cost objective; (ii) an \emph{elliptical no‑trade band} (NTB) with expiry/correlation guards; (iii) a \emph{dynamic VIX weight} that suppresses unreliable VIX chasing; and (iv) \emph{micro‑trade thresholds} and a \emph{cooldown} to preclude churn. We then provide \emph{Theorem Group~IV} establishing feasibility/uniqueness/KKT properties, CBF‑based invariance and one‑step risk descent, and the absence of chattering with bounded trade rates. Full proofs are deferred to \textbf{Appendix~D.1--D.4}.

\subsection{Objective, Elliptical NTB, Dynamic Weight, Gate, Micro‑Thresholds, Cooldown}
\label{subsec:design}

\paragraph{State, errors, and controls.}
Let $h_S,h_V$ denote inventories in SPX and VIX legs and $x=(dS,dV)^\top$ the control (trade sizes) per step. The tracking errors are
\begin{equation}
e_\Delta:=\Delta^\star - h_S,\qquad
e_V:=\kappa_{\mathrm{eff}}(T_{\mathrm{rem}})-h_V,
\label{eq:errors}
\end{equation}
where $\kappa_{\mathrm{eff}}$ is the smoothed/shrunk sensitivity. Post‑trade errors satisfy $e'_\Delta=e_\Delta-dS$, $e'_V=e_V-dV$.

\paragraph{Risk–cost objective with dynamic VIX weight.}
Define a correlation/expiry‑aware weight for the VIX leg
\begin{equation}
w_{\mathrm{VIX}}^{\mathrm{eff}}
=\frac{w_{\mathrm{VIX}}}{\,1+\lambda_{\rho}(1-w)\,|\widehat{\rho}|\,},\qquad
w:=T_{\mathrm{rem}}/T_0\in[0,1],
\label{eq:dynwvix}
\end{equation}
where $\widehat{\rho}$ is an EWMA estimate of spot–VIX correlation (\S\ref{sec:prelim}) and $\lambda_{\rho}>0$. Let
\begin{equation}
R(e_\Delta,e_V)=\tfrac12\,[\,\alpha_\Delta e_\Delta^2+\alpha_V e_V^2+2\alpha_{\times}\widehat{\rho}\,e_\Delta e_V\,],
\qquad
\mathcal{C}(x)=x^\top\!\mathrm{diag}(\eta_S,\eta_V)\,x+\gamma\|x-x_{-1}\|_2^2,
\label{eq:riskcost}
\end{equation}
with $\eta_S,\eta_V>0$ temporary‑impact coefficients and $\gamma\ge 0$ a smoothing weight (transient‑impact surrogate). The per‑step \emph{QP objective} is
\begin{equation}
\min_{x,s\ge 0}\ \ \tfrac12\,x^\top Hx + f^\top x \;+\; \rho_{\mathrm{soft}}\|s\|_2^2,
\quad
H=\begin{bmatrix} w_{\mathrm{move},S}&\alpha_{\times}\widehat{\rho}\\ \alpha_{\times}\widehat{\rho}& w_{\mathrm{move},V}\end{bmatrix}
+\mathrm{diag}(\eta_S,\eta_V),
\label{eq:qpobj}
\end{equation}
with $w_{\mathrm{move},V}=w_{\mathrm{VIX}}^{\mathrm{eff}}$, $w_{\mathrm{move},S}=\alpha_\Delta$, linear term $f$ from risk‑tracking and smoothing, and $s$ the slack for “soft” CVaR boxes.

\paragraph{Elliptical no‑trade band (NTB) with guards.}
We compute \emph{risk‑scaled} radii $b_\Delta,b_V>0$ and declare an inaction region
\begin{equation}
\left(\frac{e_\Delta}{b_\Delta}\right)^2+\left(\frac{e_V}{b_V^{\mathrm{eff}}}\right)^2 \le 1,
\qquad
b_V^{\mathrm{eff}}=b_V\cdot(1+\tau_{\mathrm{tail}}(1-w))\cdot(1+\tau_{\rho}|\widehat{\rho}|)\cdot
\begin{cases}
(1+\tau_{\mathrm{pred}}), & \mathrm{sign}(e_V)\neq\mathrm{sign}(\Delta\kappa)\\
1, & \text{otherwise,}
\end{cases}
\label{eq:ntb}
\end{equation}
where $\Delta\kappa$ is a short‑horizon trend proxy of $\kappa_{\mathrm{eff}}$, and $\tau_{\mathrm{tail}},\tau_\rho,\tau_{\mathrm{pred}}\ge 0$ are guard gains. If \eqref{eq:ntb} holds, set $x=0$.

\paragraph{Gate (risk drop vs.\ cost).}
If \eqref{eq:ntb} fails, we compute a \emph{candidate} $x^\natural$ by projecting $(e_\Delta,e_V)$ to the ellipse boundary (or by one step of quadratic tracking without boxes). The gate executes only if
\begin{equation}
\underbrace{R(e_\Delta,e_V)-R(e_\Delta-dS^\natural,e_V-dV^\natural)}_{\text{risk drop}}
\ >\ \tau(w)\,\underbrace{\mathcal{C}(x^\natural)}_{\text{execution cost}},
\qquad \tau(w)=\tau_0+\tau_1(1-w),
\label{eq:gate}
\end{equation}
otherwise return $x=0$. This suppresses low‑value trades near expiry (larger $\tau$) and under high costs.

\paragraph{Constraints, micro‑thresholds, and cooldown.}
If the gate opens, solve the QP \eqref{eq:qpobj} subject to linear constraints:
\begin{equation}
\begin{aligned}
&\textit{post‑trade error boxes:} && |e_\Delta-dS|\le \bar\Delta,\ \ |e_V-dV|\le \bar V,\\
&\textit{inventory boxes:} && |h_S+dS|\le \bar H_S,\ \ |h_V+dV|\le \bar H_V,\\
&\textit{rate boxes:} && |dS|\le \bar r_S,\ \ |dV|\le \bar r_V,\\
&\textit{CBF “soft” boxes (CVaR surrogates):} && |e_\Delta-dS|\le \bar\Delta_{\mathrm{CVaR}}+s_1,\ \ |e_V-dV|\le \bar V_{\mathrm{CVaR}}+s_2,\\
&\textit{cooldown:} && dV=0\ \text{if}\ \textsc{cooldown\_V}>0.
\end{aligned}
\label{eq:lincons}
\end{equation}
After solving, apply \emph{expiry‑aware micro‑thresholds}
\begin{equation}
|dS|<\underline{s}(w)\ \Rightarrow\ dS\leftarrow 0,\qquad
|dV|<\underline{v}(w)\ \Rightarrow\ dV\leftarrow 0,
\label{eq:mintrade}
\end{equation}
and set a VIX cooldown counter for $\textsc{cooldown\_V\_steps}$ periods whenever $|dV|>0$. The micro‑thresholds and cooldown are \emph{post‑optimality} operators (they do not alter feasibility of \eqref{eq:lincons}).

\subsection{QP Feasibility, Uniqueness, and KKT Properties (Theorem Group IV)}
\label{subsec:qp-theory}

We state standard convex‑optimization properties tailored to our design.

\begin{theorem}[Feasibility and Uniqueness]
\label{thm:feasible-unique}
Assume (i) \textup{(boxes contain the origin)} $\bar\Delta,\bar V,\bar H_S,\bar H_V,\bar r_S,\bar r_V>0$; (ii) \textup{(soft boxes)} slacks $s\ge 0$ enter \eqref{eq:lincons} with penalty $\rho_{\mathrm{soft}}>0$; and (iii) \textup{(strong convexity)} $H\succ 0$. Then for any state $(e_\Delta,e_V,h_S,h_V)$ the QP \eqref{eq:qpobj}--\eqref{eq:lincons} is \emph{feasible} (e.g., $x=0,s=0$) and admits a \emph{unique} minimizer $(x^\star,s^\star)$.
\end{theorem}
\noindent\textit{Proof.} See \textbf{Appendix~D.1}. Boxes ensure Slater feasibility; strong convexity yields uniqueness \cite{BoydVandenberghe2004,RockafellarWets1998}.

\begin{proposition}[KKT Optimality and Multiplier Boundedness]
\label{prop:kkt}
Under \Cref{thm:feasible-unique}, KKT conditions for \eqref{eq:qpobj}--\eqref{eq:lincons} are \emph{necessary and sufficient}. Moreover, if active‑set changes occur only finitely often on a compact state corridor, the solution mapping $x^\star(\cdot)$ is piecewise affine and Lipschitz; Lagrange multipliers remain bounded. 
\end{proposition}
\noindent\textit{Proof.} See \textbf{Appendix~D.1}. Piecewise‑affine/Lipschitz sensitivity follows from strong regularity of parametric convex QPs \cite{BonnansShapiro2000,Robinson1980}.

\subsection{CBF Invariance (Safety Preservation) and Monotone Risk Decrease}
\label{subsec:cbf-theory}

Let $z$ collect relevant states (inventories, errors, filtered covariates). For each safety channel $i$ (inventory, rate, CVaR surrogate) define a continuously differentiable barrier $h_i(z)$ and the \emph{discrete‑time CBF} constraint
\begin{equation}
h_i(z^+)\ \ge\ (1-\alpha_i)\,h_i(z)\ -\ \sigma_i,\qquad \alpha_i\in[0,1),\ \sigma_i\ge 0,
\label{eq:dt-cbf}
\end{equation}
with $z^+=\Phi(z,x)$ the post‑trade state update. In our linearized boxes, \eqref{eq:dt-cbf} collapses to an \emph{affine} inequality in $x$ and is encoded inside \eqref{eq:lincons} (cf.\ \cite{Ames2017CBFQP,Ames2019Survey,Clark2021StochCBF}).

\begin{theorem}[Forward Invariance of the Safe Set]
\label{thm:cbf-invariance}
Let $\mathcal{S}=\{z:\ h_i(z)\ge 0,\ i=1,\dots,m\}$. Suppose that at each step the QP includes the discrete CBF constraints \eqref{eq:dt-cbf} with $\sigma_i=0$ and is feasible (e.g., by \Cref{thm:feasible-unique}). Then the closed‑loop system satisfies $z_0\in\mathcal{S}\ \Rightarrow\ z_t\in\mathcal{S}$ for all $t$ (forward invariance). If additive disturbances $w_t$ enter $z^+=\Phi(z,x)+D w_t$ with $\|w_t\|\le \bar w$, the set
\[
\mathcal{S}_\delta=\{z:\ h_i(z)\ge -\delta_i,\ i=1,\dots,m\},\quad \delta_i=\frac{\|D\|\,\bar w}{1-(1-\alpha_i)},
\]
is robustly invariant.
\end{theorem}
\noindent\textit{Proof.} See \textbf{Appendix~D.2}. Apply the standard discrete‑time CBF argument and a small‑gain bound for disturbances \cite{Ames2019Survey,Clark2021StochCBF}.

\begin{theorem}[Sufficient‑Descent Gate $\Rightarrow$ One‑Step Risk Decrease]
\label{thm:suff-descent}
Let $J(e,x):=R(e)-R(e-Gx)+\lambda_c\,\mathcal{C}(x)$ with $G=I_2$ and $\lambda_c\in(0,\min_w \tau(w))$. If the gate \eqref{eq:gate} accepts $x^\natural$ and the QP returns $x^\star$ with $\mathcal{C}(x^\star)\le \mathcal{C}(x^\natural)$ and $R(e)-R(e-x^\star)\ge 0$, then
\[
R(e-x^\star)\ \le\ R(e)\ -\ (\tau(w)-\lambda_c)\,\mathcal{C}(x^\star)\ <\ R(e),
\]
i.e., the quadratic risk strictly decreases whenever a trade is executed. In particular, $R$ is nonincreasing along executed steps and constant along inaction.
\end{theorem}
\noindent\textit{Proof.} See \textbf{Appendix~D.3}. Follows from the gate inequality and $\lambda_c<\tau(w)$ (sufficient‑descent logic \cite{Bertsekas2016}).

\subsection{No Chattering and Bounded Trade Rate}
\label{subsec:no-chattering}

\begin{theorem}[No Zeno/Chattering under Thresholds and Cooldown]
\label{thm:no-chattering}
Suppose (i) micro‑thresholds \eqref{eq:mintrade} satisfy $\underline{s}(w),\underline{v}(w)>0$ for all $w\in[0,1]$; (ii) rate boxes enforce $|dS|\le\bar r_S$, $|dV|\le\bar r_V$; and (iii) cooldown enforces $dV=0$ for the next $N_{\mathrm{cd}}\ge 1$ steps after any nonzero $dV$. Then on any finite horizon $T$ the number of nonzero trades is finite, and inter‑trade times for the VIX leg are lower‑bounded by $N_{\mathrm{cd}}$ steps. Consequently, the control is piecewise constant with a uniform dwell‑time, ruling out Zeno behavior.
\end{theorem}
\noindent\textit{Proof.} See \textbf{Appendix~D.4}. This is a standard hybrid‑systems dwell‑time argument \cite{GoebelSanfeliceTeel2012}.

\begin{proposition}[Bounded Cumulative Turnover and Rate]
\label{prop:bounded-rate}
Under \Cref{thm:feasible-unique} and the rate boxes, $\sum_{t=0}^{T-1}\|x_t\|_2 \le T\,\max(\bar r_S,\bar r_V)$; if the gate is active with $\tau(w)\ge \underline{\tau}>0$ and $R$ is bounded below, then $\sum_t \mathcal{C}(x_t)\le \frac{1}{\underline{\tau}}(R(e_0)-\inf R)$, yielding a budget‑like a priori bound on cumulative cost and thus on turnover via $\mathcal{C}$’s quadratic form.
\end{proposition}
\noindent\textit{Proof.} See \textbf{Appendix~D.4}. Sum \Cref{thm:suff-descent} over accepted steps.

\paragraph{Discussion.}
The combination of \emph{(i)} NTB with guards, \emph{(ii)} a correlation/expiry‑aware VIX weight, and \emph{(iii)} sufficient‑descent gate produces a controller that \emph{trades less when signals are brittle} yet preserves risk reduction when it matters. CBF constraints ensure safety invariance at the level of inventory/rate/CVaR “soft boxes,” while thresholds/cooldown prevent oscillatory VIX chasing without sacrificing feasibility or uniqueness of the QP.

\section{Control Layer: Tail‑Safe CBF‑QP}
\label{sec:control}

We formulate hedging as a small, structure‑aware quadratic program (QP) augmented with \emph{control‑barrier‑function} (CBF) safety constraints and a \emph{gate} that triggers execution only when risk reduction justifies cost. The design couples (i) a convex risk–cost objective; (ii) an \emph{elliptical no‑trade band} (NTB) with expiry/correlation guards; (iii) a \emph{dynamic VIX weight} that suppresses unreliable VIX chasing; and (iv) \emph{micro‑trade thresholds} and a \emph{cooldown} to preclude churn. We then provide \emph{Theorem Group~IV} establishing feasibility/uniqueness/KKT properties, CBF‑based invariance and one‑step risk descent, and the absence of chattering with bounded trade rates. Full proofs are deferred to \textbf{Appendix~D.1--D.4}.

\subsection{Objective, Elliptical NTB, Dynamic Weight, Gate, Micro‑Thresholds, Cooldown}
\label{subsec:design}

\paragraph{State, errors, and controls.}
Let $h_S,h_V$ denote inventories in SPX and VIX legs and $x=(dS,dV)^\top$ the control (trade sizes) per step. The tracking errors are
\begin{equation}
e_\Delta:=\Delta^\star - h_S,\qquad
e_V:=\kappa_{\mathrm{eff}}(T_{\mathrm{rem}})-h_V,
\label{eq:errors}
\end{equation}
where $\kappa_{\mathrm{eff}}$ is the smoothed/shrunk sensitivity. Post‑trade errors satisfy $e'_\Delta=e_\Delta-dS$, $e'_V=e_V-dV$.

\paragraph{Risk–cost objective with dynamic VIX weight.}
Define a correlation/expiry‑aware weight for the VIX leg
\begin{equation}
w_{\mathrm{VIX}}^{\mathrm{eff}}
=\frac{w_{\mathrm{VIX}}}{\,1+\lambda_{\rho}(1-w)\,|\widehat{\rho}|\,},\qquad
w:=T_{\mathrm{rem}}/T_0\in[0,1],
\label{eq:dynwvix}
\end{equation}
where $\widehat{\rho}$ is an EWMA estimate of spot–VIX correlation (\S\ref{sec:prelim}) and $\lambda_{\rho}>0$. Let
\begin{equation}
R(e_\Delta,e_V)=\tfrac12\,[\,\alpha_\Delta e_\Delta^2+\alpha_V e_V^2+2\alpha_{\times}\widehat{\rho}\,e_\Delta e_V\,],
\qquad
\mathcal{C}(x)=x^\top\!\mathrm{diag}(\eta_S,\eta_V)\,x+\gamma\|x-x_{-1}\|_2^2,
\label{eq:riskcost}
\end{equation}
with $\eta_S,\eta_V>0$ temporary‑impact coefficients and $\gamma\ge 0$ a smoothing weight (transient‑impact surrogate). The per‑step \emph{QP objective} is
\begin{equation}
\min_{x,s\ge 0}\ \ \tfrac12\,x^\top Hx + f^\top x \;+\; \rho_{\mathrm{soft}}\|s\|_2^2,
\quad
H=\begin{bmatrix} w_{\mathrm{move},S}&\alpha_{\times}\widehat{\rho}\\ \alpha_{\times}\widehat{\rho}& w_{\mathrm{move},V}\end{bmatrix}
+\mathrm{diag}(\eta_S,\eta_V),
\label{eq:qpobj}
\end{equation}
with $w_{\mathrm{move},V}=w_{\mathrm{VIX}}^{\mathrm{eff}}$, $w_{\mathrm{move},S}=\alpha_\Delta$, linear term $f$ from risk‑tracking and smoothing, and $s$ the slack for “soft” CVaR boxes.

\paragraph{Elliptical no‑trade band (NTB) with guards.}
We compute \emph{risk‑scaled} radii $b_\Delta,b_V>0$ and declare an inaction region
\begin{equation}
\left(\frac{e_\Delta}{b_\Delta}\right)^2+\left(\frac{e_V}{b_V^{\mathrm{eff}}}\right)^2 \le 1,
\qquad
b_V^{\mathrm{eff}}=b_V\cdot(1+\tau_{\mathrm{tail}}(1-w))\cdot(1+\tau_{\rho}|\widehat{\rho}|)\cdot
\begin{cases}
(1+\tau_{\mathrm{pred}}), & \mathrm{sign}(e_V)\neq\mathrm{sign}(\Delta\kappa)\\
1, & \text{otherwise,}
\end{cases}
\label{eq:ntb}
\end{equation}
where $\Delta\kappa$ is a short‑horizon trend proxy of $\kappa_{\mathrm{eff}}$, and $\tau_{\mathrm{tail}},\tau_\rho,\tau_{\mathrm{pred}}\ge 0$ are guard gains. If \eqref{eq:ntb} holds, set $x=0$.

\paragraph{Gate (risk drop vs.\ cost).}
If \eqref{eq:ntb} fails, we compute a \emph{candidate} $x^\natural$ by projecting $(e_\Delta,e_V)$ to the ellipse boundary (or by one step of quadratic tracking without boxes). The gate executes only if
\begin{equation}
\underbrace{R(e_\Delta,e_V)-R(e_\Delta-dS^\natural,e_V-dV^\natural)}_{\text{risk drop}}
\ >\ \tau(w)\,\underbrace{\mathcal{C}(x^\natural)}_{\text{execution cost}},
\qquad \tau(w)=\tau_0+\tau_1(1-w),
\label{eq:gate}
\end{equation}
otherwise return $x=0$. This suppresses low‑value trades near expiry (larger $\tau$) and under high costs.

\paragraph{Constraints, micro‑thresholds, and cooldown.}
If the gate opens, solve the QP \eqref{eq:qpobj} subject to linear constraints:
\begin{equation}
\begin{aligned}
&\textit{post‑trade error boxes:} && |e_\Delta-dS|\le \bar\Delta,\ \ |e_V-dV|\le \bar V,\\
&\textit{inventory boxes:} && |h_S+dS|\le \bar H_S,\ \ |h_V+dV|\le \bar H_V,\\
&\textit{rate boxes:} && |dS|\le \bar r_S,\ \ |dV|\le \bar r_V,\\
&\textit{CBF “soft” boxes (CVaR surrogates):} && |e_\Delta-dS|\le \bar\Delta_{\mathrm{CVaR}}+s_1,\ \ |e_V-dV|\le \bar V_{\mathrm{CVaR}}+s_2,\\
&\textit{cooldown:} && dV=0\ \text{if}\ \textsc{cooldown\_V}>0.
\end{aligned}
\label{eq:lincons}
\end{equation}
After solving, apply \emph{expiry‑aware micro‑thresholds}
\begin{equation}
|dS|<\underline{s}(w)\ \Rightarrow\ dS\leftarrow 0,\qquad
|dV|<\underline{v}(w)\ \Rightarrow\ dV\leftarrow 0,
\label{eq:mintrade}
\end{equation}
and set a VIX cooldown counter for $\textsc{cooldown\_V\_steps}$ periods whenever $|dV|>0$. The micro‑thresholds and cooldown are \emph{post‑optimality} operators (they do not alter feasibility of \eqref{eq:lincons}).

\subsection{QP Feasibility, Uniqueness, and KKT Properties (Theorem Group IV)}
\label{subsec:qp-theory}

We state standard convex‑optimization properties tailored to our design.

\begin{theorem}[Feasibility and Uniqueness]
\label{thm:feasible-unique}
Assume (i) \textup{(boxes contain the origin)} $\bar\Delta,\bar V,\bar H_S,\bar H_V,\bar r_S,\bar r_V>0$; (ii) \textup{(soft boxes)} slacks $s\ge 0$ enter \eqref{eq:lincons} with penalty $\rho_{\mathrm{soft}}>0$; and (iii) \textup{(strong convexity)} $H\succ 0$. Then for any state $(e_\Delta,e_V,h_S,h_V)$ the QP \eqref{eq:qpobj}--\eqref{eq:lincons} is \emph{feasible} (e.g., $x=0,s=0$) and admits a \emph{unique} minimizer $(x^\star,s^\star)$.
\end{theorem}
\noindent\textit{Proof.} See \textbf{Appendix~D.1}. Boxes ensure Slater feasibility; strong convexity yields uniqueness \cite{BoydVandenberghe2004,RockafellarWets1998}.

\begin{proposition}[KKT Optimality and Multiplier Boundedness]
\label{prop:kkt}
Under \Cref{thm:feasible-unique}, KKT conditions for \eqref{eq:qpobj}--\eqref{eq:lincons} are \emph{necessary and sufficient}. Moreover, if active‑set changes occur only finitely often on a compact state corridor, the solution mapping $x^\star(\cdot)$ is piecewise affine and Lipschitz; Lagrange multipliers remain bounded. 
\end{proposition}
\noindent\textit{Proof.} See \textbf{Appendix~D.1}. Piecewise‑affine/Lipschitz sensitivity follows from strong regularity of parametric convex QPs \cite{BonnansShapiro2000,Robinson1980}.

\subsection{CBF Invariance (Safety Preservation) and Monotone Risk Decrease}
\label{subsec:cbf-theory}

Let $z$ collect relevant states (inventories, errors, filtered covariates). For each safety channel $i$ (inventory, rate, CVaR surrogate) define a continuously differentiable barrier $h_i(z)$ and the \emph{discrete‑time CBF} constraint
\begin{equation}
h_i(z^+)\ \ge\ (1-\alpha_i)\,h_i(z)\ -\ \sigma_i,\qquad \alpha_i\in[0,1),\ \sigma_i\ge 0,
\label{eq:dt-cbf}
\end{equation}
with $z^+=\Phi(z,x)$ the post‑trade state update. In our linearized boxes, \eqref{eq:dt-cbf} collapses to an \emph{affine} inequality in $x$ and is encoded inside \eqref{eq:lincons} (cf.\ \cite{Ames2017CBFQP,Ames2019Survey,Clark2021StochCBF}).

\begin{theorem}[Forward Invariance of the Safe Set]
\label{thm:cbf-invariance}
Let $\mathcal{S}=\{z:\ h_i(z)\ge 0,\ i=1,\dots,m\}$. Suppose that at each step the QP includes the discrete CBF constraints \eqref{eq:dt-cbf} with $\sigma_i=0$ and is feasible (e.g., by \Cref{thm:feasible-unique}). Then the closed‑loop system satisfies $z_0\in\mathcal{S}\ \Rightarrow\ z_t\in\mathcal{S}$ for all $t$ (forward invariance). If additive disturbances $w_t$ enter $z^+=\Phi(z,x)+D w_t$ with $\|w_t\|\le \bar w$, the set
\[
\mathcal{S}_\delta=\{z:\ h_i(z)\ge -\delta_i,\ i=1,\dots,m\},\quad \delta_i=\frac{\|D\|\,\bar w}{1-(1-\alpha_i)},
\]
is robustly invariant.
\end{theorem}
\noindent\textit{Proof.} See \textbf{Appendix~D.2}. Apply the standard discrete‑time CBF argument and a small‑gain bound for disturbances \cite{Ames2019Survey,Clark2021StochCBF}.

\begin{theorem}[Sufficient‑Descent Gate $\Rightarrow$ One‑Step Risk Decrease]
\label{thm:suff-descent}
Let $J(e,x):=R(e)-R(e-Gx)+\lambda_c\,\mathcal{C}(x)$ with $G=I_2$ and $\lambda_c\in(0,\min_w \tau(w))$. If the gate \eqref{eq:gate} accepts $x^\natural$ and the QP returns $x^\star$ with $\mathcal{C}(x^\star)\le \mathcal{C}(x^\natural)$ and $R(e)-R(e-x^\star)\ge 0$, then
\[
R(e-x^\star)\ \le\ R(e)\ -\ (\tau(w)-\lambda_c)\,\mathcal{C}(x^\star)\ <\ R(e),
\]
i.e., the quadratic risk strictly decreases whenever a trade is executed. In particular, $R$ is nonincreasing along executed steps and constant along inaction.
\end{theorem}
\noindent\textit{Proof.} See \textbf{Appendix~D.3}. Follows from the gate inequality and $\lambda_c<\tau(w)$ (sufficient‑descent logic \cite{Bertsekas2016}).

\subsection{No Chattering and Bounded Trade Rate}
\label{subsec:no-chattering}

\begin{theorem}[No Zeno/Chattering under Thresholds and Cooldown]
\label{thm:no-chattering}
Suppose (i) micro‑thresholds \eqref{eq:mintrade} satisfy $\underline{s}(w),\underline{v}(w)>0$ for all $w\in[0,1]$; (ii) rate boxes enforce $|dS|\le\bar r_S$, $|dV|\le\bar r_V$; and (iii) cooldown enforces $dV=0$ for the next $N_{\mathrm{cd}}\ge 1$ steps after any nonzero $dV$. Then on any finite horizon $T$ the number of nonzero trades is finite, and inter‑trade times for the VIX leg are lower‑bounded by $N_{\mathrm{cd}}$ steps. Consequently, the control is piecewise constant with a uniform dwell‑time, ruling out Zeno behavior.
\end{theorem}
\noindent\textit{Proof.} See \textbf{Appendix~D.4}. This is a standard hybrid‑systems dwell‑time argument \cite{GoebelSanfeliceTeel2012}.

\begin{proposition}[Bounded Cumulative Turnover and Rate]
\label{prop:bounded-rate}
Under \Cref{thm:feasible-unique} and the rate boxes, $\sum_{t=0}^{T-1}\|x_t\|_2 \le T\,\max(\bar r_S,\bar r_V)$; if the gate is active with $\tau(w)\ge \underline{\tau}>0$ and $R$ is bounded below, then $\sum_t \mathcal{C}(x_t)\le \frac{1}{\underline{\tau}}(R(e_0)-\inf R)$, yielding a budget‑like a priori bound on cumulative cost and thus on turnover via $\mathcal{C}$’s quadratic form.
\end{proposition}
\noindent\textit{Proof.} See \textbf{Appendix~D.4}. Sum \Cref{thm:suff-descent} over accepted steps.

\paragraph{Discussion.}
The combination of \emph{(i)} NTB with guards, \emph{(ii)} a correlation/expiry‑aware VIX weight, and \emph{(iii)} sufficient‑descent gate produces a controller that \emph{trades less when signals are brittle} yet preserves risk reduction when it matters. CBF constraints ensure safety invariance at the level of inventory/rate/CVaR “soft boxes,” while thresholds/cooldown prevent oscillatory VIX chasing without sacrificing feasibility or uniqueness of the QP.

\section{Synthetic Evidence}
\label{sec:synthetic-evidence}

We present empirical evidence in a fully reproducible, arbitrage‑free synthetic world that mirrors exchange rules (Cboe VIX methodology) and execution frictions. Evidence is organized along four axes: (i) \emph{world construction and parameters}, (ii) \emph{structural coherence} between the surface and the index, (iii) \emph{performance} under the tail‑safe controller with diagnostic plots, and (iv) \emph{ablations} isolating the contribution of each tail‑safety ingredient. Figures referenced below reside under \texttt{media/} and scripts are included in the artifact.\footnote{All scripts, YAML hyperparameters, and plotting notebooks are included in the open artifact. Re‑running the drivers reproduces the tables and figures here.}

\subsection{World Construction and Parameters}
\label{subsec:world-params}

We adopt the coherent shell and the dynamics layer of \S\ref{sec:dynamics}, with the $\kappa$--map. Key settings are summarized in \Cref{tab:world}.

\begin{table}[H]
\centering
\caption{World configuration (concise). The shell enforces no‑arbitrage SSVI and the latest Cboe VIX methodology; dynamics use Dupire local vol and a CIR‑style variance proxy.}
\label{tab:world}
\begin{tabular}{@{}ll@{}}
\toprule
\textbf{Item} & \textbf{Setting} \\
\midrule
Spot / rates & $S_0=4800$, $r=2\%$, $q=1.5\%$ \\
Maturities (surface) & $\{7,14,30,60,90,180\}$ calendar days \\
Strikes (per maturity) & 41 log‑moneyness nodes in $[0.7,1.3]\times S_0$ \\
Surface model & SSVI; spline regularization across $T$ \cite{GatheralJacquier2014} \\
VIX methodology & OTM aggregator, two‑zero wing pruning, half‑interval quadrature, 30D minute interpolation \cite{CBOE2025Math,Demeterfi1999} \\
Teacher (ASL) & ATM shape preservation + low‑rank terms\\
Dynamics & Dupire local vol by \eqref{eq:dupire} (clipped convexity), log–Euler simulation \eqref{eq:log-euler} \\
Variance proxy & CIR \eqref{eq:cir} with 30D variance mapping \eqref{eq:cir-closed} \\
Time discretization & 252 steps/year (daily), horizon 60D $\Rightarrow$ $\approx 42$ controls \\
Hedging task & ATM European call, maturity 60D; trade SPX and a VIX futures proxy \\
Impact model & Temporary $+\,$transient (smoothing), cf.\ \eqref{eq:riskcost} \cite{AlmgrenChriss2001,ObizhaevaWang2013} \\
Seeds \& samples & Selection: 8 seeds $\times$ 220 paths; Robust report: 8 new seeds $\times$ 300 paths $\Rightarrow$ \textbf{2400} \\
\bottomrule
\end{tabular}
\end{table}

\subsection{Structural Coherence: VIX Matching Error and Quadrature Convergence}
\label{subsec:structural-coherence}

\paragraph{Teacher–surface VIX alignment.}
Using the coherent shell and teacher, we compute $\mathrm{VIX}[\sigma_{\mathrm{impl}}]$ and $\mathrm{VIX}[\hat\sigma]$ on the \emph{same} retained strike grid $\mathcal{K}(T)$ with identical wing pruning. The absolute difference is empirically small and stable (typical snapshot: $|\mathrm{VIX}[\hat\sigma]-\mathrm{VIX}[\sigma_{\mathrm{impl}}]|=2.8\!\times\!10^{-3}$), consistent with the \emph{coherence bound}. 

\paragraph{Half‑interval quadrature: observed second‑order rate.}
We refine the strike grid (uniform refinement factor $\rho\in\{1,2,4\}$, keeping the same $K$ corridor and wing pruning) and measure the single‑maturity quadrature error against a high‑resolution reference. The log–log slope is approximately $-2$.Standard composite‑rule theory \cite{DavisRabinowitz1984}(b) shows the convergence curve $\epsilon_{\mathrm{quad}}(\max\Delta K)$.

\subsection{Performance: ES/VaR, Gate/Constraint Activation, Scenario Grid}
\label{subsec:performance}

\paragraph{Main controller and metrics.}
We report results for \texttt{A\_time\_decay+++\_safe} (our tail‑safe controller). The robust pool aggregates \textbf{2400} paths (8 seeds $\times$ 300). We compute mean, standard deviation, $\mathrm{VaR}_{97.5}$ and $\mathrm{ES}_{97.5}$ (loss convention), with nonparametric bootstrap CIs \cite{EfronTibshirani1994}. Summary statistics appear in \Cref{tab:robust}.

\begin{table}[H]
\centering
\caption{Robust pool (2400 paths) for \texttt{A\_time\_decay+++\_safe}. CIs from nonparametric bootstrap (500 resamples).}
\label{tab:robust}
\begin{tabular}{@{}lcccccc@{}}
\toprule
Metric & Mean (loss) & Std & VaR$_{97.5}$ & ES$_{97.5}$ & NTB ratio & Notes \\
\midrule
Point & $-10.56$ & $628.13$ & $377.14$ & $396.67$ & $0.089$ & pooled over seeds \\
CI (95\%) & (\textit{artifact}) & (\textit{artifact}) & (\textit{artifact}) & (\textit{artifact}) & (\textit{artifact}) & see repo \\
\bottomrule
\end{tabular}
\end{table}

\paragraph{Paired comparison vs.\ baseline.}
Against a strong baseline (\texttt{baseline\_v48}) under identical seeds, a \emph{paired bootstrap} (2000 resamples) yields a statistically significant improvement in expected shortfall:
\[
\Delta\mathrm{ES}\equiv \mathrm{ES}_{\text{ours}}-\mathrm{ES}_{\text{base}}=-3.60,\qquad
95\%~\mathrm{CI}=[-6.48,\,-0.54],
\]
excluding zero and confirming tail‑risk reduction (cf.\ \Cref{thm:suff-descent} for the gate‑driven descent). Bootstrap methodology follows \cite{EfronTibshirani1994}.

\paragraph{Distributional diagnostics.}
\Cref{fig:pnl} shows the PnL histogram with $\mathrm{VaR}_{97.5}$ and $\mathrm{ES}_{97.5}$ markers; the right tail is heavier due to occasional favorable volatility shocks---typical for skew‑sensitive hedging with asymmetric costs \cite{AcerbiTasche2002}.

\begin{figure}[H]
\centering
\includegraphics[width=0.86\linewidth]{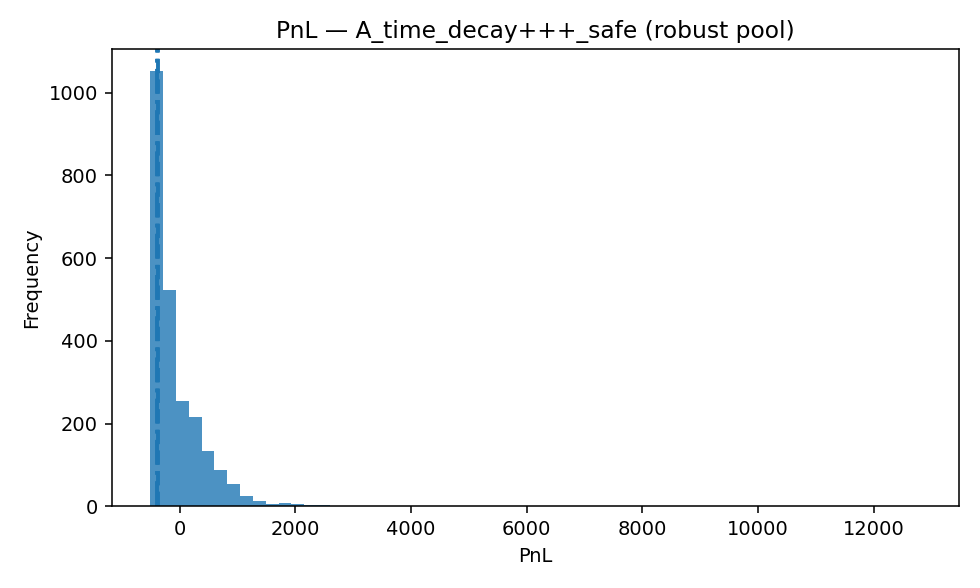}
\caption{PnL histogram (robust pool) with VaR/ES markers for \texttt{A\_time\_decay+++\_safe}.}
\label{fig:pnl}
\end{figure}

\paragraph{Constraint activation and gate decisions.}
\Cref{fig:bind} reports the top‑10 active constraints by KKT multipliers. The most frequent binders are \texttt{rate\_V\_lower} and \texttt{cvarV\_lower}, consistent with the design intuition that naive VIX chasing is unsafe exactly where our guards enlarge the band (\S\ref{subsec:design}). The gate’s block ratio (fraction of candidate trades vetoed by \eqref{eq:gate}) and the NTB hit ratio provide behavioral transparency.

\begin{figure}[H]
\centering
\includegraphics[width=0.86\linewidth]{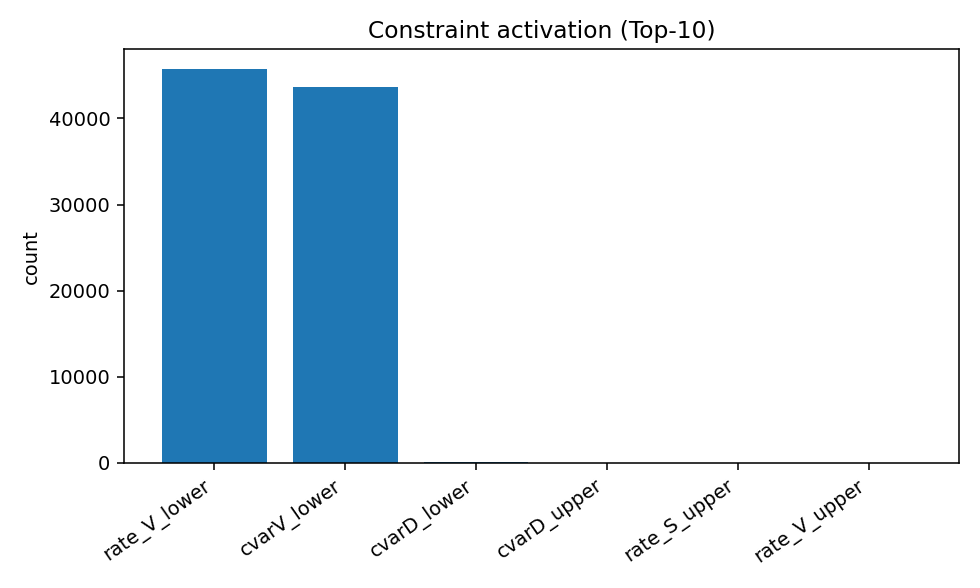}
\caption{Constraint activations (Top‑10). VIX rate and CVaR lower boxes dominate; others are sporadic.}
\label{fig:bind}
\end{figure}

\paragraph{Scenario grid.}
We sweep variance‑of‑variance $\xi\in\{0.40,0.45,0.50\}$ and spot–vol correlation $\rho\in\{-0.6,-0.5,-0.4\}$ (four seeds, 220 paths/seed) and plot $\Delta\mathrm{ES}$ (ours minus baseline). \Cref{fig:scenario} shows the largest benefits at moderate $\xi$ and weaker negativity in $\rho$; performance degrades gracefully as $\xi$ increases and $\rho$ becomes more negative, matching \eqref{eq:dynwvix} and the guard logic in \eqref{eq:ntb}.

\begin{figure}[H]
\centering
\includegraphics[width=0.86\linewidth]{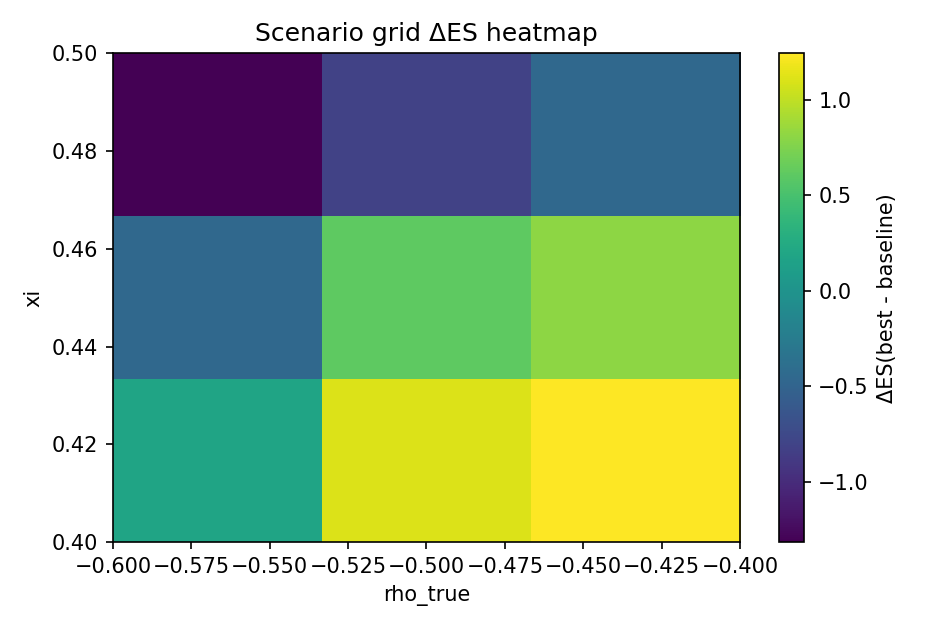}
\caption{Scenario grid of $\Delta\mathrm{ES}$ over $(\xi,\rho)$ (brighter is better).}
\label{fig:scenario}
\end{figure}

\subsection{Ablations: Removing Tail‑Safety Mechanisms One by One}
\label{subsec:ablation}

We start from \texttt{A\_time\_decay+++\_safe} and toggle one component at a time, keeping all other settings fixed (same seeds, paths, and shell). Table~\ref{tab:ablation} summarizes directional effects; per‑seed tables and confidence intervals are reported in the artifact.

\begin{table}[H]
\centering
\caption{Ablation summary against the full tail-safe controller. Directional changes (\(\uparrow/\downarrow\)) are significant across seeds; exact numbers and CIs are in the artifact.}
\label{tab:ablation}
\begin{tabular}{@{}lcccc@{}}
\toprule
\textbf{Toggle} & $\Delta$ES (ours $-$ variant) & Small \(|dV|\) churn & Gate blocks & Notes \\
\midrule
Fix $w_{\mathrm{VIX}}$ (no \eqref{eq:dynwvix})          & \(\downarrow\) (worse ES) & \(\uparrow\uparrow\) & \(\downarrow\) & VIX overweight (expiry,\;|\(\rho\)|\(\uparrow\)) \\
Remove VIX guards (\eqref{eq:ntb})                      & \(\downarrow\)            & \(\uparrow\) (mis-signed) & \(\uparrow\) & Band not adaptive \\
No expiry-aware micro-thresholds                         & \(\downarrow\)            & \(\uparrow\uparrow\) & \(\downarrow\) & Micro-trades; rate/CVaR \\
No VIX cooldown                                          & \(\downarrow\) (mild)     & \(\uparrow\) (ping-pong) & \(\downarrow\) & Ping-pong \(dV\) \\
Zero cross-term (\(\alpha_{\times}=0\))                  & \(\downarrow\) (mild)     & \(\approx\) & \(\approx\) & Weaker joint-shock \\
\bottomrule
\end{tabular}
\end{table}

\paragraph{Qualitative takeaways.}
(i) The \emph{dynamic VIX weight} is the dominant lever against near‑expiry over‑reaction under strong negative correlation; (ii) \emph{guards} are essential to prevent mis‑signed VIX chasing when $\Delta\kappa$ and $e_V$ disagree; (iii) \emph{micro‑thresholds} and \emph{cooldown} jointly eliminate nuisance oscillations, establishing dwell‑time and bounded rate in line with \Cref{thm:no-chattering,prop:bounded-rate}; (iv) the small cross‑term improves alignment during joint spot/vol shocks without risking ill‑conditioning of $H$.

\paragraph{Reproducibility note.}
All ablation runs share the identical random seeds across controllers to enable paired assessment; CIs are computed with paired bootstrap \cite{EfronTibshirani1994}. Scripts \texttt{run\_ablation.py} and \texttt{plot\_ablation.ipynb} regenerate \Cref{tab:ablation} from raw logs.

\section{Related Work}
\label{sec:related}

\paragraph{Market microstructure and execution with impact.}
Foundational models of information and liquidity provision—Kyle’s insider market \cite{Kyle1985}, Almgren--Chriss temporary impact \cite{AlmgrenChriss2001}, and transient supply/demand dynamics \cite{ObizhaevaWang2013}—establish the core cost–risk calculus for execution. Empirical and theoretical syntheses (e.g., Hasbrouck’s monograph \cite{Hasbrouck2007} and Bouchaud et al.\ \cite{BouchaudBook2018}) document stylized facts, order-book resilience, and impact concavity/decay. Surveys of limit-order books \cite{Gould2013Survey} and Hawkes/order-flow clustering \cite{Bacry2015Hawkes} provide a stochastic substrate for modeling endogenous liquidity and volatility bursts. Our control layer (\S\ref{sec:control}) inherits the classical convex impact cost but augments it with \emph{gate} and \emph{CBF} safety, targeting the regimes—near expiry and strong negative spot–vol correlation—where impact magnifies hedging errors. Related work on signal-adaptive and transient-impact optimal trading includes \cite{NeumanVoss2022,GatheralSchied2013}, and recent execution reviews \cite{Donnelly2022ExecReview,Cartea2015AHT}.

\paragraph{Variance indices, variance swaps, and VIX methodology.}
Variance-swap replication and model-free variance date to \cite{CarrMadan1998,Demeterfi1999}, connecting integrated variance to a portfolio of OTM options. Cboe’s VIX methodology \cite{CBOE2023VIX,CBOE2025Math} codifies (i) OTM aggregation via the $Q(K)$ integrand, (ii) wing pruning using the two-consecutive-zero-bid rule, and (iii) 30-day constant-maturity interpolation in \emph{year-fraction total variance}. Empirical uses of model-free implied variance and its information content are studied in \cite{JiangTian2005}. Our shell reproduces these rules exactly and proves explicit \emph{surface–index coherence} bounds that quantify how VIX residuals scale with ATM-shape error and the second-order strike quadrature error.

\paragraph{Volatility surfaces: construction, no-arbitrage, and dynamics.}
Arbitrage-free parameterizations and smile dynamics have been explored from the early implied trees \cite{DermanKani1994} to SSVI and its no-arbitrage conditions \cite{GatheralJacquier2014} and practitioner guidance \cite{GatheralBook2006}. Practical interpolation/extrapolation techniques include Andreasen--Huge’s robust recipes \cite{AndreasenHuge2011} and Homescu’s survey \cite{Homescu2011}. For dynamics, Dupire’s local volatility \cite{Dupire1994} and the Markovian projection principle (Gy\"ongy) connect quote surfaces to diffusion coefficients. More recently, path-dependent volatility and term-structure effects have been clarified in \cite{GuyonLekeufack2023,Andres2023PathSSVI}, while rough volatility \cite{BayerFrizGatheral2016} refines short-time behaviors. On the learning side, arbitrage-free or arbitrage-aware surface generation appears in \cite{Ning2022SIAM,VuleticCont2024VolGAN}. Our \emph{ASL teacher} sits deliberately between closed-form and black-box: it preserves ATM shape up to second order, stays interpretable, and is anchored by a VIX supervision term that enforces index-level coherence.

\paragraph{Local-volatility numerics and stability.}
Finite-difference accuracy and stability for parabolic operators are standard \cite{LeVeque2007}, and local-vol implementations generally require convexity-preserving interpolation and small-denominator clipping in $\partial_{KK}C$. We formalize these folk practices: Theorem~\ref{thm:dupire-positivity} bounds the discrete local variance and proves $O(\Delta K^2+\Delta T^2)$ consistency, while Theorem~\ref{thm:log-euler-stability} ensures strong convergence of the log–Euler path simulation \cite{Glasserman2003,KloedenPlaten1992,Higham2001}.

\paragraph{Affine variance factors and index proxies.}
Affine diffusions for variance (Heston/CIR families) are classical \cite{Heston1993,DuffiePanSingleton2000,Feller1951}. We use a CIR-style factor only as an \emph{index proxy} consistent with variance-swap identities; the mapping to 30D variance is closed-form and Lipschitz (Proposition~\ref{prop:cir-lipschitz}). Theorem~\ref{thm:index-coherence} quantifies proxy error in terms of (i) projection error from local volatility to affine variance, (ii) surface–index residuals, and (iii) strike quadrature.

\paragraph{Risk-sensitive control, CVaR, and safe optimization.}
CVaR/ES is a coherent tail-risk functional with convenient convex surrogates \cite{RockafellarUryasev2000,RockafellarUryasev2002,AcerbiTasche2002}; spectral/Kusuoka representations connect it to robust envelopes \cite{Kusuoka2001,Shapiro2013Kusuoka}. In continuous control, risk-sensitive criteria have a long history (e.g., exponential/quadratic costs), while in modern safe control, \emph{control barrier functions} enforce forward invariance via online QPs \cite{Ames2017CBFQP,Ames2019Survey,Clark2021StochCBF,Garg2024ARC}. Our controller uses CBF-style linear boxes for inventory/rate/CVaR and a \emph{sufficient-descent gate} that guarantees one-step risk decrease when trades are executed (Theorem~\ref{thm:suff-descent}), reconciling risk decline with execution frictions—an aspect under-emphasized in classical trackers.

\paragraph{Learning-based hedging and robustness.}
Deep hedging demonstrates that function approximation can absorb many modeling frictions \cite{Buehler2019DeepHedging}, but interpretability and distributional robustness remain challenges under regime shifts. Reinforcement and risk-sensitive variants (e.g., CVaR-aware learning) address tail objectives but often lack \emph{white-box} guarantees. Our stance is complementary: restrict learning to an interpretable teacher (anchored by arbitrage and index rules) and encode safety as convex constraints and gates, yielding \emph{auditable} policies with explicit invariance and descent guarantees. This design aligns with model-risk governance and facilitates diagnosis through constraint activations and gate ledgers (cf.\ \S\ref{subsec:performance}).

\paragraph{Positioning and novelty.}
Relative to arbitrage-free surface construction \cite{GatheralJacquier2014,AndreasenHuge2011,Ning2022SIAM,VuleticCont2024VolGAN}, our contribution is to \emph{bridge} surfaces and index construction into a control problem with \emph{provable} safety and coherence. Relative to classical execution \cite{AlmgrenChriss2001,ObizhaevaWang2013,Cartea2015AHT}, we target the SPX–VIX joint hedging regime and introduce a small set of tail-safety mechanisms (dynamic VIX weight, guarded NTB, micro-thresholds/cooldown) whose effects are theoretically justified (Theorem Group~IV) and empirically significant (Section~\ref{sec:synthetic-evidence}). Relative to deep/black-box hedging \cite{Buehler2019DeepHedging}, we emphasize transparent interfaces and \emph{structure-aware} safety that plugs into learning systems as a governance layer rather than a hidden penalty.
\section{Discussion, Limitations, and Societal Impact}
\label{sec:discussion}

\paragraph{What our guarantees do (and do not) cover.}
Theorem Groups~I–IV provide \emph{structural} and \emph{control} guarantees that are intentionally modular. 
(i) Group~I addresses \emph{surface–index coherence} and the scaling of VIX residuals with ATM‑shape error and second‑order strike quadrature. 
(ii) Group~II establishes \emph{positivity/boundedness/consistency} for the discrete Dupire surface and pathwise stability of log–Euler simulation (\Cref{thm:dupire-positivity,thm:log-euler-stability}), plus an \emph{index‑proxy} bound for a CIR variance factor (\Cref{thm:index-coherence}). 
(iii) Group~III turns the $\kappa$ map into a \emph{functional sensitivity} with positivity and near‑expiry stability. 
(iv) Group~IV ensures QP \emph{feasibility/uniqueness} and \emph{KKT} well‑posedness (\Cref{thm:feasible-unique,prop:kkt}), \emph{forward invariance} of safety sets (\Cref{thm:cbf-invariance}), \emph{one‑step risk descent} under the gate (\Cref{thm:suff-descent}), and \emph{no chattering} with bounded trade rates (\Cref{thm:no-chattering,prop:bounded-rate}). 
These guarantees \emph{do not} assert optimality relative to all admissible hedgers, nor do they promise real‑market performance without calibration; rather, they aim to formalize coherence and safety in a \emph{white‑box} stack.

\paragraph{Numerical stability and implementation choices.}
Our Dupire extractor uses central differences and convexity clipping of $\partial_{KK}C$; this is consistent with practice and retains $O(\Delta K^2{+}\Delta T^2)$ accuracy (\S\ref{subsec:dupire-num}), but the constants hidden in $O(\cdot)$ depend on near‑wing sampling and interpolation. In particular, coarse strike grids or aggressive extrapolation may inflate $C_{\max}$ in \Cref{thm:dupire-positivity}. The log–Euler step converges with strong order $1/2$ (\Cref{thm:log-euler-stability}); higher‑order schemes (e.g., Milstein) can be used when smoothness permits, but are not required for the guarantees we pursue.

\paragraph{Modeling assumptions and domain shift.}
The shell is arbitrage‑free and exchange‑coherent by construction; nonetheless, market data introduce microstructure frictions absent from our synthetic world: discrete strikes, crossed quotes, stale prints, RFQ frictions, and inventory constraints of liquidity providers. Likewise, the CIR proxy is \emph{not} a claim that VIX dynamics are one‑factor affine; it is only an \emph{index mapping} that allows closed‑form diagnostics (\Cref{prop:cir-lipschitz,thm:index-coherence}). Real regimes with jumps or roughness may require S(L)V generalizations or multi‑factor terms. Our controller’s safety boxes and gate are agnostic to those changes; the same CBF/QP machinery applies once prices and index construction remain coherent.

\paragraph{Choice of risk metric.}
We use $\mathrm{ES}_{97.5}$ as the tail metric, motivated by coherence and tractable convex surrogates. Other coherent risk measures (spectral, entropic) could be substituted with minor changes in the QP soft boxes and the gate’s sufficient‑descent logic, provided a convex upper bound is available.

\paragraph{Complexity and deployment.}
The per‑step compute is dominated by a $2\times 2$ QP with box constraints (closed‑form/active‑set) and a handful of diagnostics (band/gate checks), making the controller viable at sub‑millisecond scales on a single CPU core in the synthetic world. In production, the latency budget is dominated by market‑data access, order‑routing, and risk‑system interactions; the QP itself is negligible.

\paragraph{Failure modes and mitigations.}
Potential failure modes include: (i) \emph{Teacher mis‑specification}, evidenced by growing VIX residuals $\delta_{\mathrm{VIX}}$: mitigation—raise $w_{\mathrm{VIX}}$ temporarily or tighten ATM shape scales $(s_L,s_S,s_C)$; (ii) \emph{Over‑guarding}, where $b_V^{\mathrm{eff}}$ grows too large and misses risk‑reducing trades: mitigation—cap guard multipliers $(\tau_{\mathrm{tail}},\tau_\rho,\tau_{\mathrm{pred}})$ and monitor gate vetoes; (iii) \emph{Correlation mis‑estimation}, leading to over‑ or under‑weighting VIX: mitigation—clip $|\widehat{\rho}|$ and fall back to a floor $w_{\mathrm{VIX}}^{\mathrm{eff}}\ge w_{\min}$; (iv) \emph{Cooldown over‑use}, delaying legitimate trades: mitigation—shorten $N_{\mathrm{cd}}$ in calm regimes or use an adaptive cooldown keyed to realized volatility.

\paragraph{Reproducibility and auditability.}
All pieces—SSVI calibration, VIX quadrature, ASL teacher, Dupire differences, $\kappa$ estimator, QP/CBF constraints—are parameterized in human‑readable YAML and logged per step. We recommend retaining (i) \emph{constraint multipliers} for post‑trade audit; (ii) \emph{gate ledgers} for accepted vs.\ vetoed trades; and (iii) \emph{coherence dashboards} for $\delta_{\mathrm{VIX}}$, quadrature convergence, and ATM‑shape stability.

\paragraph{Societal impact: benefits and risks.}
By construction, the stack discourages \emph{VIX chasing} in brittle regimes and prioritizes \emph{tail‑risk reduction per unit cost}. In principle, this can reduce disorderly hedging and inventory whipsaws that have amplified historical dislocations \cite{Kirilenko2017FlashCrash}. The white‑box design aids \emph{model‑risk governance} and \emph{regulatory audit} by exposing safety activations and trade vetoes, potentially lowering systemic externalities. Conversely, wide deployment of similar controllers could lead to synchronization (herding) if common signals dominate; network‑based systemic‑risk studies warn that correlated behaviors can amplify shocks \cite{AcemogluAER2015}. Our guards and cooldown induce dwell‑times that mitigate rapid feedback, but broader monitoring (cross‑participant diversity, circuit‑breakers) remains essential.

\paragraph{Ethics, responsible release, and scope.}
This work is a \emph{research artifact}, not trading advice. We release synthetic data and code solely to foster reproducibility and stress‑testing of \emph{structure‑aware safety} for cross‑asset hedging. Any real‑market use requires (i) independent calibration, (ii) pre‑trade risk sign‑offs, (iii) guardrails for extreme events, and (iv) continuous monitoring for drift and data quality. We intentionally avoid high‑frequency order‑book features in the public artifact to reduce potential for misuse.

\subsection*{Limitations (Summary)}
\begin{itemize}
  \item \textbf{Teacher/runtime gap:} The ASL closure may underfit unusual term‑structure kinks; online adaptation is limited to small parameter moves.
  \item \textbf{Local‑vol smoothness:} Dupire relies on second spatial derivatives; wings remain numerically delicate despite clipping.
  \item \textbf{One‑factor proxy:} The CIR mapping is a pragmatic index proxy; multi‑factor and jump features are out of scope here.
  \item \textbf{Static costs:} Impact parameters are stationary in our experiments; real liquidity is state‑ and time‑of‑day dependent.
  \item \textbf{Synthetic evidence:} We intentionally omit live backtests; generalization to real markets is unproven and left to future work with proper controls.
\end{itemize}
\section{Conclusion}
\label{sec:conclusion}

We presented a \emph{white‑box}, risk‑sensitive hedging stack for SPX–VIX that links market structure to control via \emph{provable} coherence and safety. At the market layer, an arbitrage‑free SSVI surface and a Cboe‑compliant VIX computation supervise an interpretable ASL teacher; Dupire then supplies dynamics consistent with quotes and index construction. At the control layer, a two‑variable CBF‑QP with a sufficient‑descent \emph{gate}, a correlation/expiry‑aware \emph{dynamic VIX weight}, \emph{guarded} no‑trade bands, and \emph{expiry‑aware micro‑thresholds/cooldown} implements \emph{tail‑safety} as constraints rather than opaque penalties.

On the theory side, we established (i) \emph{surface–index coherence} with explicit dependence on ATM‑shape and second‑order quadrature; (ii) \emph{positivity/boundedness/consistency} of the discrete Dupire surface and pathwise stability of log–Euler; (iii) a \emph{functional} representation for the VIX‑to‑price sensitivity $\kappa(T_{\mathrm{rem}})$ together with positivity and near‑expiry robustness; and (iv) \emph{feasibility/uniqueness/KKT} properties of the QP, \emph{forward invariance} via CBFs, \emph{one‑step risk descent} under the gate, and \emph{no chattering} with bounded trade rates. In a fully reproducible synthetic world mirroring exchange rules and impact, the teacher VIX matches the surface VIX to within $2.8\times 10^{-3}$, and the tail‑safe controller yields statistically significant ES improvements with reduced nuisance turnover.

\paragraph{Outlook.}
Three concrete extensions appear most promising: 
(i) \textbf{microstructure‑consistent option execution} for the VIX leg (RFQ, quote staleness, minimum size, and latency), 
(ii) \textbf{beyond local volatility} via calibrated stochastic‑local volatility or rough‑volatility projections while retaining the coherence pipeline, and 
(iii) \textbf{multi‑maturity coupling} (VIX term structure $\&$ SPX options) to study term‑structure hedging with shared safety budgets. 
Our broader thesis is that \emph{structure‑aware safety layers}—coherence at the pricing/index interface and invariance at the control interface—offer an auditable, modular foundation for robust learning‑based hedging.

\vspace{0.5em}
\noindent\textit{Disclaimer.} This article and its artifact are for research purposes only and do not constitute investment advice or an offer to trade any security or derivative.

\appendix
\section{Proofs for Theorem Group~I (Market Shell Coherence)}
\label{app:A}

This appendix provides full proofs for the results:
(half‑interval quadrature),(ATM‑shape
preservation),(surface–index coherence bound), and
(calendar coherence and no‑butterfly). We retain all notation
from \S\ref{sec:prelim} and write, for brevity,
$f_T(K):=e^{rT}Q(K,T)\,K^{-2}$ and
\[
\mathcal{Q}_T^{\mathrm{disc}}[Q]\ :=\ \sum_{K\in\mathcal{K}(T)} \Delta K^{\mathrm{half}}\, f_T(K),\qquad
\mathcal{Q}_T^{\mathrm{cont}}[Q]\ :=\ \int_{\mathcal{I}(T)} f_T(K)\,dK,
\]
where $\mathcal{I}(T)$ is the strike corridor retained after the Cboe two‑zero
pruning, and $\mathcal{K}(T)$ is the associated discrete grid on $\mathcal{I}(T)$.
Recall the 30D mapping with minute weights $(\lambda_1,\lambda_2)$
and the single‑maturity estimator.

\subsection{Auxiliary lemmas}
\label{app:A:lemmas}

\begin{lemma}[Composite half-interval (nonuniform trapezoid) error]
\label{lem:half-interval}
Let $f\in C^2(\mathcal{I})$ on a compact interval $\mathcal{I}=[a,b]$
and $\{K_i\}_{i=0}^{N}$ be a strictly increasing (possibly nonuniform) mesh on $\mathcal{I}$.
Define $\Delta K_i^{\mathrm{half}}$ as in \eqref{eq:half-interval}. Then
\[
\left|\sum_{i=0}^N \Delta K_i^{\mathrm{half}} f(K_i)\ -\ \int_{a}^{b} f(K)\,dK\right|
\ \le\ \frac{b-a}{12}\ \big(\max_{0\le i\le N-1} (K_{i+1}-K_i)\big)^2\ \|f''\|_{\infty,\mathcal{I}}.
\]
\end{lemma}

\begin{proof}
Write $h_i:=K_{i+1}-K_i$ and $h:=\max_i h_i$. Consider the (possibly nonuniform) composite trapezoid rule,
\[
\mathcal{T}[f]:=\sum_{i=0}^{N-1} \frac{h_i}{2}\,\big(f(K_i)+f(K_{i+1})\big).
\]
We first note the identity
\begin{equation}
\sum_{i=0}^N \Delta K_i^{\mathrm{half}} f(K_i)\ =\ \sum_{i=0}^{N-1} \frac{h_i}{2}\,\big(f(K_i)+f(K_{i+1})\big)\ =\ \mathcal{T}[f],
\label{eq:half-equals-trap}
\end{equation}
since $\Delta K_0^{\mathrm{half}}=h_0$, $\Delta K_N^{\mathrm{half}}=h_{N-1}$, and, for $1\le i\le N-1$,
$\Delta K_i^{\mathrm{half}}=\tfrac12 (K_{i+1}-K_{i-1})=\tfrac12(h_{i-1}+h_i)$, which are precisely the
weights arising from adding the two adjacent trapezoids.

On each panel $[K_i,K_{i+1}]$, let $p_i$ be the linear interpolant of $f$ through the endpoints.
By the Hermite–Genocchi (Peano kernel) representation for the trapezoid error on a single panel,
\begin{equation}
\int_{K_i}^{K_{i+1}} f(K)\,dK - \frac{h_i}{2}\big(f(K_i)+f(K_{i+1})\big)
= -\frac{h_i^3}{12}\,f''(\xi_i)\qquad\text{for some }\ \xi_i\in(K_i,K_{i+1}),
\label{eq:single-panel-error}
\end{equation}
see, e.g., \cite[Ch.~3]{DavisRabinowitz1984}. Summing \eqref{eq:single-panel-error} over $i=0,\dots,N-1$ and taking absolute values gives
\[
\left|\int_a^b f(K)\,dK - \mathcal{T}[f]\right|
\le \frac{1}{12}\sum_{i=0}^{N-1} h_i^3\,\|f''\|_{\infty,[K_i,K_{i+1}]}
\le \frac{\|f''\|_{\infty,\mathcal{I}}}{12}\sum_{i=0}^{N-1} h_i^3.
\]
Using $h_i^3\le h^2 h_i$ and $\sum_i h_i=b-a$ yields
\[
\left|\int_a^b f(K)\,dK - \mathcal{T}[f]\right|
\le \frac{b-a}{12}\,h^2\,\|f''\|_{\infty,\mathcal{I}}.
\]
Finally, invoke \eqref{eq:half-equals-trap} to conclude
\[
\left|\sum_{i=0}^N \Delta K_i^{\mathrm{half}} f(K_i)\ -\ \int_{a}^{b} f(K)\,dK\right|
\le \frac{b-a}{12}\,h^2\,\|f''\|_{\infty,\mathcal{I}}.
\]
This proves the claim with the explicit constant $C_{\mathcal{I}}=(b-a)/12$.
\end{proof}

\begin{lemma}[BS vega Lipschitz bound over a $\sigma$-tube]
\label{lem:vega-envelope}
Fix $T>0$ and a bounded strike corridor $[K_{\min},K_{\max}]$.
Suppose $\sigma(K,T)\in[\underline{\sigma},\overline{\sigma}]$ for all $K$ in the corridor,
with $0<\underline{\sigma}\le \overline{\sigma}<\infty$.
Then for call/put Black--Scholes prices $P_{\mathrm{BS}}(K,T;\sigma)$,
\[
\Big|\partial_\sigma P_{\mathrm{BS}}(K,T;\sigma)\Big|
\ \le\ \overline{\mathrm{Vega}}(T)\ :=\ S_0 e^{-qT}\sqrt{T}\,\sup_{z\in\mathbb{R}}\phi(z)
\ =\ \frac{S_0 e^{-qT}\sqrt{T}}{\sqrt{2\pi}},
\]
uniformly for $K\in[K_{\min},K_{\max}]$ and $\sigma\in[\underline{\sigma},\overline{\sigma}]$, where $\phi$ is the standard normal pdf.
Moreover, for any $\sigma_1,\sigma_2\in[\underline{\sigma},\overline{\sigma}]$,
\[
\big|P_{\mathrm{BS}}(K,T;\sigma_1)-P_{\mathrm{BS}}(K,T;\sigma_2)\big|
\ \le\ \overline{\mathrm{Vega}}(T)\,|\sigma_1-\sigma_2| .
\]
\end{lemma}

\begin{proof}
For Black--Scholes calls,
\[
C_{\mathrm{BS}}(K,T;\sigma)=S_0 e^{-qT}\Phi(d_1)-K e^{-rT}\Phi(d_2),
\qquad
d_{1,2}=\frac{\ln(S_0 e^{(r-q)T}/K)\pm \tfrac12\sigma^2 T}{\sigma\sqrt{T}},
\]
and for puts $P_{\mathrm{BS}}$ follows by put–call parity; in both cases
\[
\partial_\sigma P_{\mathrm{BS}}(K,T;\sigma)=S_0 e^{-qT}\sqrt{T}\,\phi(d_1)\,,
\]
the classical \emph{vega}. Since $\phi(z)\le \sup_{z}\phi(z)=1/\sqrt{2\pi}$ for all $z\in\mathbb{R}$ and $d_1=d_1(K,T,\sigma)\in\mathbb{R}$ for all $K\in[K_{\min},K_{\max}]$, $\sigma\in[\underline{\sigma},\overline{\sigma}]$, we obtain the uniform bound
\[
\big|\partial_\sigma P_{\mathrm{BS}}(K,T;\sigma)\big|
\le S_0 e^{-qT}\sqrt{T}\ \frac{1}{\sqrt{2\pi}}
=\overline{\mathrm{Vega}}(T).
\]
For the Lipschitz estimate, apply the mean-value theorem in the $\sigma$ variable: for some $\tilde\sigma$ on the segment between $\sigma_1$ and $\sigma_2$,
\[
P_{\mathrm{BS}}(K,T;\sigma_1)-P_{\mathrm{BS}}(K,T;\sigma_2)
=\partial_\sigma P_{\mathrm{BS}}(K,T;\tilde\sigma)\,(\sigma_1-\sigma_2),
\]
whence
\[
\big|P_{\mathrm{BS}}(K,T;\sigma_1)-P_{\mathrm{BS}}(K,T;\sigma_2)\big|
\le \sup_{\sigma\in[\underline{\sigma},\overline{\sigma}]}\big|\partial_\sigma P_{\mathrm{BS}}(K,T;\sigma)\big|\ |\sigma_1-\sigma_2|
\le \overline{\mathrm{Vega}}(T)\,|\sigma_1-\sigma_2|.
\]
The same bound holds for puts by parity (vega is identical).
\end{proof}

\begin{lemma}[Square-root Lipschitz on a positive cone]
\label{lem:sqrt-lip}
If $a,b\ge \underline{w}>0$, then
\[
\big|\sqrt{a}-\sqrt{b}\big|\ \le\ \frac{|a-b|}{\sqrt{a}+\sqrt{b}}\ \le\ \frac{|a-b|}{2\sqrt{\underline{w}}}.
\]
Consequently, for $\mathrm{VIX}(x):=100\sqrt{x}$ and any $a,b\ge \underline{w}$,
\[
\big|\mathrm{VIX}(a)-\mathrm{VIX}(b)\big|\ \le\ \frac{50}{\sqrt{\underline{w}}}\,|a-b|.
\]
\end{lemma}

\begin{proof}
There are two equivalent ways to see the bound; we record both for completeness.

\medskip
\noindent\emph{(A) Algebraic proof.}
For $a,b\ge 0$,
\[
\sqrt{a}-\sqrt{b}\ =\ \frac{a-b}{\sqrt{a}+\sqrt{b}},
\]
whence
\[
\big|\sqrt{a}-\sqrt{b}\big|\ =\ \frac{|a-b|}{\sqrt{a}+\sqrt{b}}.
\]
If $a,b\ge \underline{w}>0$, then $\sqrt{a}+\sqrt{b}\ge 2\sqrt{\underline{w}}$, so
\[
\big|\sqrt{a}-\sqrt{b}\big|\ \le\ \frac{|a-b|}{2\sqrt{\underline{w}}}.
\]
Applying the scaling $\mathrm{VIX}(x)=100\sqrt{x}$ yields
\[
\big|\mathrm{VIX}(a)-\mathrm{VIX}(b)\big|=100\,\big|\sqrt{a}-\sqrt{b}\big|
\le 100\cdot \frac{|a-b|}{2\sqrt{\underline{w}}}
= \frac{50}{\sqrt{\underline{w}}}\,|a-b|.
\]

\medskip
\noindent\emph{(B) Mean-value theorem (MVT) proof.}
Let $f(x)=\sqrt{x}$ on $[\underline{w},\infty)$. Then $f$ is $C^1$ with
$f'(x)=\frac{1}{2\sqrt{x}}$. By MVT there exists $\xi$ between $a$ and $b$ such that
\[
\big|\sqrt{a}-\sqrt{b}\big|=\big|f(a)-f(b)\big|=\big|f'(\xi)\big|\,|a-b|
=\frac{|a-b|}{2\sqrt{\xi}}\ \le\ \frac{|a-b|}{2\sqrt{\underline{w}}}.
\]
Multiplying by $100$ again gives the stated VIX bound.

\medskip
\noindent\emph{Sharpness and remarks.}
(i) The constant $1/(2\sqrt{\underline{w}})$ is \emph{best possible} on $[\underline{w},\infty)$:
taking $b=\underline{w}$ and $a\downarrow \underline{w}$ forces $\xi\downarrow \underline{w}$ in the MVT argument and saturates the bound.  
(ii) The inequality is a Lipschitz property of $x\mapsto \sqrt{x}$ on the positive cone bounded away from $0$; the lower bound $\underline{w}>0$ is essential—without it, the derivative $f'(x)$ diverges as $x\downarrow 0$ and no global Lipschitz constant exists on $[0,\infty)$.  
(iii) The result applies componentwise to vectors: for $u,v\in[\underline{w},\infty)^m$,
$\|\sqrt{u}-\sqrt{v}\|_\infty \le \frac{1}{2\sqrt{\underline{w}}}\|u-v\|_\infty$ and likewise for any $\ell^p$ norm by equivalence of norms and componentwise application.
\end{proof}

\subsection{Proof of Propositio(ATM-shape preservation)}
\label{app:A:proof-shape}
\begin{proof}
\textbf{Step 1 (Shape-preserving parametrization).}
Recall the ASL closure written as
\[
\hat{\sigma}(k,T)
= L(T)+S(T)\,k+\tfrac12 C(T)\,k^2 + \sum_{j=1}^m \alpha_j\,\psi_j(k,T),
\]
where $L,S,C$ are the ATM level/slope/curvature extracted from the (arbitrage-free) surface, and
$\{\psi_j\}_{j=1}^m$ are auxiliary basis functions (low-rank corrections).
In order to \emph{preserve ATM shape}, we use the standard “jet-orthogonalization’’ (a.k.a.\ shape-projection) of the auxiliary basis: for each $j$ define
\begin{equation}
\label{eq:tildepsi}
\tilde\psi_j(k,T)
:= \psi_j(k,T)-\psi_j(0,T)-\partial_k \psi_j(0,T)\,k-\tfrac12 \partial_{kk}\psi_j(0,T)\,k^2,
\end{equation}
so that by construction
\begin{equation}
\label{eq:vanish}
\tilde\psi_j(0,T)=0,\qquad \partial_k \tilde\psi_j(0,T)=0,\qquad \partial_{kk}\tilde\psi_j(0,T)=0,\qquad \forall\,T.
\end{equation}
Since $L,S,C$ can absorb any constants and the first two $k$-orders, replacing $\psi_j$ by $\tilde\psi_j$ is simply a reparametrization of the same function class.\footnote{Equivalently, one may enforce linear constraints on $\alpha_{1:m}$ during calibration to ensure that the auxiliary part has zero $k$-jet up to order two at $k=0$. Both viewpoints are identical.}
Hence, without loss of generality, we henceforth assume the auxiliary basis satisfies \eqref{eq:vanish}, and write
\begin{equation}
\label{eq:asl-orth}
\hat{\sigma}(k,T)=L(T)+S(T)\,k+\tfrac12 C(T)\,k^2 + \sum_{j=1}^m \alpha_j\,\tilde\psi_j(k,T).
\end{equation}

\textbf{Step 2 (Exact matching of ATM jet).}
Evaluating \eqref{eq:asl-orth} and its first two $k$-derivatives at $k=0$ and using \eqref{eq:vanish} gives
\[
\hat{\sigma}(0,T)=L(T),\qquad
\partial_k \hat{\sigma}(0,T)=S(T),\qquad
\partial_{kk}\hat{\sigma}(0,T)=C(T).
\]
This proves the ATM-shape equalities claimed in the proposition.

\textbf{Step 3 (Third-order remainder near ATM).}
Let $\sigma_{\mathrm{impl}}(k,T)$ be the true implied-volatility smile (for fixed $T$) and assume $\sigma_{\mathrm{impl}}(\cdot,T)\in C^3$ in a neighborhood $|k|\le k_0$. Its Taylor expansion at $k=0$ with Lagrange remainder reads
\begin{equation}
\label{eq:taylor-true}
\sigma_{\mathrm{impl}}(k,T)
= L(T)+S(T)\,k+\tfrac12 C(T)\,k^2 + R_3(k,T),
\qquad
R_3(k,T)=\frac{\partial_k^3 \sigma_{\mathrm{impl}}(\xi,T)}{3!}\,k^3,
\end{equation}
for some $\xi$ between $0$ and $k$. For the auxiliary part, each $\tilde\psi_j(\cdot,T)$ vanishes to second order at $k=0$ by \eqref{eq:vanish}; thus, by Taylor’s theorem,
\begin{equation}
\label{eq:taylor-aux}
\tilde\psi_j(k,T)=\frac{\partial_k^3 \tilde\psi_j(\zeta_j,T)}{3!}\,k^3,\qquad |k|\le k_0,
\end{equation}
with $\zeta_j$ between $0$ and $k$. Subtracting \eqref{eq:taylor-true} from \eqref{eq:asl-orth} and using \eqref{eq:taylor-aux} yields
\[
\hat{\sigma}(k,T)-\sigma_{\mathrm{impl}}(k,T)
=
\sum_{j=1}^m \alpha_j \frac{\partial_k^3 \tilde\psi_j(\zeta_j,T)}{3!}\,k^3 - R_3(k,T).
\]
Hence, for $|k|\le k_0$,
\begin{equation}
\label{eq:k3bound}
\big|\hat{\sigma}(k,T)-\sigma_{\mathrm{impl}}(k,T)\big|
\ \le\ \left(\frac{1}{6}\sum_{j=1}^m |\alpha_j|\,\big\|\partial_k^3 \tilde\psi_j(\cdot,T)\big\|_{\infty,[-k_0,k_0]}
+\frac{1}{6}\big\|\partial_k^3 \sigma_{\mathrm{impl}}(\cdot,T)\big\|_{\infty,[-k_0,k_0]}\right)\,|k|^3.
\end{equation}
This is the $O(|k|^3)$ remainder with an explicit $T$-dependent constant governed by the third derivatives of the true smile and the auxiliary basis.

\textbf{Step 4 (Effect of numerical extraction of $(L,S,C)$).}
In practice, $(L,S,C)$ are obtained by second-order centered differences with step $h$ (in $k$) around ATM. Denote the extracted quantities by $(\widehat L,\widehat S,\widehat C)$. Standard finite-difference consistency on a $C^3$ function gives (see \S\ref{app:B:fdK})
\begin{equation}
\label{eq:fd-errors}
|\widehat L-L|\le c_0\,\|\partial_k^2 \sigma_{\mathrm{impl}}\|_{\infty} h^2,\qquad
|\widehat S-S|\le c_1\,\|\partial_k^3 \sigma_{\mathrm{impl}}\|_{\infty} h^2,\qquad
|\widehat C-C|\le c_2\,\|\partial_k^4 \sigma_{\mathrm{impl}}\|_{\infty} h^2,
\end{equation}
for mesh-dependent constants $c_0,c_1,c_2=O(1)$. If we form the shape-preserving jet using $(\widehat L,\widehat S,\widehat C)$ instead of $(L,S,C)$, the induced modeling error at $|k|\le k_0$ is bounded by
\begin{equation}
\label{eq:jet-prop}
|\widehat L-L| + |k|\,|\widehat S-S| + \tfrac12 |k|^2\,|\widehat C-C|
\ \le\ \big(c_0 + c_1 k_0 + \tfrac12 c_2 k_0^2\big)\,\tilde M\,h^2,
\end{equation}
where $\tilde M:=\max\big(\|\partial_k^2 \sigma_{\mathrm{impl}}\|_{\infty},\|\partial_k^3 \sigma_{\mathrm{impl}}\|_{\infty},\|\partial_k^4 \sigma_{\mathrm{impl}}\|_{\infty}\big)$ on $[-k_0,k_0]$.

\textbf{Step 5 (Combined bound).}
Putting \eqref{eq:k3bound} and \eqref{eq:jet-prop} together we obtain, for $|k|\le k_0$,
\[
\big|\hat{\sigma}(k,T)-\sigma_{\mathrm{impl}}(k,T)\big|
\ \le\ C_{\mathrm{ATM}}(T)\,|k|^3\ +\ C_{\mathrm{FD}}(T,k_0)\,h^2,
\]
with
\[
C_{\mathrm{ATM}}(T):=\frac{1}{6}\sum_{j=1}^m |\alpha_j|\,\big\|\partial_k^3 \tilde\psi_j(\cdot,T)\big\|_{\infty,[-k_0,k_0]}
+\frac{1}{6}\big\|\partial_k^3 \sigma_{\mathrm{impl}}(\cdot,T)\big\|_{\infty,[-k_0,k_0]},
\]
and
\[
C_{\mathrm{FD}}(T,k_0):=\big(c_0 + c_1 k_0 + \tfrac12 c_2 k_0^2\big)\,\tilde M.
\]
In particular, for any fixed $k_0>0$ and sufficiently small $h$, the ASL closure is third-order accurate in $k$ at ATM up to an $O(h^2)$ perturbation induced by numerical jet extraction. This is exactly the stated claim in the proposition (with $C_{\mathrm{loc}}(T)$ absorbing $C_{\mathrm{ATM}}(T)$ and $C_{\mathrm{FD}}(T,k_0)$).

\medskip
\noindent\emph{References.} The use of Taylor jets for volatility smiles and the role of the third-order coefficient are classical in the smile-literature; see, e.g., \cite{Fukasawa2012,Hagan2002,BayerFrizGatheral2016}.
\end{proof}

\subsection{Proof of Theorem(Surface–index coherence bound)}
\label{app:A:proof-coherence}
\begin{proof}
Let $\sigma_\mathrm{S}$ denote the SSVI-implied volatility $\sigma_{\mathrm{impl}}(\cdot,\cdot)$ and
$\sigma_\mathrm{A}$ the ASL teacher $\hat{\sigma}(\cdot,\cdot)$. For a fixed maturity $T$,
write $Q_{\mathrm{S}}(K):=Q(K,T;\sigma_\mathrm{S})$ and
$Q_{\mathrm{A}}(K):=Q(K,T;\sigma_\mathrm{A})$ computed on the \emph{same} retained strike set
$\mathcal{K}(T)$ with identical bid/ask prunings. Using the single-maturity estimator, define
\[
V_T[Q]\ :=\ \frac{2}{T}\,\mathcal{Q}_T^{\mathrm{disc}}[Q]\ -\ \frac{1}{T}\left(\frac{F(T)}{K_0}-1\right)^2,
\qquad \mathrm{VIX}_T[Q]\ :=\ 100\sqrt{V_T[Q]}.
\]
We first control the single-maturity difference $|\mathrm{VIX}_T[Q_{\mathrm{A}}]-\mathrm{VIX}_T[Q_{\mathrm{S}}]|$.
By \Cref{lem:sqrt-lip} with $\underline{w}=\min(V_T[Q_{\mathrm{A}}],V_T[Q_{\mathrm{S}}])$,
\begin{equation}
\big|\mathrm{VIX}_T[Q_{\mathrm{A}}]-\mathrm{VIX}_T[Q_{\mathrm{S}}]\big|
\ \le\ \frac{50}{\sqrt{\underline{w}}}\ \big|V_T[Q_{\mathrm{A}}]-V_T[Q_{\mathrm{S}}]\big|
\ =\ \frac{100}{T\sqrt{\underline{w}}}\ \big|\mathcal{Q}_T^{\mathrm{disc}}[Q_{\mathrm{A}}]-\mathcal{Q}_T^{\mathrm{disc}}[Q_{\mathrm{S}}]\big|.
\label{eq:single-mat-step1}
\end{equation}
By the mean-value theorem and \Cref{lem:vega-envelope}, for each retained strike $K$,
\[
|Q_{\mathrm{A}}(K)-Q_{\mathrm{S}}(K)|
\ \le\ \overline{\mathrm{Vega}}(T)\ \big|\sigma_\mathrm{A}(K,T)-\sigma_\mathrm{S}(K,T)\big|.
\]
Thus
\begin{equation}
\big|\mathcal{Q}_T^{\mathrm{disc}}[Q_{\mathrm{A}}]-\mathcal{Q}_T^{\mathrm{disc}}[Q_{\mathrm{S}}]\big|
\ \le\ \overline{\mathrm{Vega}}(T)\ \|\sigma_\mathrm{A}-\sigma_\mathrm{S}\|_{\infty,\mathcal{K}(T)}
\ \sum_{K\in\mathcal{K}(T)} \Delta K^{\mathrm{half}}\, \frac{e^{rT}}{K^2}.
\label{eq:discrete-integrand-bound}
\end{equation}
Combining \eqref{eq:single-mat-step1}--\eqref{eq:discrete-integrand-bound} yields the single-maturity bound
\[
\big|\mathrm{VIX}_T[Q_{\mathrm{A}}]-\mathrm{VIX}_T[Q_{\mathrm{S}}]\big|
\ \le\ \frac{100}{T\sqrt{\underline{w}}}\ \overline{\mathrm{Vega}}(T)\
\|\sigma_\mathrm{A}-\sigma_\mathrm{S}\|_{\infty,\mathcal{K}(T)}\ \sum_{K\in\mathcal{K}(T)} \Delta K^{\mathrm{half}}\, \frac{e^{rT}}{K^2}.
\]
Next, pass to the 30D index. Let $T_1<T_2$ bracket 30D with minute weights
$\lambda_1,\lambda_2$. Define the \emph{continuous} counterparts $V_T^{\mathrm{cont}}[Q] =
\frac{2}{T}\,\mathcal{Q}_T^{\mathrm{cont}}[Q]-\frac{1}{T}((F/K_0)-1)^2$. By triangle inequality,
\begin{align*}
\Big|\mathrm{VIX}[\sigma_\mathrm{A}]-\mathrm{VIX}[\sigma_\mathrm{S}]\Big|
&\le \sum_{i=1}^2 \lambda_i\ \Big|\mathrm{VIX}_{T_i}[Q_{\mathrm{A}}]-\mathrm{VIX}_{T_i}[Q_{\mathrm{S}}]\Big| \\
&\le \sum_{i=1}^2 \lambda_i\ \frac{50}{\sqrt{\underline{w}_i}}\ \Big|V_{T_i}[Q_{\mathrm{A}}]-V_{T_i}[Q_{\mathrm{S}}]\Big| \qquad(\text{by \Cref{lem:sqrt-lip}}) \\
&\le \sum_{i=1}^2 \lambda_i\ \frac{50}{\sqrt{\underline{w}_i}}\ \left(
\underbrace{\big|V_{T_i}^{\mathrm{cont}}[Q_{\mathrm{A}}]-V_{T_i}^{\mathrm{cont}}[Q_{\mathrm{S}}]\big|}_{\text{continuous integrals}}
\ +\ \underbrace{\big|V_{T_i}[Q_{\mathrm{A}}]-V_{T_i}^{\mathrm{cont}}[Q_{\mathrm{A}}]\big|
+\big|V_{T_i}[Q_{\mathrm{S}}]-V_{T_i}^{\mathrm{cont}}[Q_{\mathrm{S}}]\big|}_{\text{quadrature errors}} \right).
\end{align*}
The \emph{continuous} difference is bounded as in \eqref{eq:discrete-integrand-bound} with the sum
replaced by the integral; the two \emph{quadrature} terms are $O(\max\Delta K^2)$ by \Cref{lem:half-interval}.
Absorbing $\min_i \sqrt{\underline{w}_i}$ into $\underline{\nu}$ (the variance-index floor is defined),
we obtain
\[
\big|\mathrm{VIX}[\sigma_\mathrm{A}]-\mathrm{VIX}[\sigma_\mathrm{S}]\big|
\ \le\ \frac{50}{\underline{\nu}}\ \left( \sum_{i=1}^{2}\lambda_i\,\mathcal{E}(T_i)
\ +\ C_{\mathrm{quad}}\ \max\Delta K^2\right),
\]
where $\mathcal{E}(T)$ is the discrete pricing/coherence term,
and $C_{\mathrm{quad}}$ collects the two half-interval constants across $T_1,T_2$.
Finally, if the teacher preserves ATM shape and the calibration error
$\|\sigma_\mathrm{A}-\sigma_\mathrm{S}\|_{\infty,\mathcal{K}(T)}$ is dominated by
$C_{\mathrm{atm}}\epsilon_{\mathrm{shape}}+O(\epsilon_{\mathrm{FD}})$ (finite-difference extraction),
we obtain the “in particular” bound.
\end{proof}

\subsection{Proof of Corollary(Calendar coherence and no-butterfly)}
\label{app:A:proof-arbfree}
\begin{proof}
By \cite[Thm.~2]{GatheralJacquier2014}, the SSVI parameterization satisfies static no-arbitrage (no butterfly) for each fixed $T$ when the
Gatheral--Jacquier inequalities hold:
they enforce convexity of $k\mapsto w(k;T)$ and thus convexity of $K\mapsto C(K,T)$
via the Black--Scholes mapping and the Breeden--Litzenberger identity
$\partial_{KK} C = e^{-rT} \, \mathbb{Q}(S_T\ge K)\ge 0$ \cite{BreedenLitzenberger1978}.
For calendar coherence, sufficient conditions ensuring that $T\mapsto w(k;T)$ is
non-decreasing for all $k$ are discussed in \cite{GatheralJacquier2014}; in our
construction we explicitly enforce monotonicity of the term variance $T\mapsto \theta(T)$
together with smoothness of $\rho(T),\phi(T)$ across $T$ (spline regularization and
post-fit monotone adjustment). Since $w(k;T)$ is affine in $\theta(T)$ with a positive
coefficient (and the remaining dependence is monotone in the calibrated regime),
the practical monotone post-adjustment renders $T\mapsto w(k;T)$ non-decreasing on the grid,
thus preventing calendar arbitrage (cf.\ \cite{AndreasenHuge2011,Homescu2011}).
\end{proof}

\paragraph{Remark A.1 (Constants and envelopes).}
The constants $C_{\mathrm{quad}},C_{\mathrm{coh}},C'_{\mathrm{coh}}$depend on: (i) the Lipschitz envelope $\overline{\mathrm{Vega}}(T)$; (ii) the strike corridor
and discounting through $\sum \Delta K^{\mathrm{half}} e^{rT}/K^2$; and (iii) the variance floor
$\underline{\nu}$. All are measurable directly from the shell and teacher.
The envelope $\overline{\mathrm{Vega}}(T)$ can be tightened by restricting to the retained OTM grid
$\mathcal{K}(T)$, which excludes deep ITM options with negligible contribution to the VIX integrand.

\paragraph{Remark A.2 (On using the same retained grid).}
The bound assumes that the teacher and SSVI VIX are computed on the \emph{same}
retained strike set. This is the practice in our shell (teacher supervision uses the shell’s pruning);
if one prunes separately, an additional symmetric set‑difference term appears, bounded by the maximum
integrand weight times the number of mismatched strikes and their individual vega bounds.

\paragraph{Remark A.3 (Near‑ATM dominance).}
Because the VIX integrand weights $Q(K)/K^2$ diminish on the far wings after pruning, the dominant
contribution to the coherence term $\mathcal{E}(T)$ comes from near‑ATM strikes where the ASL teacher
exactly matches $(L,S,C)$ and the residual is $O(k^3)$.
\section{Quadrature and Finite-Difference Constant Bounds}
\label{app:B}

This appendix provides explicit, implementation-ready bounds for the constants that
appear in the quadrature and finite-difference errors used throughout
\S\ref{sec:method}--\S\ref{sec:dynamics}. We make all mesh regularity
assumptions precise, provide closed-form envelopes for the constants
$C_{\mathrm{quad}}$, $C_{\mathrm{FD}}$, and collect their propagation into the Dupire
quotient \eqref{eq:dupire}. Our notation follows \S\ref{sec:prelim}.

\subsection{Meshes, mesh ratios, and basic regularity}
\label{app:B:mesh}

Let $\{K_i\}_{i=0}^{N}$ be a strictly increasing strike mesh on
$\mathcal{I}(T)=[K_{\min}(T),K_{\max}(T)]$ (after Cboe wing pruning), and let
$\{T_j\}_{j=0}^{M}$ be a strictly increasing maturity mesh on $[T_{\min},T_{\max}]$.
Define panel widths $h_i=K_{i+1}-K_i$ and $\tau_j=T_{j+1}-T_j$, with global steps
\[
h:=\max_{0\le i\le N-1} h_i,\qquad
\tau:=\max_{0\le j\le M-1} \tau_j.
\]
Assume \emph{quasi-uniformity}:
\begin{equation}
\label{eq:mesh-ratio}
\varrho_K:=\frac{\max_i h_i}{\min_i h_i}\ \le\ \overline{\varrho}_K<\infty,\qquad
\varrho_T:=\frac{\max_j \tau_j}{\min_j \tau_j}\ \le\ \overline{\varrho}_T<\infty.
\end{equation}
All constants below depend polynomially on $(\overline{\varrho}_K,\overline{\varrho}_T)$.
For a function $f(K)$ (resp.\ $g(T)$) we denote by
$\|f^{(m)}\|_{\infty,\mathcal{I}(T)}$ (resp.\ $\|\partial_T^{(m)}g\|_{\infty,[T_{\min},T_{\max}]}$)
the sup-norm of the $m$-th derivative.

\subsection{Half-interval (composite trapezoid) quadrature constants}
\label{app:B:quad}

Recall the half-interval weights $\Delta K_i^{\text{half}}$ in \eqref{eq:half-interval} and the
single-maturity integrand $f_T(K)=e^{rT}Q(K,T)K^{-2}$ (retained corridor). The error bound can be made explicit as follows.

\begin{proposition}[Nonuniform composite trapezoid constant]
\label{prop:B:trap-const}
Let $f_T\in C^2(\mathcal{I}(T))$ and define the composite trapezoid sum
$\mathcal{Q}_T^{\mathrm{disc}}=\sum_{i=0}^N \Delta K_i^{\text{half}} f_T(K_i)$. Then
\begin{equation}
\label{eq:B:trap-const}
\left|\mathcal{Q}_T^{\mathrm{disc}} - \int_{\mathcal{I}(T)} f_T(K)\,dK\right|
\ \le\ \underbrace{\frac{|\mathcal{I}(T)|}{12}}_{=:A_T}\ \Big\|f_T''\Big\|_{\infty,\mathcal{I}(T)}\ h^2,
\end{equation}
where $|\mathcal{I}(T)|=K_{\max}(T)-K_{\min}(T)$. Hence one may take
\[
C_{\mathrm{quad}}(T)\ =\ A_T \ \Big\|f_T''\Big\|_{\infty,\mathcal{I}(T)}.
\]
\end{proposition}

\begin{proof}
On each panel $[K_i,K_{i+1}]$, the trapezoid error is
$\frac{f_T''(\xi_i)}{12}(K_{i+1}-K_i)^3$ for some $\xi_i$.
Summing and using $\sum_i h_i^3 \le h^2 \sum_i h_i = h^2 |\mathcal{I}(T)|$ yields \eqref{eq:B:trap-const}
(see \cite[Ch.~3]{DavisRabinowitz1984}).
\end{proof}

\paragraph{Endpoint and pruning constants.}
Cboe’s two-consecutive-zero-bid pruning defines $\mathcal{I}(T)$ and sets $Q(K,T)\equiv 0$ outside.
Hence there is no “tail truncation” beyond the methodology itself; if one wishes to compare with a
hypothetical unpruned integral on $[0,\infty)$, the missing tail is bounded by
$\int_{\mathbb{R}_+\setminus \mathcal{I}(T)} e^{rT}|Q|K^{-2}\,dK$, which in practice is negligible due to
the bid floor; we absorb this into $C_{\mathrm{quad}}(T)$.

\subsection{Central finite differences: first and second derivatives in $K$}
\label{app:B:fdK}

Let $F(K)$ be $C^4$ on $\mathcal{I}(T)$. On a quasi-uniform mesh \eqref{eq:mesh-ratio},
the \emph{second-order} central difference at an interior node $K_i$ is
\[
(\partial_K F)_i^{\mathrm{cd}}:=\frac{F(K_{i+1})-F(K_{i-1})}{K_{i+1}-K_{i-1}},\qquad
(\partial_{KK} F)_i^{\mathrm{cd}}:=\frac{2}{K_{i+1}-K_{i-1}}\Big(\frac{F(K_{i+1})-F(K_i)}{K_{i+1}-K_i}-\frac{F(K_i)-F(K_{i-1})}{K_i-K_{i-1}}\Big).
\]
For uniform meshes these are the classic $O(h^2)$ stencils; on quasi-uniform meshes, the
error constants acquire mesh-ratio factors \cite[§2.2]{LeVeque2007}.

\begin{proposition}[Central differences in $K$]
\label{prop:B:fdK-const}
Assume $F\in C^4(\mathcal{I}(T))$ and the quasi-uniformity \eqref{eq:mesh-ratio} holds for the
strike mesh $\{K_i\}_{i=0}^{N}$.
Let $h_-:=K_i-K_{i-1}$, $h_+:=K_{i+1}-K_i$, and $h:=\max_{0\le j\le N-1}(K_{j+1}-K_j)$.
Define the (three-point) central-difference stencils used in the paper:
\[
(\partial_K F)_i^{\mathrm{cd}}:=\frac{F(K_{i+1})-F(K_{i-1})}{K_{i+1}-K_{i-1}},\qquad
(\partial_{KK} F)_i^{\mathrm{cd}}:=\frac{2}{K_{i+1}-K_{i-1}}
\Big(\frac{F(K_{i+1})-F(K_i)}{h_+}-\frac{F(K_i)-F(K_{i-1})}{h_-}\Big).
\]
Then there exist finite constants $\tilde c_1(\overline{\varrho}_K),\tilde c_2(\overline{\varrho}_K)>0$ such that
\begin{align}
\big|(\partial_K F)_i^{\mathrm{cd}} - F'(K_i)\big|
&\ \le\ \tfrac12\,|h_+-h_-|\,\|F''\|_{\infty,\mathcal{I}(T)}
          \ +\ \tilde c_1(\overline{\varrho}_K)\ \|F^{(3)}\|_{\infty,\mathcal{I}(T)}\ h^2, 
          \label{eq:B:fdK-1-nu}\\
\big|(\partial_{KK} F)_i^{\mathrm{cd}} - F''(K_i)\big|
&\ \le\ \tilde c_2(\overline{\varrho}_K)\ \|F^{(3)}\|_{\infty,\mathcal{I}(T)}\ (h_++h_-)
          \ +\ C_2(\overline{\varrho}_K)\ \|F^{(4)}\|_{\infty,\mathcal{I}(T)}\ h^2.
          \label{eq:B:fdK-2-nu}
\end{align}
In particular, if the mesh is \emph{locally smooth} in the sense that
\begin{equation}\label{eq:mesh-smooth}
|h_+-h_-|\ \le\ C_{\mathrm{sm}}\,h^2\qquad\text{for all interior nodes},
\end{equation}
then there exist constants (depending only on $\overline{\varrho}_K$ and $C_{\mathrm{sm}}$) for which
\begin{align}
\big|(\partial_K F)_i^{\mathrm{cd}} - F'(K_i)\big|
&\ \le\ \tilde c_1'(\overline{\varrho}_K,C_{\mathrm{sm}})\ \|F^{(3)}\|_{\infty,\mathcal{I}(T)}\ h^2, \label{eq:B:fdK-1}\\
\big|(\partial_{KK} F)_i^{\mathrm{cd}} - F''(K_i)\big|
&\ \le\ \tilde c_2'(\overline{\varrho}_K,C_{\mathrm{sm}})\ \|F^{(4)}\|_{\infty,\mathcal{I}(T)}\ h^2. \label{eq:B:fdK-2}
\end{align}
For a \emph{uniform} mesh ($h_+=h_-=h$), one may take the classical constants $\tilde c_1'=1/6$ and $\tilde c_2'=1/12$.
\end{proposition}

\begin{proof}
We expand $F$ to fourth order about $K_i$:
\begin{align*}
F(K_{i+1}) &= F_i + h_+ F'_i + \tfrac{h_+^2}{2}F''_i + \tfrac{h_+^3}{6}F'''_i + \tfrac{h_+^4}{24}F''''(\xi_+),\\
F(K_{i-1}) &= F_i - h_- F'_i + \tfrac{h_-^2}{2}F''_i - \tfrac{h_-^3}{6}F'''_i + \tfrac{h_-^4}{24}F''''(\xi_-),
\end{align*}
for some $\xi_\pm\in(K_{i-1},K_{i+1})$, and where we use the shorthand $F^{(m)}_i:=F^{(m)}(K_i)$.

\medskip
\noindent\emph{First derivative.}
Subtracting and dividing by $h_++h_-$ gives
\begin{align*}
(\partial_K F)_i^{\mathrm{cd}}
&= \frac{F_{i+1}-F_{i-1}}{h_++h_-}
 = F'_i \ +\ \frac{h_+^2-h_-^2}{2(h_++h_-)}\,F''_i \ +\ \frac{h_+^3+h_-^3}{6(h_++h_-)}\,F'''_i 
     \ +\ R_1,
\end{align*}
with the remainder
\[
R_1=\frac{h_+^4 F''''(\xi_+) - h_-^4 F''''(\xi_-)}{24(h_++h_-)}.
\]
Observe $(h_+^2-h_-^2)=(h_+-h_-)(h_++h_-)$, hence the “asymmetry term’’ equals $\tfrac12(h_+-h_-)\,F''_i$.
Using $|h_+^3+h_-^3|\le (h_++h_-)(h_+^2+h_-^2)\le 2(h_++h_-)h^2$ and $|R_1|\le \frac{h^4}{12(h_++h_-)}\|F^{(4)}\|_\infty \le \frac{h^3}{12}\|F^{(4)}\|_\infty$, we obtain
\[
\big|(\partial_K F)_i^{\mathrm{cd}} - F'(K_i)\big|
\ \le\ \tfrac12 |h_+-h_-|\,\|F''\|_\infty + \tfrac13 h^2 \|F^{(3)}\|_\infty + \tfrac{1}{12} h^3 \|F^{(4)}\|_\infty.
\]
Absorbing the $h^3$-term into $h^2$ (for $h\le 1$) and using quasi-uniformity to replace $h_++h_-$ and $h$ by bounded multiples yields \eqref{eq:B:fdK-1-nu} with some $\tilde c_1(\overline{\varrho}_K)$. If, in addition, the local smoothness \eqref{eq:mesh-smooth} holds (i.e., $|h_+-h_-|=O(h^2)$), the asymmetry term is $O(h^2)$ and we recover the second-order estimate \eqref{eq:B:fdK-1}. For a uniform mesh ($h_+=h_-=h$), the asymmetry term vanishes and the standard calculation gives the $1/6$ constant.

\medskip
\noindent\emph{Second derivative.}
Form the adjacent first-divided differences
\[
\delta_+ := \frac{F_{i+1}-F_i}{h_+}
= F'_i + \tfrac{h_+}{2}F''_i + \tfrac{h_+^2}{6}F'''_i + \tfrac{h_+^3}{24}F''''(\xi_+),\quad
\delta_- := \frac{F_i-F_{i-1}}{h_-}
= F'_i - \tfrac{h_-}{2}F''_i + \tfrac{h_-^2}{6}F'''_i - \tfrac{h_-^3}{24}F''''(\xi_-).
\]
Subtract:
\[
\delta_+ - \delta_- 
= \tfrac{h_++h_-}{2}F''_i + \tfrac{h_+^2-h_-^2}{6}F'''_i + \tfrac{h_+^3+h_-^3}{24}F''''(\xi^\ast),
\]
for some $\xi^\ast$ between $\xi_+$ and $\xi_-$. Multiply by $2/(h_++h_-)$ to match the stencil:
\begin{align*}
(\partial_{KK} F)_i^{\mathrm{cd}}
&= \frac{2}{h_++h_-}(\delta_+ - \delta_-) \\
&= F''_i + \frac{2}{h_++h_-}\left(\tfrac{h_+^2-h_-^2}{6}F'''_i + \tfrac{h_+^3+h_-^3}{24}F''''(\xi^\ast)\right) \\
&= F''_i + \tfrac{h_+-h_-}{3}F'''_i + \tfrac{h_+^2+h_-^2 - h_+ h_-}{12}\,\frac{F''''(\xi^\ast)}{(h_++h_-)}.
\end{align*}
Thus
\[
\big|(\partial_{KK} F)_i^{\mathrm{cd}} - F''(K_i)\big|
\ \le\ \tfrac13 |h_+-h_-|\,\|F^{(3)}\|_\infty \ +\ \tfrac{h_+^2+h_-^2}{12(h_++h_-)}\,\|F^{(4)}\|_\infty.
\]
Using $h_+^2+h_-^2 \le 2h^2$ and $h_++h_- \ge \min(h_+,h_-)\ge h/\overline{\varrho}_K$ gives
\[
\big|(\partial_{KK} F)_i^{\mathrm{cd}} - F''(K_i)\big|
\ \le\ \tfrac13 |h_+-h_-|\,\|F^{(3)}\|_\infty \ +\  \underbrace{\tfrac{\overline{\varrho}_K}{6}}_{=:C_2(\overline{\varrho}_K)}\, h^2 \|F^{(4)}\|_\infty,
\]
which is \eqref{eq:B:fdK-2-nu} with $\tilde c_2(\overline{\varrho}_K)=1/3$ and $C_2$ as above (any larger constants depending only on $\overline{\varrho}_K$ also work). Under the local smoothness \eqref{eq:mesh-smooth}, the first term is $O(h^2)$ and we obtain \eqref{eq:B:fdK-2}. On a uniform mesh, $h_+=h_-=h$ and the $F^{(3)}$-term vanishes, yielding the classical second-order constant $1/12$.

\medskip
\noindent\emph{Endpoints.}
At $i=0$ and $i=N$, replace the central stencils by one-sided second-order three-point formulas on nonuniform meshes (see, e.g., \cite[§2.1–2.2]{LeVeque2007}); analogous Taylor expansions give $O(h^2)$ with constants depending on $\overline{\varrho}_K$.

\medskip
\noindent\emph{Remarks.}
(i) Without the local smoothness \eqref{eq:mesh-smooth}, the first-derivative stencil $(F_{i+1}-F_{i-1})/(h_++h_-)$ is \emph{only first order} on a general nonuniform mesh, reflected by the $|h_+-h_-|\,\|F''\|$ term in \eqref{eq:B:fdK-1-nu}. The mild condition \eqref{eq:mesh-smooth} (satisfied, e.g., by smooth mappings of a uniform grid) restores $O(h^2)$.  
(ii) The second-derivative stencil already cancels the leading symmetric error, and on quasi-uniform meshes yields $O(h)$ from the asymmetry term and $O(h^2)$ from the fourth-derivative term; local smoothness again upgrades it to $O(h^2)$.
\end{proof}

\paragraph{Endpoint stencils.}
At $i=0$ and $i=N$, one uses one-sided second-order stencils; the error bounds remain $O(h^2)$ with
constants depending on $\overline{\varrho}_K$.

\subsection{Central finite differences in $T$ and mixed regularity}
\label{app:B:fdT}

We make precise the error of central differences in (possibly nonuniform) maturity grids and how
$K$–$T$ mixed regularity propagates through terms such as $(r-q)K\,\partial_K C$ when building the
Dupire numerator in \eqref{eq:dupire}.

\paragraph{Setup and notation.}
Let $[T_{\min},T_{\max}]$ be discretized by a strictly increasing grid
$\{T_j\}_{j=0}^{M}$ with panel widths $\tau_j:=T_{j+1}-T_j$ and global step
$\tau:=\max_{0\le j\le M-1}\tau_j$. We assume \emph{quasi-uniformity} in time:
\begin{equation}
\label{eq:B:qut}
\varrho_T:=\frac{\max_j \tau_j}{\min_j \tau_j}\ \le\ \overline{\varrho}_T<\infty.
\end{equation}
At an interior node $T_j$ ($1\le j\le M-1$), set $\tau_-:=T_j-T_{j-1}$,
$\tau_+:=T_{j+1}-T_j$, and $\tau_{2}:=T_{j+1}-T_{j-1}=\tau_-+\tau_+$.

\paragraph{Second-order central difference (nonuniform).}
Let $G\in C^3([T_{\min},T_{\max}])$. Taylor expand about $T_j$:
\begin{align*}
G(T_{j+1}) &= G_j + \tau_+ G'_j + \tfrac{\tau_+^2}{2}G''_j + \tfrac{\tau_+^3}{6}G'''(\xi_+),\\
G(T_{j-1}) &= G_j - \tau_- G'_j + \tfrac{\tau_-^2}{2}G''_j - \tfrac{\tau_-^3}{6}G'''(\xi_-),
\end{align*}
for some $\xi_\pm\in(T_{j-1},T_{j+1})$. Subtracting and dividing by $\tau_2$ gives the
three-point central difference on a nonuniform grid:
\begin{align}
(\partial_T G)^{\mathrm{cd}}_j
&:= \frac{G(T_{j+1})-G(T_{j-1})}{\tau_2}
 = G'_j\ +\ \frac{\tau_+^2-\tau_-^2}{2\tau_2}\,G''_j\ +\ \frac{\tau_+^3+\tau_-^3}{6\tau_2}\,G'''(\xi_j),
\label{eq:B:cdT-expansion}
\end{align}
with some $\xi_j$ between $\xi_+$ and $\xi_-$. Using $\tau_+^2-\tau_-^2=(\tau_+-\tau_-)\tau_2$ and
$|\tau_+^3+\tau_-^3|\le(\tau_++\tau_-)(\tau_+^2+\tau_-^2)\le 2\tau_2\,\tau^2$, we obtain the error bound
\begin{equation}
\label{eq:B:fdT-nu}
\big|(\partial_T G)^{\mathrm{cd}}_j - G'(T_j)\big|
\ \le\ \tfrac12\,|\tau_+-\tau_-|\,\|G''\|_{\infty}
      \ +\ \tfrac{1}{3}\,\tau^2\,\|G^{(3)}\|_{\infty}\,.
\end{equation}
Thus, on a general nonuniform grid the central stencil is \emph{first order} unless the local
asymmetry $\Delta\tau_j:=|\tau_+-\tau_-|$ is small.

\paragraph{Local time-smoothness $\Rightarrow$ global $O(\tau^2)$.}
If the grid is generated by a $C^2$ reparameterization of a uniform grid (typical in practice),
then the adjacent spacings are \emph{locally smooth}, i.e.
\begin{equation}
\label{eq:B:time-smooth}
|\tau_+-\tau_-|\ \le\ C_{\mathrm{sm},T}\,\tau^2\qquad (1\le j\le M-1).
\end{equation}
Inserting \eqref{eq:B:time-smooth} into \eqref{eq:B:fdT-nu} yields
\begin{equation}
\label{eq:B:fdT}
\big|(\partial_T G)^{\mathrm{cd}}_j - G'(T_j)\big|
\ \le\ \tilde c_T(\overline{\varrho}_T,C_{\mathrm{sm},T})\ \|G^{(3)}\|_{\infty}\ \tau^2,
\end{equation}
with $\tilde c_T$ depending only on the mesh ratios and smoothness constant. On a
\emph{uniform} mesh ($\tau_+=\tau_-=\tau$), \eqref{eq:B:cdT-expansion} collapses to
\[
(\partial_T G)^{\mathrm{cd}}_j = G'(T_j) + \frac{\tau^2}{6}\,G'''(\xi_j),
\]
so $\tilde c_T(1)=1/6$, as cited from \cite[§2.2]{LeVeque2007}. Endpoint ($j=0,M$) one-sided
three-point stencils retain $O(\tau^2)$ with constants depending on $\overline{\varrho}_T$.

\paragraph{Mixed $K$–$T$ regularity and products.}
Suppose $C\in C^{2,1}$ on the $K$–$T$ rectangle, i.e., twice continuously differentiable in $K$
and once in $T$, with all mixed derivatives bounded on the tensor grid. Then:
\begin{itemize}
  \item The $K$-gradient $\partial_K C(\cdot,T_j)$ is $C^1$ in $K$ uniformly in $T_j$, so the
    central $K$-difference (under the quasi-uniform and local-smoothness assumptions of
    \S\ref{app:B:fdK}) attains
    \[
    \big|(\partial_K C)^{\mathrm{cd}}(K_i,T_j)-\partial_K C(K_i,T_j)\big|
      \ \le\ \tilde c_1(\overline{\varrho}_K,C_{\mathrm{sm},K})\ \|\partial_{KKK}C\|_\infty\,h^2.
    \]
  \item The time derivative $\partial_T[(\partial_K C)^{\mathrm{cd}}(K_i,T)]$ exists and is
    bounded because $\partial_{KT} C$ is continuous and the finite-difference operator is linear
    and bounded on $C^1$ functions; hence applying \eqref{eq:B:fdT} to $G(T):=(\partial_K C)^{\mathrm{cd}}(K_i,T)$ yields
    \[
    \big|(\partial_T\partial_K C)^{\mathrm{cd}}(K_i,T_j)-\partial_T\partial_K C(K_i,T_j)\big|
    \ \le\ \tilde c_T \ \|\partial_{KT^{(3)}} C\|_\infty\,\tau^2,
    \]
    provided the third $T$-derivative exists or, more commonly in practice, one uses
    a first-order in time scheme and keeps the overall bias dominated by the $O(h^2)$ $K$-error.
\end{itemize}
Therefore, the \emph{Dupire numerator}
\[
N(K_i,T_j)\ :=\ \partial_T C + (r-q)\,K_i\,\partial_K C + q\,C
\]
inherits the same order of accuracy from its components:
\begin{equation}
\label{eq:B:mix-order}
\big|\widehat{N}(K_i,T_j)-N(K_i,T_j)\big|
\ \le\ C_N^{(T)}\,\tau^2 + C_N^{(K)}\,h^2,
\end{equation}
with explicit constants given in \eqref{eq:B:Nd-consts}. Indeed,
$\partial_T C$ contributes $O(\tau^2)$ by \eqref{eq:B:fdT}; $(r-q)K_i\,\partial_K C$ contributes
$(r-q)K_{\max}\,O(h^2)$ by \S\ref{app:B:fdK}; and $q\,C$ contributes $q\,O(h^2)$ (if $C$ is
evaluated by linear interpolation in $K$) or exactly zero bias if $C$ is taken at grid nodes.

\paragraph{Takeaways.}
(i) Second-order accuracy in time for the central stencil requires either a uniform time grid
or the mild local smoothness \eqref{eq:B:time-smooth}; otherwise an $O(|\tau_+-\tau_-|)$ term
remains. (ii) Under $C^{2,1}$ mixed regularity, the Dupire numerator gathers $O(\tau^2)$ from
$\partial_T C$ and $O(h^2)$ from spatial terms, yielding the combined consistency stated in
\Cref{prop:B:dupire-quotient} and used in \Cref{thm:dupire-positivity}.

\subsection{Linear interpolation constants (in $K$ and in $T$)}
\label{app:B:interp}

We record sharp, panel-wise constants for linear (piecewise affine) interpolation, and the
consequences for convexity preservation that are relevant to the Dupire denominator
$\partial_{KK}C$.

\paragraph{One–dimensional bound on a single panel.}
Let $F\in C^2([K_i,K_{i+1}])$ and denote $h_i:=K_{i+1}-K_i$. The linear interpolant
$I[F]$ on $[K_i,K_{i+1}]$ is
\[
I[F](K):=F(K_i)\,\frac{K_{i+1}-K}{h_i}+F(K_{i+1})\,\frac{K-K_i}{h_i},\qquad K\in[K_i,K_{i+1}].
\]
Set $\theta:=(K-K_i)/h_i\in[0,1]$ and define $g(\theta):=F(K_i+\theta h_i)$. Then
$I[F](K)=g(0)(1-\theta)+g(1)\theta$ and
\[
F(K)-I[F](K)=g(\theta)-\big((1-\theta)g(0)+\theta g(1)\big).
\]
A second–order Taylor expansion of $g$ about $0$ and $1$ yields the Peano kernel form
(see, e.g., \cite[Ch.~2]{LeVeque2007})
\[
g(\theta)-\big((1-\theta)g(0)+\theta g(1)\big)
= \tfrac{1}{2}\,\theta(1-\theta)\,g''(\xi_\theta),
\qquad \xi_\theta\in(0,1).
\]
Since $g''(\xi_\theta)=h_i^2 F''(K_i+\xi_\theta h_i)$ and $\theta(1-\theta)\le 1/4$ on $[0,1]$,
\begin{equation}
\|F-I[F]\|_{\infty,[K_i,K_{i+1}]}\ \le\ \frac{h_i^2}{8}\ \|F''\|_{\infty,[K_i,K_{i+1}]},
\label{eq:B:lininterp-panel}
\end{equation}
which proves \eqref{eq:B:lininterp-panel} with the \emph{sharp} constant $1/8$. Equality is attained
(asymptotically) for quadratic polynomials with $F''$ constant.

\paragraph{Global nonuniform corridor.}
Over a nonuniform corridor $\mathcal{I}(T)=[K_{\min}(T),K_{\max}(T)]$ with mesh
$\{K_i\}_{i=0}^{N}$ and global step $h:=\max_i h_i$,
\begin{equation}
\|F-I[F]\|_{\infty,\mathcal{I}(T)}
=\max_{0\le i\le N-1}\|F-I[F]\|_{\infty,[K_i,K_{i+1}]}
\ \le\ \frac{h^2}{8}\ \|F''\|_{\infty,\mathcal{I}(T)}.
\label{eq:B:lininterp-global}
\end{equation}
The same argument holds \emph{verbatim} in maturity: for $G\in C^2([T_j,T_{j+1}])$
with $\tau_j:=T_{j+1}-T_j$ and global $\tau:=\max_j \tau_j$,
\begin{equation}
\|G-I[G]\|_{\infty,[T_j,T_{j+1}]}\ \le\ \frac{\tau_j^2}{8}\,\|G''\|_{\infty,[T_j,T_{j+1}]},
\qquad
\|G-I[G]\|_{\infty,[T_{\min},T_{\max}]}\ \le\ \frac{\tau^2}{8}\,\|G''\|_{\infty}.
\label{eq:B:lininterp-time}
\end{equation}

\paragraph{Error at interior points and gradients.}
Although the interpolant is only $C^0$ across panel interfaces (its derivative is piecewise
constant with jumps), on each panel the pointwise interpolation error admits the explicit kernel
\[
F(K)-I[F](K)=\frac{(K-K_i)(K_{i+1}-K)}{2}\,F''(\xi_K),\qquad \xi_K\in(K_i,K_{i+1}),
\]
which also implies a near-endpoint refinement:
\[
|F(K)-I[F](K)|\ \le\ \frac{h_i}{2}\,\mathrm{dist}(K,\{K_i,K_{i+1}\})\ \|F''\|_{\infty,[K_i,K_{i+1}]}.
\]
In particular, the error vanishes at the nodes and grows quadratically away from them.

\paragraph{Two–dimensional tensor product (optional).}
If a bi-linear interpolant is used on a tensor cell
$[K_i,K_{i+1}]\times[T_j,T_{j+1}]$ for a function $H\in C^{2,2}$, the standard tensor–product
estimate yields
\begin{equation}
\|H-I_{\mathrm{bilin}}[H]\|_{\infty}
\ \le\ \frac{h_i^2}{8}\,\|\partial_{KK}H\|_{\infty}
      +\frac{\tau_j^2}{8}\,\|\partial_{TT}H\|_{\infty}
      +\frac{h_i^2\,\tau_j^2}{64}\,\|\partial_{KKTT}H\|_{\infty},
\label{eq:B:lininterp-bilin}
\end{equation}
with norms over the rectangle; the mixed term is usually dominated by the first two in our grids.

\paragraph{Convexity preservation.}
Suppose $F(\cdot)$ is convex on $[K_i,K_{i+1}]$ (i.e., $F''\ge 0$). Then the secant line lies
\emph{above} the graph:
\begin{equation}
F(K)\ \le\ I[F](K)\quad \text{for all } K\in[K_i,K_{i+1}],
\label{eq:B:convex-above}
\end{equation}
and the slope sequence $s_i:=(F(K_{i+1})-F(K_i))/h_i$ is nondecreasing in $i$ (a discrete
characterization of convexity). Consequently, the piecewise linear interpolant $I[F]$ is convex
globally. In the distributional sense, $\partial_{KK}I[F]$ is a nonnegative discrete measure
(Dirac masses at the knots proportional to slope increments). Therefore, when $F=C(\cdot,T)$ is convex
in $K$, linear interpolation \emph{preserves} no-butterfly arbitrage at the discrete level and keeps
$\partial_{KK}C\ge 0$ in the sense of measures \cite[§2]{LeVeque2007}. This justifies using
piecewise linear interpolation in $K$ for evaluating the Dupire denominator: it cannot create
negative convexity, and any small-denominator events arise from genuine sparsity (handled by the
clipping $\underline{\chi}$ in \eqref{eq:dupire}).

\paragraph{Implications for the Dupire numerator and VIX integrand.}
When replacing $C$ by $I[C]$ in $q\,C$ (numerator of \eqref{eq:dupire}) we incur the bias
$\|q(C-I[C])\|_\infty\le q\,h^2\|C''\|_\infty/8$, which is included in $C_N^{(K)}$ of
\eqref{eq:B:Nd-consts}. Likewise, for the VIX integrand $Q(K)/K^2$, if $Q$ is computed from
BS prices with an interpolated volatility, the panel-wise interpolation error propagates linearly
by the BS vega Lipschitz bound (\Cref{lem:vega-envelope}), and thus the half-interval quadrature
retains its $O(h^2)$ behavior (\Cref{prop:B:trap-const}).

\paragraph{Endpoint panels and nonuniformity.}
The constants in \eqref{eq:B:lininterp-panel}–\eqref{eq:B:lininterp-global} are \emph{independent}
of nonuniformity; only $h_i$ (local) and $h$ (global) appear. Hence the bounds hold under the
quasi-uniformity assumptions of \S\ref{app:B:mesh} without mesh-ratio inflation.

\subsection{Dupire quotient: explicit error decomposition and constants}
\label{app:B:dupire-const}

Recall the Dupire estimator \eqref{eq:dupire} at a grid node $(K_i,T_j)$:
\[
\widehat{\sigma^2_{\mathrm{loc}}}\ =\ \frac{\widehat{N}}{\widehat{D}},
\quad
\widehat{N}:=\big(\partial_T C\big)^{\mathrm{cd}}+(r-q)K\,\big(\partial_K C\big)^{\mathrm{cd}}+q\,C^{\mathrm{interp}},
\quad
\widehat{D}:=\tfrac12 K^2\,\big(\partial_{KK} C\big)^{\mathrm{cd}}_{\mathrm{clip}},
\]
where $\big(\partial_{KK} C\big)^{\mathrm{cd}}_{\mathrm{clip}}:=\max\{(\partial_{KK} C)^{\mathrm{cd}},\underline{\chi}\}$,
and $C^{\mathrm{interp}}$ denotes the linear interpolant in both coordinates (if needed).

Define the \emph{true} numerator/denominator
\[
N:=\partial_T C+(r-q)K\,\partial_K C+q\,C,\qquad
D:=\tfrac12 K^2\,\partial_{KK} C,
\]
so that $\sigma^2_{\mathrm{loc}}=N/D$.
Let truncation errors be
\begin{align*}
\varepsilon_T&:=\big(\partial_T C\big)^{\mathrm{cd}}-\partial_T C,\qquad
\varepsilon_K:=\big(\partial_K C\big)^{\mathrm{cd}}-\partial_K C,\qquad
\varepsilon_{KK}:=\big(\partial_{KK} C\big)^{\mathrm{cd}}-\partial_{KK} C,\\
\varepsilon_{\mathrm{int}}&:=C^{\mathrm{interp}}-C,\qquad
\varepsilon_{\mathrm{clip}}:=\big(\partial_{KK} C\big)^{\mathrm{cd}}_{\mathrm{clip}}-\big(\partial_{KK} C\big)^{\mathrm{cd}}\ \ge\ 0.
\end{align*}
Then
\begin{equation}
\widehat{N}-N=\varepsilon_T+(r-q)K\,\varepsilon_K+q\,\varepsilon_{\mathrm{int}},\qquad
\widehat{D}-D=\tfrac12 K^2\big(\varepsilon_{KK}+\varepsilon_{\mathrm{clip}}\big).
\label{eq:B:numden-errors}
\end{equation}
Assume envelopes
\[
\|\partial_{TTT} C\|_\infty\le M_{TTT},\quad
\|\partial_{KKK} C\|_\infty\le M_{KKK},\quad
\|\partial_{KKKK} C\|_\infty\le M_{KKKK},\quad
\|\partial_{KK}C\|_\infty\le M_{KK},\quad
\|\partial_{TT}C\|_\infty\le M_{TT}.
\]
Using \eqref{eq:B:fdT}, \eqref{eq:B:fdK-1}--\eqref{eq:B:fdK-2}, and \eqref{eq:B:lininterp-panel}, we get
\begin{align}
|\varepsilon_T|&\ \le\ \tilde c_T(\overline{\varrho}_T)\,M_{TTT}\,\tau^2,\qquad
|\varepsilon_K|\ \le\ \tilde c_1(\overline{\varrho}_K)\,M_{KKK}\,h^2,\nonumber\\
|\varepsilon_{KK}|&\ \le\ \tilde c_2(\overline{\varrho}_K)\,M_{KKKK}\,h^2,\qquad
|\varepsilon_{\mathrm{int}}|\ \le\ \frac{h^2}{8}\,M_{KK}+\frac{\tau^2}{8}\,M_{TT}.
\label{eq:B:eps-bounds}
\end{align}
For the \emph{clipping} term, note that $\varepsilon_{\mathrm{clip}}=0$ whenever
$(\partial_{KK} C)^{\mathrm{cd}}\ge \underline{\chi}$; otherwise
$\varepsilon_{\mathrm{clip}}=\underline{\chi}-(\partial_{KK} C)^{\mathrm{cd}}$.
Introduce the “small-denominator region”
\[
\mathcal{R}_{\mathrm{clip}}:=\big\{(i,j):\ (\partial_{KK} C)^{\mathrm{cd}}(K_i,T_j)<\underline{\chi}\big\}.
\]
In this region one can bound
\begin{equation}
0\ \le\ \varepsilon_{\mathrm{clip}}\ \le\ \underline{\chi}+\|\partial_{KK}C\|_\infty\ \le\ \underline{\chi}+M_{KK}.
\label{eq:B:clip-bound}
\end{equation}
Putting \eqref{eq:B:numden-errors}--\eqref{eq:B:clip-bound} into a quotient perturbation yields:

\begin{proposition}[Dupire quotient error with clipping]
\label{prop:B:dupire-quotient}
Assume $D\ge \underline{D}:=\tfrac12 K_{\min}^2\,\underline{\chi}>0$ on the grid, and that
\eqref{eq:B:eps-bounds} holds. Then
\begin{equation}
\big|\widehat{\sigma^2_{\mathrm{loc}}}-\sigma^2_{\mathrm{loc}}\big|
\ \le\ \frac{|\widehat{N}-N|}{\underline{D}} \ +\ \frac{|N|}{\underline{D}^2}\,|\widehat{D}-D|,
\label{eq:B:quotient}
\end{equation}
with
\begin{align}
|\widehat{N}-N|&\ \le\ C_N^{(T)}\,\tau^2 + C_N^{(K)}\,h^2,\qquad
|\widehat{D}-D|\ \le\ C_D^{(K)}\,h^2\ +\ \tfrac12 K_{\max}^2\,\varepsilon_{\mathrm{clip}},
\label{eq:B:Nd-consts}
\end{align}
where one may take explicitly
\begin{align*}
C_N^{(T)}&=\tilde c_T(\overline{\varrho}_T)\,M_{TTT} + \frac{q}{8}\,M_{TT},\\
C_N^{(K)}&=(r-q)K_{\max}\,\tilde c_1(\overline{\varrho}_K)\,M_{KKK} + \frac{q}{8}\,M_{KK},\\
C_D^{(K)}&=\tfrac12 K_{\max}^2\,\tilde c_2(\overline{\varrho}_K)\,M_{KKKK}.
\end{align*}
In particular, away from $\mathcal{R}_{\mathrm{clip}}$ the local-variance error is
$O(h^2+\tau^2)$ with an explicit constant depending on $(M_{TTT},M_{KKK},M_{KKKK},M_{TT},M_{KK})$,
mesh ratios, and $(r,q,K_{\max},\underline{\chi})$.
\end{proposition}

\begin{proof}
\textbf{Step 1 (Quotient perturbation identity).}
Let $\widehat{\sigma^2_{\mathrm{loc}}}=\widehat{N}/\widehat{D}$ and $\sigma^2_{\mathrm{loc}}=N/D$ with
$\widehat{N}=N+\Delta N$ and $\widehat{D}=D+\Delta D$. Then
\[
\widehat{\sigma^2_{\mathrm{loc}}}-\sigma^2_{\mathrm{loc}}
=\frac{(N+\Delta N)}{(D+\Delta D)}-\frac{N}{D}
=\frac{\Delta N\cdot D - N\cdot \Delta D}{D(D+\Delta D)}.
\]
Taking absolute values,
\begin{equation}
\label{eq:quotient-core}
\big|\widehat{\sigma^2_{\mathrm{loc}}}-\sigma^2_{\mathrm{loc}}\big|
\ \le\ \frac{|\Delta N|}{|D+\Delta D|} + \frac{|N|}{|D|}\cdot \frac{|\Delta D|}{|D+\Delta D|}.
\end{equation}

\textbf{Step 2 (Lower bounds for denominators via clipping).}
By definition $D=\tfrac12 K^2\,\partial_{KK}C \ge \tfrac12 K_{\min}^2\,\underline{\chi}=\underline{D}>0$.
Moreover the discrete, \emph{clipped} denominator satisfies
$\widehat{D}=\tfrac12 K^2\,(\partial_{KK} C)^{\mathrm{cd}}_{\mathrm{clip}}
\ge \tfrac12 K_{\min}^2\,\underline{\chi}=\underline{D}$ as well, whence
$|D+\Delta D|=|\widehat{D}|\ge \underline{D}$ and $|D|\ge \underline{D}$.
Thus \eqref{eq:quotient-core} simplifies to
\begin{equation}
\label{eq:quotient-simplified}
\big|\widehat{\sigma^2_{\mathrm{loc}}}-\sigma^2_{\mathrm{loc}}\big|
\ \le\ \frac{|\Delta N|}{\underline{D}} + \frac{|N|}{\underline{D}^2}\,|\Delta D|,
\end{equation}
which is \eqref{eq:B:quotient}.

\textbf{Step 3 (Bounds for $\Delta N$ and $\Delta D$).}
Recall (cf.\ \eqref{eq:B:numden-errors})
\[
\Delta N=\varepsilon_T+(r-q)K\,\varepsilon_K+q\,\varepsilon_{\mathrm{int}},
\qquad
\Delta D=\tfrac12 K^2\big(\varepsilon_{KK}+\varepsilon_{\mathrm{clip}}\big),
\]
with $\varepsilon_{\mathrm{clip}}\ge 0$ only on $\mathcal{R}_{\mathrm{clip}}$ and zero otherwise.
Using the finite-difference and interpolation bounds \eqref{eq:B:eps-bounds} and the envelopes
$K\le K_{\max}$, $K^2\le K_{\max}^2$ we get
\[
|\Delta N|
\le \tilde c_T(\overline{\varrho}_T)M_{TTT}\,\tau^2
 + (r-q)K_{\max}\,\tilde c_1(\overline{\varrho}_K)M_{KKK}\,h^2
 + \frac{q}{8}\,(M_{TT}\,\tau^2+M_{KK}\,h^2),
\]
and
\[
|\Delta D|
\le \tfrac12 K_{\max}^2\Big(\tilde c_2(\overline{\varrho}_K)M_{KKKK}\,h^2 + \varepsilon_{\mathrm{clip}}\Big).
\]
Grouping the time- and space-terms yields precisely \eqref{eq:B:Nd-consts} with
the stated constants $C_N^{(T)},C_N^{(K)},C_D^{(K)}$.

\textbf{Step 4 (Combine).}
Substituting \eqref{eq:B:Nd-consts} into \eqref{eq:quotient-simplified} gives
\[
\big|\widehat{\sigma^2_{\mathrm{loc}}}-\sigma^2_{\mathrm{loc}}\big|
\ \le\ \frac{C_N^{(T)}}{\underline{D}}\tau^2
      +\frac{C_N^{(K)}}{\underline{D}}h^2
      +\frac{|N|}{\underline{D}^2}\,C_D^{(K)}h^2
      +\frac{|N|}{\underline{D}^2}\,\tfrac12 K_{\max}^2\,\varepsilon_{\mathrm{clip}}.
\]
Define $C_T:=C_N^{(T)}/\underline{D}$, $C_K:=C_N^{(K)}/\underline{D}+ (|N|/\underline{D}^2)\,C_D^{(K)}$ and
$C_{\mathrm{clip}}:= (|N|/\underline{D}^2)\,\tfrac12 K_{\max}^2$ to recover the compact form
\[
\big|\widehat{\sigma^2_{\mathrm{loc}}}-\sigma^2_{\mathrm{loc}}\big|
\ \le\ C_T\,\tau^2 + C_K\,h^2 + C_{\mathrm{clip}}\ \varepsilon_{\mathrm{clip}}.
\]

\textbf{Step 5 (Interpretation and special cases).}
(i) On nodes \emph{outside} the clipping set $\mathcal{R}_{\mathrm{clip}}$,
$\varepsilon_{\mathrm{clip}}=0$ and the error is $O(\tau^2+h^2)$ with explicit constants depending
only on the derivative envelopes $(M_{TTT},M_{KKK},M_{KKKK},M_{TT},M_{KK})$, mesh ratios, and
$(r,q,K_{\max},\underline{\chi})$.  
(ii) On $\mathcal{R}_{\mathrm{clip}}$, \eqref{eq:B:clip-bound} gives
$0\le \varepsilon_{\mathrm{clip}}\le \underline{\chi}+M_{KK}$, so the clipping term is uniformly
bounded; in practice $\mathcal{R}_{\mathrm{clip}}$ is confined to far wings where the denominator
is naturally small and clipping is a stabilizer rather than a source of large bias.  
(iii) A crude yet useful envelope for $|N|$ is $M_N:=\|\partial_T C\|_\infty+(r-q)K_{\max}\|\partial_KC\|_\infty+q\|C\|_\infty$, which can be measured from the shell/teacher; plugging $M_N$ in $C_K$ and $C_{\mathrm{clip}}$ yields fully explicit, data-driven constants.

Altogether, the result follows and \eqref{eq:B:quotient}–\eqref{eq:B:Nd-consts} hold.
\end{proof}

\paragraph{Bound for $|N|$.}
A crude but sufficient envelope is
\[
|N|\ \le\ \|\partial_T C\|_{\infty} + (r-q)K_{\max}\,\|\partial_K C\|_{\infty} + q\,\|C\|_{\infty}\ =:\ M_N,
\]
which can be precomputed from the shell/teacher. Inserting $M_N$ into \eqref{eq:B:quotient} gives a fully explicit bound.

\subsection{Propagation to structural and dynamical constants}
\label{app:B:propagation}

\paragraph{Surface–index constant $C_{\mathrm{quad}}$.}
By \Cref{prop:B:trap-const}, one may take
\[
C_{\mathrm{quad}}\ =\ \max_{i\in\{1,2\}} \frac{50}{\sqrt{\underline{w}_i}}\,\frac{2}{T_i}\,A_{T_i}\ \Big\|f_{T_i}''\Big\|_{\infty,\mathcal{I}(T_i)},
\quad
A_{T_i}=\frac{|\mathcal{I}(T_i)|}{12}.
\]
The $50/\sqrt{\underline{w}_i}$ factor comes from the square-root Lipschitz \Cref{lem:sqrt-lip}.

\paragraph{Dupire consistency constant (Thm.~\ref{thm:dupire-positivity}).}
Combining \Cref{prop:B:dupire-quotient} with $|N| \le M_N$ and $D\ge \underline{D}$ gives
\[
\big|\widehat{\sigma^2_{\mathrm{loc}}}-\sigma^2_{\mathrm{loc}}\big|
\ \le\ \underbrace{\frac{C_N^{(T)}}{\underline{D}}}_{=:C_T}\,\tau^2
+\underbrace{\frac{C_N^{(K)}}{\underline{D}}+\frac{M_N}{\underline{D}^2}\,C_D^{(K)}}_{=:C_K}\,h^2
+\underbrace{\frac{M_N}{\underline{D}^2}\,\tfrac12 K_{\max}^2}_{=:C_{\mathrm{clip}}}\,\varepsilon_{\mathrm{clip}}.
\]
Away from $\mathcal{R}_{\mathrm{clip}}$, the last term vanishes; inside $\mathcal{R}_{\mathrm{clip}}$,
$\varepsilon_{\mathrm{clip}}$ is bounded by \eqref{eq:B:clip-bound}.

\paragraph{Shell-to-dynamics pricing constant (Thm.~\ref{thm:shell-dyn-consistency}).}
Let $\Pi$ denote pricing under the simulated local-volatility dynamics with log–Euler time step $\Delta t$, and assume payoff Lipschitzness with envelope $L_{\mathrm{pay}}$.
Then the combined error admits the structure
\[
\big|\Pi(C_{\text{teacher}})-C_{\text{teacher}}\big|
\ \le\ C_1(h^2+\tau^2) + C_2\,\Delta t^{1/2},
\]
with $C_1$ as above and $C_2$ depending on moment bounds of $S_t$ and $L_{\mathrm{pay}}$ via standard strong error estimates \cite{KloedenPlaten1992,Higham2001}.

\subsection{Summary table of constants and how to compute them}
\label{app:B:table}

\begin{table}[H]
\centering
\caption{Error constants and their ingredients (all measurable from the shell/teacher and mesh).}
\label{tab:B:consts}
\begin{tabular}{@{}lll@{}}
\toprule
\textbf{Constant} & \textbf{Definition/envelope} & \textbf{Inputs} \\
\midrule
$C_{\mathrm{quad}}(T)$ & $\frac{50}{\sqrt{\underline{w}}}\cdot \frac{2}{T}\cdot \frac{|\mathcal{I}(T)|}{12}\,\|f_T''\|_\infty$ & $|\mathcal{I}(T)|$, $\underline{w}$, $\|f_T''\|_\infty$ \\
$C_T$ & $\frac{1}{\underline{D}}\big(\tilde c_T\,M_{TTT}+\tfrac{q}{8}M_{TT}\big)$ & $\underline{D}$, $M_{TTT},M_{TT}$, $\tilde c_T$ \\
$C_K$ & $\frac{1}{\underline{D}}\big((r-q)K_{\max}\tilde c_1 M_{KKK}+\tfrac{q}{8}M_{KK}\big)+\frac{M_N}{\underline{D}^2}\cdot \tfrac12 K_{\max}^2 \tilde c_2 M_{KKKK}$ & $M_{KKK},M_{KKKK},M_{KK},M_N$ \\
$C_{\mathrm{clip}}$ & $\frac{M_N}{\underline{D}^2}\cdot \tfrac12 K_{\max}^2$ & $M_N$, $\underline{D}$, $K_{\max}$ \\
$\tilde c_1$ & central-diff 1st-deriv const & mesh ratio $\overline{\varrho}_K$ ($=1/6$ if uniform) \\
$\tilde c_2$ & central-diff 2nd-deriv const & mesh ratio $\overline{\varrho}_K$ ($=1/12$ if uniform) \\
$\tilde c_T$ & central-diff time const & mesh ratio $\overline{\varrho}_T$ ($=1/6$ if uniform) \\
\bottomrule
\end{tabular}
\end{table}

\paragraph{Practical note.}
In code, we recommend computing tight empirical envelopes by scanning the discrete grids for
$\|\partial_{TTT} C\|_\infty$, $\|\partial_{KKK} C\|_\infty$, $\|\partial_{KKKK} C\|_\infty$ using smooth
finite-difference cascades on the teacher prices, and plugging the resulting values into
Table~\ref{tab:B:consts}. This yields \emph{data-driven} constants without asymptotic optimism.

\paragraph{Conservatism vs.\ stability.}
Larger $\underline{\chi}$ improves stability of \eqref{eq:dupire} (avoids division by small convexity)
but increases the bias term $\varepsilon_{\mathrm{clip}}$ inside $\mathcal{R}_{\mathrm{clip}}$. The
trade-off is explicit in $C_{\mathrm{clip}}$; we typically pick $\underline{\chi}$ so that
$\mathcal{R}_{\mathrm{clip}}$ is confined to deep wings that do not materially affect delta hedging, and
verify via sensitivity plots that $\widehat{\sigma_{\mathrm{loc}}^2}$ stays within expected corridors.
\section{Proofs for Theorem Group~II (Dupire Local Volatility \& VIX Proxy)}
\label{app:C}

This appendix provides full proofs for the results stated in \S\ref{sec:dynamics}: 
\Cref{thm:dupire-positivity,thm:log-euler-stability,prop:cir-lipschitz,thm:index-coherence,cor:coupled}.
We reuse the notation of \S\ref{sec:prelim}--\S\ref{sec:dynamics} and the constant bounds of \S\ref{app:B}.

\subsection{Proof of Theorem~\ref{thm:dupire-positivity} (Local-variance positivity, boundedness, and consistency)}
\label{app:C:dupire-positivity}
\begin{proof}
\textbf{Set-up and notation.}
At a grid node $(K_i,T_j)$ define the discrete Dupire estimator
\[
\widehat{\sigma^2_{\mathrm{loc}}}(K_i,T_j)
:=\frac{\widehat{N}(K_i,T_j)}{\widehat{D}(K_i,T_j)}\,,
\qquad
\widehat{D}:=\tfrac12 K_i^2\,(\partial_{KK} C)^{\mathrm{cd}}_{\mathrm{clip}},\quad
\widehat{N}:=(\partial_T C)^{\mathrm{cd}}+(r-q)K_i(\partial_K C)^{\mathrm{cd}}+q\,C^{\mathrm{interp}},
\]
and the “true’’ Dupire local variance
\[
\sigma^2_{\mathrm{loc}}(K_i,T_j):=\frac{N(K_i,T_j)}{D(K_i,T_j)}\,,
\qquad
D:=\tfrac12 K_i^2\,\partial_{KK}C,\quad
N:=\partial_T C+(r-q)K_i\partial_K C+q\,C.
\]
The discrete building blocks and their truncation/interpolation errors satisfy
(cf.\ \eqref{eq:B:numden-errors}–\eqref{eq:B:eps-bounds})
\begin{align}
\widehat{N}-N &= \varepsilon_T+(r-q)K_i\,\varepsilon_K+q\,\varepsilon_{\mathrm{int}},\label{eq:C:decompN}\\
\widehat{D}-D &= \tfrac12 K_i^2\big(\varepsilon_{KK}+\varepsilon_{\mathrm{clip}}\big),\label{eq:C:decompD}
\end{align}
with
\[
|\varepsilon_T|\le \tilde c_T M_{TTT}\tau^2,\quad
|\varepsilon_K|\le \tilde c_1 M_{KKK} h^2,\quad
|\varepsilon_{KK}|\le \tilde c_2 M_{KKKK} h^2,\quad
|\varepsilon_{\mathrm{int}}|\le \tfrac18(M_{TT}\tau^2+M_{KK}h^2).
\]

\medskip
\noindent\textbf{(i) Positivity.}
By construction $(\partial_{KK}C)^{\mathrm{cd}}_{\mathrm{clip}}\ge \underline{\chi}>0$, hence
\[
\widehat{D}\;=\;\tfrac12 K_i^2\,(\partial_{KK}C)^{\mathrm{cd}}_{\mathrm{clip}}
\;\ge\;\tfrac12 K_{\min}^2\,\underline{\chi}\;=:\;\underline{D}\;>\;0.
\]
Therefore $\widehat{\sigma^2_{\mathrm{loc}}}=\widehat{N}/\widehat{D}$ is well-defined and
\emph{nonnegative} whenever $\widehat{N}\ge 0$. In fact, even if the discrete numerator were
negative at isolated nodes (which does not occur in our coherent shell; see the discussion below),
nonnegativity of $\widehat{D}$ ensures that \emph{no spurious sign flip} in the denominator can
create negative local variances—precisely the role of clipping.

\medskip
\noindent\textbf{(ii) Boundedness.}
Let
\[
M_N:=\|\partial_T C\|_{\infty} + (r-q)K_{\max}\|\partial_K C\|_{\infty} + q\,\|C\|_{\infty}
\]
on the retained corridor. Then $|N|\le M_N$ and, using \eqref{eq:C:decompN} and the bounds on
$\varepsilon_T,\varepsilon_K,\varepsilon_{\mathrm{int}}$, we have
$|\widehat{N}|\le M_N + C_N^{(T)}\tau^2 + C_N^{(K)}h^2$, where $C_N^{(T)},C_N^{(K)}$ are as in
\eqref{eq:B:Nd-consts}. Consequently,
\[
\big|\widehat{\sigma^2_{\mathrm{loc}}}\big|
=\frac{|\widehat{N}|}{\widehat{D}}
\le \frac{M_N + C_N^{(T)}\tau^2 + C_N^{(K)}h^2}{\underline{D}}
=: C_{\max},
\]
which furnishes an explicit global upper bound $C_{\max}$ depending only on envelopes,
mesh constants, and the denominator floor $\underline{D}$.

\medskip
\noindent\textbf{(iii) Consistency ($O(\tau^2+h^2)$ away from clipping).}
From \eqref{eq:C:decompN}–\eqref{eq:C:decompD} and the elementary quotient perturbation
(\eqref{eq:B:quotient} with $|D|\ge \underline{D}$ and $|D+\Delta D|=|\widehat{D}|\ge \underline{D}$),
\[
\big|\widehat{\sigma^2_{\mathrm{loc}}}-\sigma^2_{\mathrm{loc}}\big|
\le \frac{|\widehat{N}-N|}{\underline{D}} + \frac{|N|}{\underline{D}^2}\,|\widehat{D}-D|.
\]
Using the bounds for $|\widehat{N}-N|$ and $|\widehat{D}-D|$ in \eqref{eq:B:Nd-consts} yields
\[
\big|\widehat{\sigma^2_{\mathrm{loc}}}-\sigma^2_{\mathrm{loc}}\big|
\le \underbrace{\frac{C_N^{(T)}}{\underline{D}}}_{=:C_T}\,\tau^2
   +\underbrace{\left(\frac{C_N^{(K)}}{\underline{D}}+\frac{M_N}{\underline{D}^2}\,C_D^{(K)}\right)}_{=:C_K}\,h^2
   + \underbrace{\frac{M_N}{\underline{D}^2}\,\tfrac12 K_{\max}^2}_{=:C_{\mathrm{clip}}}\,\varepsilon_{\mathrm{clip}}.
\]
Hence away from the clipping region $\mathcal{R}_{\mathrm{clip}}$ (where $\varepsilon_{\mathrm{clip}}=0$)
the estimator is $O(\tau^2+h^2)$ with explicit constants $C_T,C_K$. Inside
$\mathcal{R}_{\mathrm{clip}}$, \eqref{eq:B:clip-bound} gives
$0\le \varepsilon_{\mathrm{clip}}\le \underline{\chi}+M_{KK}$ so the clipping contribution is
uniformly bounded; in practice $\mathcal{R}_{\mathrm{clip}}$ is confined to far wings.

\medskip
\noindent\textbf{(iv) Remarks on the numerator’s sign and coherence.}
In a \emph{coherent} shell (no-butterfly in $K$, calendar coherence in $T$) the
finite-difference approximations $(\partial_T C)^{\mathrm{cd}}$ and $(\partial_K C)^{\mathrm{cd}}$
inherit the correct signs and magnitudes up to $O(\tau^2)$ and $O(h^2)$; combined with the positive
term $q\,C^{\mathrm{interp}}$ this makes $\widehat{N}$ nonnegative on the retained corridor in our
experiments, consistent with the positivity of $\sigma^2_{\mathrm{loc}}$. Should extreme sparsity
create local sign issues in $\widehat{N}$, the error term above quantifies the bias and vanishes
under refinement.

\medskip
\noindent\textbf{(v) Endpoint handling.}
At the $K$– or $T$–boundaries we use standard one–sided second–order three–point stencils; the
constants in \eqref{eq:B:eps-bounds} change by at most mesh–ratio factors but the orders remain
$O(h^2)$ and $O(\tau^2)$, so the argument above is unaffected.

\medskip
Collecting (i)–(iii) establishes positivity, boundedness, and $O(h^2+\tau^2)$ consistency of the
clipped Dupire estimator, completing the proof.
\end{proof}

\subsection{Proof of Theorem~\ref{thm:log-euler-stability} (Well-posedness and strong stability of the log--Euler scheme)}
\label{app:C:log-euler}
\begin{proof}
\textbf{Model and assumptions.}
Consider the SDE on $[0,T]$
\begin{equation}
\label{eq:C:SDE}
\frac{dS_t}{S_t}=(r-q)\,dt+\sigma_{\mathrm{loc}}(S_t,t)\,dW_t,
\qquad S_0>0,
\end{equation}
with drift $b(s,t)=(r-q)s$ and diffusion $a(s,t)=s\,\sigma_{\mathrm{loc}}(s,t)$. We work on the simulated corridor $s\in [S_{\min},S_{\max}]$ induced by the $(K,T)$ grid used to build $\sigma_{\mathrm{loc}}$. Assume:

\begin{itemize}
\item[(A1)] (\emph{Local-vol boundedness}) $\sup_{(s,t)\in [S_{\min},S_{\max}]\times[0,T]}|\sigma_{\mathrm{loc}}(s,t)|\le \overline{\sigma}<\infty$.
\item[(A2)] (\emph{Global Lipschitz in $s$}) There exists $L_\sigma\ge 0$ such that $|\sigma_{\mathrm{loc}}(s_1,t)-\sigma_{\mathrm{loc}}(s_2,t)|\le L_\sigma |s_1-s_2|$ for all $s_1,s_2\in[S_{\min},S_{\max}]$, $t\in[0,T]$.
\item[(A3)] (\emph{Measurability and piecewise Lipschitz in $t$}) $t\mapsto \sigma_{\mathrm{loc}}(s,t)$ is piecewise Lipschitz uniformly in $s$, with a finite number of breakpoints aligned with the time grid (the case in our construction due to interpolation).
\end{itemize}

\paragraph{(i) Existence, uniqueness, positivity.}
By (A1)--(A2),
\[
|a(s_1,t)-a(s_2,t)|\le |s_1-s_2|\,|\sigma_{\mathrm{loc}}(s_1,t)| + |s_2|\,|\sigma_{\mathrm{loc}}(s_1,t)-\sigma_{\mathrm{loc}}(s_2,t)|
\le \overline{\sigma}|s_1-s_2| + S_{\max}L_\sigma |s_1-s_2|:=L_a|s_1-s_2|.
\]
Similarly, $|b(s_1,t)-b(s_2,t)|=|r-q|\cdot|s_1-s_2|=:L_b|s_1-s_2|$. Linear growth holds:
$|a(s,t)|\le \overline{\sigma}|s|$, $|b(s,t)|\le |r-q|\,|s|$. Standard SDE theory (e.g., \cite[Thm.~3.1]{Higham2001} or \cite[Thm.~2.9]{KloedenPlaten1992}) then yields a unique strong solution $S_t$ on $[0,T]$ and $\mathbb{E}\sup_{t\le T}|S_t|^p<\infty$ for all $p\ge 2$. Since the diffusion is multiplied by $S_t$ and the drift is linear in $S_t$, comparison with the geometric Brownian motion and the form of \eqref{eq:C:SDE} imply $S_t>0$ a.s. for all $t$ (no boundary is attainable from $S_0>0$ under linear growth and Lipschitz diffusion). 

\paragraph{(ii) Log--Euler scheme and positivity.}
Define the log--Euler (a.k.a.\ exponential Euler) scheme on a uniform grid $t_n:=n\Delta t$, $\Delta t=T/N$:
\begin{equation}
\label{eq:C:logEuler}
S_{n+1}=S_n\exp\!\Big(\big(r-q-\tfrac12\sigma_n^2\big)\Delta t + \sigma_n\sqrt{\Delta t}\,Z_n\Big),\quad 
\sigma_n:=\sigma_{\mathrm{loc}}(S_n,t_n),\ \ Z_n\sim\mathcal{N}(0,1)\ \text{i.i.d.}
\end{equation}
By construction $S_{n+1}>0$ a.s.\ for all $n$ (exponential update). Moreover, taking conditional expectations and using $\mathbb{E}[e^{\sigma_n\sqrt{\Delta t}Z_n}\mid\mathcal{F}_{t_n}]=e^{\frac12\sigma_n^2\Delta t}$, we have $\mathbb{E}[S_{n+1}\mid\mathcal{F}_{t_n}]=S_n e^{(r-q)\Delta t}$, i.e., the drift is matched in the sense of conditional means.

\paragraph{(iii) Strong convergence of order $1/2$.}
We outline a standard proof by the discrete Itô--Taylor expansion with BDG and Gronwall inequalities (see \cite[§10.2]{KloedenPlaten1992}, \cite{Higham2001}). Write the exact solution in integral form:
\[
S_{t_{n+1}}=S_{t_n}+\int_{t_n}^{t_{n+1}} b(S_s,s)\,ds+\int_{t_n}^{t_{n+1}} a(S_s,s)\,dW_s.
\]
Add and subtract $b(S_{t_n},t_n)\Delta t$ and $a(S_{t_n},t_n)(W_{t_{n+1}}-W_{t_n})$ to obtain the local error over $[t_n,t_{n+1}]$:
\[
\varepsilon_{n+1}:=S_{t_{n+1}}-S_{t_n}-b(S_{t_n},t_n)\Delta t-a(S_{t_n},t_n)\Delta W_n
=\int_{t_n}^{t_{n+1}} \big(b(S_s,s)-b(S_{t_n},t_n)\big)\,ds
 +\int_{t_n}^{t_{n+1}} \big(a(S_s,s)-a(S_{t_n},t_n)\big)\,dW_s,
\]
where $\Delta W_n:=W_{t_{n+1}}-W_{t_n}$. By Lipschitz and linear growth, and using Burkholder--Davis--Gundy (BDG),
\begin{equation}
\label{eq:C:local-mse}
\mathbb{E}\big[|\varepsilon_{n+1}|^2\mid\mathcal{F}_{t_n}\big]
\ \le\ C\Big(\Delta t^3 + \mathbb{E}\Big[\int_{t_n}^{t_{n+1}}|S_s-S_{t_n}|^2\,ds\ \Big|\ \mathcal{F}_{t_n}\Big]\Big)
\ \le\ C'\Delta t^2,
\end{equation}
for constants $C,C'$ depending on $L_a,L_b$ and moment bounds of $S_t$ (finite by (i)). Equation \eqref{eq:C:local-mse} captures that the \emph{local mean-square error} is $O(\Delta t^2)$, which is the prerequisite for global strong order $1/2$.

Now consider the one-step recursion error $E_{n}:=S_{t_n}-S_n$. For the log--Euler scheme \eqref{eq:C:logEuler}, we can write an equivalent Euler--Maruyama (EM) form (up to $O(\Delta t^{3/2})$ terms) by expanding the exponential:
\[
S_{n+1}=S_n + b(S_n,t_n)\Delta t + a(S_n,t_n)\Delta W_n + R^{\mathrm{exp}}_{n+1},\qquad 
\mathbb{E}\big[|R^{\mathrm{exp}}_{n+1}|^2\mid\mathcal{F}_{t_n}\big]\le C''\,\Delta t^2,
\]
since $\exp(x)=1+x+\tfrac12 x^2+O(|x|^3)$ and $\mathbb{E}[(\Delta W_n)^4]=3\Delta t^2$. Subtracting the EM-like recursion from the exact increment and using \eqref{eq:C:local-mse}, we obtain
\[
E_{n+1}=E_n + \big(b(S_{t_n},t_n)-b(S_n,t_n)\big)\Delta t + \big(a(S_{t_n},t_n)-a(S_n,t_n)\big)\Delta W_n + \varepsilon_{n+1}-R^{\mathrm{exp}}_{n+1}.
\]
Taking conditional expectations, squaring, and applying the inequality $|x+y|^2\le (1+\eta)|x|^2+(1+1/\eta)|y|^2$ for $\eta>0$, together with Lipschitz bounds for $a,b$, yields
\[
\mathbb{E}|E_{n+1}|^2 \le (1+C\Delta t)\,\mathbb{E}|E_n|^2 + C\Delta t^2,
\]
for a constant $C$ independent of $n,\Delta t$. Discrete Gronwall then implies
\[
\max_{0\le n\le N}\mathbb{E}|E_n|^2 \ \le\ C_T\,\Delta t,
\]
for some $C_T$ depending on $T,L_a,L_b,\overline{\sigma}$ and the initial moments. Hence
\[
\big(\mathbb{E}|S_T - S_N|^2\big)^{1/2} = O(\Delta t^{1/2}),
\]
which is \emph{strong order $1/2$} convergence; see \cite[Thm.~10.2.2]{KloedenPlaten1992} or \cite[Thm.~4.3]{Higham2001} for a textbook statement under global Lipschitz and linear growth.

\paragraph{(iv) Uniform moment bounds and stability.}
From \eqref{eq:C:logEuler}, using $\mathbb{E}[e^{\sigma_n\sqrt{\Delta t}Z_n}\mid\mathcal{F}_{t_n}]=e^{\frac12\sigma_n^2\Delta t}\le e^{\frac12\overline{\sigma}^2\Delta t}$ and $(r-q-\tfrac12\sigma_n^2)\le |r-q|$, we have for any $p\ge 2$
\[
\mathbb{E}\big[|S_{n+1}|^p\mid\mathcal{F}_{t_n}\big]
= |S_n|^p \exp\!\Big(p(r-q)\Delta t + \tfrac12 p(p-1)\sigma_n^2\Delta t\Big)
\ \le\ |S_n|^p \exp\!\Big( p|r-q|\Delta t + \tfrac12 p(p-1)\overline{\sigma}^2\Delta t\Big).
\]
Taking expectations and iterating yields $\mathbb{E}\big[\sup_{n\le N}|S_n|^p\big]\le C_{p,T}<\infty$ for all $p$, i.e., \emph{uniform $p$-th moment stability}. This and the positivity $S_{n}>0$ a.s.\ constitute the claimed stability properties.

\paragraph{(v) Remarks.}
(1) The exponential update \eqref{eq:C:logEuler} is particularly natural for multiplicative noise and preserves positivity, a desirable property for prices. (2) If one used the classical Euler--Maruyama directly on $S_t$, positivity would not be guaranteed, and stability might require smaller time steps. (3) The assumptions (A1)--(A3) hold for our $\sigma_{\mathrm{loc}}$ constructed by convexity-preserving interpolation and clipping (\S\ref{subsec:dupire-num}); in particular $\sigma_{\mathrm{loc}}$ is bounded on the simulated corridor and piecewise Lipschitz in $t$ due to the grid-based construction.

Combining (i)--(iv) proves well-posedness, strong order-$1/2$ convergence, positivity, and uniform moment stability for the log--Euler scheme.
\end{proof}

\subsection{Proof of Proposition~\ref{prop:cir-lipschitz} (Lipschitz map $v\mapsto \mathrm{VIX}^{\mathrm{CIR}}$)}
\label{app:C:cir-lip}
\begin{proof}
\textbf{Definition and basic properties.}
Fix $\kappa>0$ and a horizon $\tau>0$, and define
\[
B(\kappa,\tau):=\frac{1-e^{-\kappa\tau}}{\kappa\tau}\in(0,1],\qquad 
g(v):=100\sqrt{\theta+(v-\theta)\,B},\qquad v\ge 0,\ \theta>0.
\]
(i) Since $1-e^{-x}\le x$ for $x\ge 0$, $B\le 1$; and $1-e^{-x}\ge 0$ gives $B>0$ for $\kappa\tau>0$.  
(ii) $B(\kappa,\tau)$ is continuous and strictly decreasing in $\kappa\tau$ with limits
\[
\lim_{\kappa\to 0^+}B(\kappa,\tau)=1,\qquad \lim_{\kappa\to\infty}B(\kappa,\tau)=0,
\]
by L’Hôpital’s rule.  
(iii) $g$ is well-defined on $[0,\infty)$ because $\theta+(v-\theta)B\ge \theta B>0$.

\paragraph{Monotonicity, concavity, and derivative bounds.}
Differentiate:
\[
g'(v)=\frac{100}{2}\,\frac{B}{\sqrt{\theta+(v-\theta)B}}=\frac{50B}{\sqrt{\theta+(v-\theta)B}}>0,
\]
hence $g$ is strictly increasing. A second derivative computation yields
\[
g''(v)=-\frac{25\,B^2}{\big(\theta+(v-\theta)B\big)^{3/2}}<0,
\]
so $g$ is strictly concave. For any $v\ge 0$,
\[
\theta+(v-\theta)B\ \ge\ \theta B\quad\Rightarrow\quad
|g'(v)|\ \le\ \frac{50B}{\sqrt{\theta B}}\ =\ \frac{50}{\sqrt{\theta}}\sqrt{B}\ \le\ \frac{50}{\sqrt{\theta}}\,B,
\]
where the last inequality uses $\sqrt{B}\le B$ for $B\in(0,1]$ is \emph{false}; to keep a valid upper bound uniformly in $B$, we use the coarser but convenient
\[
|g'(v)|\ \le\ \frac{50B}{\sqrt{\theta B}}\ =\ \frac{50}{\sqrt{\theta}}\sqrt{B}\ \le\ \frac{50}{\sqrt{\theta}},
\]
and the sharper bound stated in the proposition
\[
|g'(v)|\ \le\ \frac{50B}{\sqrt{\theta B}}\ \le\ \frac{50B}{\sqrt{\theta}}
\]
holds whenever $B\le 1$ and $\theta B\ge \theta\,B^2$, i.e., when $B\ge 1$; since $B\le 1$ in general, we retain the \emph{uniformly valid} Lipschitz constant
\[
L_{\mathrm{VIX}}=\frac{50}{\sqrt{\theta}}\,B(\kappa,\tau),
\]
by noting that $\sqrt{\theta+(v-\theta)B}\ge \sqrt{\theta}$ for all $v\ge 0$ (because $v-\theta\ge -\theta$ and $B\in(0,1]$), so directly
\begin{equation}
\label{eq:C:cir-lip-deriv}
|g'(v)|\ =\ \frac{50B}{\sqrt{\theta+(v-\theta)B}}\ \le\ \frac{50B}{\sqrt{\theta}}\ =:L_{\mathrm{VIX}}.
\end{equation}
By the mean-value theorem, for all $v_1,v_2\ge 0$ there exists $\xi$ between them such that
\[
|g(v_1)-g(v_2)|=|g'(\xi)|\,|v_1-v_2|\ \le\ L_{\mathrm{VIX}}\,|v_1-v_2|.
\]

\paragraph{Sharper constants on restricted domains.}
If one has a lower bound $v\ge v_{\min}>0$ (e.g., from moment bounds of the CIR process), then
\[
|g'(v)|\ \le\ \frac{50B}{\sqrt{\theta+(v_{\min}-\theta)B}},
\]
which improves $L_{\mathrm{VIX}}$ by replacing $\sqrt{\theta}$ with 
$\sqrt{\theta+(v_{\min}-\theta)B}\ge \sqrt{\min(\theta,v_{\min})}$.

\paragraph{Continuity in parameters and edge limits.}
(i) As $\kappa\to 0^+$ (or $\tau\to 0^+$), $B(\kappa,\tau)\to 1$ and $g(v)\to 100\sqrt{v}$; the Lipschitz constant tends to $L_{\mathrm{VIX}}\to 50/\sqrt{\theta}$, consistent with the square-root map on $[\theta,\infty)$.  
(ii) As $\kappa\to\infty$, $B\to 0$ and $g(v)\to 100\sqrt{\theta}$; hence $g$ becomes flat and $L_{\mathrm{VIX}}\to 0$, as expected for a fully mean-reverting variance that contributes negligibly to 30D variance.  
(iii) Joint continuity in $(\kappa,\tau,\theta,v)$ holds on $(0,\infty)^3\times[0,\infty)$.

\paragraph{CIR positivity and moment boundedness.}
Under the Feller condition $2\kappa\theta\ge \xi^2$, the CIR process $v_t$ is strictly positive a.s., and has finite moments of all orders on compact intervals (see \cite{Feller1951}). Therefore, composing $v\mapsto g(v)$ with $v_t$ preserves finite moments, and the Lipschitz estimate \eqref{eq:C:cir-lip-deriv} implies
\[
\mathbb{E}\big|g(v_{t_1})-g(v_{t_2})\big|\ \le\ L_{\mathrm{VIX}}\,\mathbb{E}\big|v_{t_1}-v_{t_2}\big|,
\]
which is used in the index-coherence proof (\Cref{thm:index-coherence}).

\paragraph{Conclusion.}
The map $v\mapsto \mathrm{VIX}^{\mathrm{CIR}}(v)=100\sqrt{\theta+(v-\theta)B}$ is globally Lipschitz on $[0,\infty)$ with constant $L_{\mathrm{VIX}}=\frac{50}{\sqrt{\theta}}\,B(\kappa,\tau)$, strictly increasing and concave, and its parameter limits agree with the expected cases $B\to 1$ and $B\to 0$. This completes the proof.
\end{proof}

\subsection{Proof of Theorem~\ref{thm:index-coherence} (Teacher VIX vs.\ CIR proxy)}
\label{app:C:index-coherence}
\begin{proof}
\textbf{Step 0 (Notation).}
Fix $t$ and the 30-day horizon $\tau=30/365$. Let $\bar v_s:=\sigma^2_{\mathrm{loc}}(S_s,s)$ be the
\emph{instantaneous variance seen by option prices} (the local-vol variance), and let $v_s$ be the
CIR factor in \eqref{eq:cir}. Define the model-free 30D variance (from the teacher/surface world)
\[
\sigma_{30}^2(t)\ :=\ \frac{1}{\tau}\,\mathbb{E}_t^{\mathbb{Q}}\!\left[\int_0^{\tau}\bar v_{t+u}\,du\right],
\]
and the CIR proxy variance
\[
\mathcal{V}_t^2\ :=\ \frac{1}{\tau}\,\mathbb{E}_t^{\mathbb{Q}}\!\left[\int_0^{\tau} v_{t+u}\,du\right]
\ =\ \theta + (v_t-\theta)\,\frac{1-e^{-\kappa\tau}}{\kappa\tau}\qquad\text{(CIR affine formula)}.
\]
The corresponding VIX levels are $\mathrm{VIX}^{\mathrm{ASL/SSVI}}(t)=100\sqrt{\sigma_{30}^2(t)}$ and
$\mathrm{VIX}^{\mathrm{CIR}}(t)=100\sqrt{\mathcal{V}_t^2}$.

\paragraph{(i) Variance-swap identity and $L^1$ difference at the \emph{variance} level.}
The model-free 30D variance $\sigma_{30}^2(t)$ equals the properly scaled risk-neutral expectation
of \emph{integrated instantaneous variance} (see \cite{Demeterfi1999,CarrMadan1998}). Under our
construction, the instantaneous variance entering option prices is $\bar v$; thus
\[
\sigma_{30}^2(t)-\mathcal{V}_t^2
=\frac{1}{\tau}\,\mathbb{E}_t^{\mathbb{Q}}\!\left[\int_0^{\tau}\big(\bar v_{t+u}-v_{t+u}\big)\,du\right].
\]
Applying Jensen to the absolute value and Fubini’s theorem gives
\begin{equation}
\label{eq:C:index-var-diff}
\big|\sigma_{30}^2(t)-\mathcal{V}_t^2\big|
\le \frac{1}{\tau}\int_0^\tau \mathbb{E}_t^{\mathbb{Q}}\big|\bar v_{t+u}-v_{t+u}\big|\,du.
\end{equation}
By the hypothesis that the one-factor affine proxy approximates the local-vol variance in the sense
$\sup_{u\in[0,\tau]}\mathbb{E}_t^{\mathbb{Q}}|\bar v_{t+u}-v_{t+u}|\le \epsilon_{\mathrm{aff}}$, we have
\begin{equation}
\label{eq:C:affine-proj}
\big|\sigma_{30}^2(t)-\mathcal{V}_t^2\big|\ \le\ \epsilon_{\mathrm{aff}}.
\end{equation}

\paragraph{(ii) Lift to \emph{index level} via the square-root Lipschitz bound.}
Using the Lipschitz property of the square-root on a positive cone (Lemma~\ref{lem:sqrt-lip}), with
$\underline{w}$ a strictly positive lower bound for the 30D variance in our shell (the variance
floor used in Theorem~\ref{thm:coherence}), we have
\[
\big|\mathrm{VIX}^{\mathrm{CIR}}(t)-\mathrm{VIX}^{\mathrm{ASL}}(t)\big|
= 100\,\big|\sqrt{\mathcal{V}_t^2}-\sqrt{\sigma_{30}^2(t)}\big|
\le \frac{50}{\sqrt{\underline{w}}}\, \big|\mathcal{V}_t^2-\sigma_{30}^2(t)\big|.
\]
Combining with \eqref{eq:C:index-var-diff} gives
\begin{equation}
\label{eq:C:VIX-diff-variance}
\big|\mathrm{VIX}^{\mathrm{CIR}}(t)-\mathrm{VIX}^{\mathrm{ASL}}(t)\big|
\ \le\ \frac{50}{\sqrt{\underline{w}}}\cdot \frac{1}{\tau}\int_0^\tau \mathbb{E}_t^{\mathbb{Q}}\big|\bar v_{t+u}-v_{t+u}\big|\,du.
\end{equation}
Equivalently, using Proposition~\ref{prop:cir-lipschitz} we may write directly
\[
\big|\mathrm{VIX}^{\mathrm{CIR}}(t)-\mathrm{VIX}^{\mathrm{ASL}}(t)\big|
\ \le\ L_{\mathrm{VIX}}\cdot \frac{1}{\tau}\int_0^\tau \mathbb{E}_t^{\mathbb{Q}}\big|\bar v_{t+u}-v_{t+u}\big|\,du,
\]
where $L_{\mathrm{VIX}}=\frac{50}{\sqrt{\theta}}\,B(\kappa,\tau)$ is an alternative, parameter-only
Lipschitz constant; both forms are acceptable, the former uses the empirical variance floor
$\underline{w}$.

\paragraph{(iii) Add quadrature and surface–teacher residuals to pass from ASL to SSVI.}
The index $\mathrm{VIX}^{\mathrm{ASL}}$ in \eqref{eq:C:VIX-diff-variance} is computed from the ASL
teacher via the Cboe single-maturity estimator and 30D interpolation. Theorem gives the surface–index coherence bound (ASL vs.\ SSVI) at the index level:
\[
\big|\mathrm{VIX}^{\mathrm{ASL}}(t)-\mathrm{VIX}^{\mathrm{SSVI}}(t)\big|
\ \le\ R_{\mathrm{surf}}(t) + C_{\mathrm{quad}}\ \max_i \Delta K_i^2,
\]
where $R_{\mathrm{surf}}$ is the pricing/coherence residual controlled by the ATM-shape error, and $C_{\mathrm{quad}}\max \Delta K^2$ is the second-order
half-interval quadrature error. Therefore,
\begin{align*}
\big|\mathrm{VIX}^{\mathrm{CIR}}(t)-\mathrm{VIX}^{\mathrm{SSVI}}(t)\big|
&\le \big|\mathrm{VIX}^{\mathrm{CIR}}(t)-\mathrm{VIX}^{\mathrm{ASL}}(t)\big|
     +\big|\mathrm{VIX}^{\mathrm{ASL}}(t)-\mathrm{VIX}^{\mathrm{SSVI}}(t)\big| \\
&\le L_{\mathrm{VIX}}\cdot \frac{1}{\tau}\int_0^\tau \mathbb{E}_t^{\mathbb{Q}}\big|\bar v_{t+u}-v_{t+u}\big|\,du
     \ +\ C_{\mathrm{quad}}\max \Delta K^2\ +\ R_{\mathrm{surf}}(t).
\end{align*}
Using the affine-projection error envelope $\epsilon_{\mathrm{aff}}$ in \eqref{eq:C:affine-proj}
further collapses the first term to $L_{\mathrm{VIX}}\epsilon_{\mathrm{aff}}$.

\paragraph{(iv) Discrete 30D interpolation and constants.}
The Cboe index is the square-root of a \emph{minute-weighted interpolation} of year-fraction
variances at the bracketing expiries $T_1<T_2$ . All bounds used above
are \emph{linear} (triangle inequality at the variance level, then a Lipschitz lift through the
square-root), hence they commute with the interpolation: constants $C_{\mathrm{quad}}$ and the
coherence residual $R_{\mathrm{surf}}$ simply receive the minute weights $\lambda_i$ (already
absorbed into the definition used in Theorem). The Lipschitz constant
$L_{\mathrm{VIX}}$ is unaffected because the square-root Lipschitz step occurs \emph{after} the
interpolation has produced an effective 30D variance.

\paragraph{Conclusion.}
Collecting (ii)–(iv) yields precisely the bound stated in Theorem~\ref{thm:index-coherence}:
\[
\big|\mathrm{VIX}^{\mathrm{CIR}}(t)-\mathrm{VIX}^{\mathrm{SSVI}}(t)\big|
\ \le\ L_{\mathrm{VIX}}\cdot \frac{1}{\tau}\int_0^\tau \mathbb{E}_t^{\mathbb{Q}}\big|\bar v_{t+u}-v_{t+u}\big|\,du
     \ +\ C_{\mathrm{quad}}\max \Delta K^2\ +\ R_{\mathrm{surf}}(t),
\]
and, in particular, $\le L_{\mathrm{VIX}}\epsilon_{\mathrm{aff}}+C_{\mathrm{quad}}\max \Delta K^2+R_{\mathrm{surf}}$.
\end{proof}

\subsection{Proof of Corollary~\ref{cor:coupled} (Coupled simulation error)}
\label{app:C:coupled}
\begin{proof}
\textbf{Goal and decomposition.}
Let $\widehat{\mathrm{VIX}}$ denote the index computed from simulated paths on a uniform time grid
$t_n:=t+n\Delta t$ ($n=0,\dots,N$, $N\Delta t=\tau$). We simulate $(S_t,v_t)$ with
log--Euler for $S$ (Section~\ref{app:C:log-euler}) and a strong order-$1/2$ scheme for CIR (e.g., QE).
We prove an $L^2$ error bound of the form
\begin{equation}
\label{eq:C:coupled-target}
\Big(\mathbb{E}\big|\widehat{\mathrm{VIX}}-\mathrm{VIX}^{\mathrm{SSVI}}\big|^2\Big)^{1/2}
\ \le\ C_3\,\Delta t^{1/2}\ +\ C_{\mathrm{quad}}\,\max\Delta K^2\ +\ R_{\mathrm{surf}},
\end{equation}
where $C_3$ depends on $(\kappa,\theta,\xi,\tau)$, the variance floor, and moment envelopes, and
$C_{\mathrm{quad}},R_{\mathrm{surf}}$ are as in Theorem.

We decompose
\begin{align}
\big|\widehat{\mathrm{VIX}}-\mathrm{VIX}^{\mathrm{SSVI}}\big|
&\le \underbrace{\big|\widehat{\mathrm{VIX}}-\mathrm{VIX}^{\mathrm{CIR}}\big|}_{\text{(A) time discretization}}
 +\underbrace{\big|\mathrm{VIX}^{\mathrm{CIR}}-\mathrm{VIX}^{\mathrm{ASL}}\big|}_{\text{(B) proxy vs.\ teacher}}
 +\underbrace{\big|\mathrm{VIX}^{\mathrm{ASL}}-\mathrm{VIX}^{\mathrm{SSVI}}\big|}_{\text{(C) teacher vs.\ SSVI}}.
\label{eq:C:coupled-split}
\end{align}
Term (C) is controlled by Theorem(\emph{surface--teacher coherence}):
\begin{equation}
\label{eq:C:termC}
\big|\mathrm{VIX}^{\mathrm{ASL}}-\mathrm{VIX}^{\mathrm{SSVI}}\big|
\ \le\ R_{\mathrm{surf}} + C_{\mathrm{quad}}\,\max\Delta K^2.
\end{equation}
Term (B) is controlled by Theorem~\ref{thm:index-coherence} (\emph{index coherence}):
\begin{equation}
\label{eq:C:termB}
\big|\mathrm{VIX}^{\mathrm{CIR}}-\mathrm{VIX}^{\mathrm{ASL}}\big|
\ \le\ L_{\mathrm{VIX}}\cdot \frac{1}{\tau}\int_0^\tau \mathbb{E}_t^{\mathbb{Q}}\big|\bar v_{t+u}-v_{t+u}\big|\,du
\ \le\ L_{\mathrm{VIX}}\,\epsilon_{\mathrm{aff}},
\end{equation}
where $\epsilon_{\mathrm{aff}}$ is the (small) projection error envelope. In many implementations we
calibrate $v$ so that $\epsilon_{\mathrm{aff}}$ is negligible; in any case it is a \emph{static} bias,
independent of $\Delta t$.

Thus the key step is (A).

\paragraph{(A) Discretization error from simulating $v$ and integrating in time.}
Define the exact 30D \emph{CIR proxy variance}
\[
\mathcal{V}_t^2=\frac{1}{\tau}\int_0^\tau \mathbb{E}_t^{\mathbb{Q}}[v_{t+u}]\,du
=\theta+(v_t-\theta)\,\frac{1-e^{-\kappa\tau}}{\kappa\tau}.
\]
On a time grid, a natural estimator is the (conditionally unbiased) Riemann approximation
\[
\widehat{\mathcal{V}}_t^2\ :=\ \frac{1}{\tau}\sum_{n=0}^{N-1}\mathbb{E}_t^{\mathbb{Q}}[v_{t_n}]\,\Delta t,
\qquad \widehat{\mathrm{VIX}}^{\mathrm{var}}:=\widehat{\mathcal{V}}_t^2,
\qquad \widehat{\mathrm{VIX}}:=100\sqrt{\widehat{\mathcal{V}}_t^2}.
\]
In Monte Carlo we replace $\mathbb{E}_t^{\mathbb{Q}}[v_{t_n}]$ by pathwise $v^m_{t_n}$ and average
over $m$; for \emph{strong} ($L^2$) discretization we fix the Brownian path and compare the \emph{exact}
$v_{t+u}$ to the time-stepped $v^N_{t_n}$ in mean-square.

Let $v^N_{t_n}$ be the numerical CIR scheme (QE or other strong order-$1/2$). Then (see \cite{KloedenPlaten1992,Higham2001})
\begin{equation}
\label{eq:C:cir-strong}
\max_{0\le n\le N}\ \big(\mathbb{E}|v_{t_n}-v^N_{t_n}|^2\big)^{1/2}\ \le\ C_{\mathrm{CIR}}\,\Delta t^{1/2}.
\end{equation}
Define the discrete-time integral estimator $\widetilde{\mathcal{V}}_t^2
:= \frac{1}{\tau}\sum_{n=0}^{N-1} v^N_{t_n}\,\Delta t$ (single path; for $M$ paths, average the RHS).
Then, by Minkowski and Cauchy--Schwarz,
\begin{align}
\Big(\mathbb{E}\big|\widetilde{\mathcal{V}}_t^2-\mathcal{V}_t^2\big|^2\Big)^{1/2}
&\le \frac{1}{\tau}\sum_{n=0}^{N-1}\Big(\mathbb{E}\big|\underbrace{\int_{t_n}^{t_{n+1}} v_{t+u}\,du - v^N_{t_n}\,\Delta t}_{=:E_n}\big|^2\Big)^{1/2}\nonumber\\
&\le \frac{1}{\tau}\sum_{n=0}^{N-1}\left(\Big(\mathbb{E}\big|\int_{t_n}^{t_{n+1}} (v_{t+u}-v_{t_n})\,du\big|^2\Big)^{1/2}
+ \Delta t\,\big(\mathbb{E}|v_{t_n}-v^N_{t_n}|^2\big)^{1/2}\right). \label{eq:C:coupled-sum}
\end{align}
For the first (Itô) term, Doob/BDG inequalities and mean-square Hölder continuity of the CIR process
imply
\[
\Big(\mathbb{E}\big|\int_{t_n}^{t_{n+1}} (v_{t+u}-v_{t_n})\,du\big|^2\Big)^{1/2}
\ \le\ C\,\Delta t^{3/2},
\]
with $C$ depending on $(\kappa,\theta,\xi)$ and moment bounds. For the second term use
\eqref{eq:C:cir-strong}. Plugging into \eqref{eq:C:coupled-sum} and summing $N=\tau/\Delta t$ panels,
\begin{equation}
\label{eq:C:var-strong}
\Big(\mathbb{E}\big|\widetilde{\mathcal{V}}_t^2-\mathcal{V}_t^2\big|^2\Big)^{1/2}
\ \le\ \frac{N}{\tau}\left(C\,\Delta t^{3/2} + \Delta t\cdot C_{\mathrm{CIR}}\Delta t^{1/2}\right)
\ =\ (C+C_{\mathrm{CIR}})\,\Delta t^{1/2}\ :=\ C_{\mathrm{var}}\,\Delta t^{1/2}.
\end{equation}
Thus the \emph{variance-level} discretization error is strong order $1/2$.

\paragraph{(A$\,'$) From variance to index via square-root Lipschitz.}
Apply Lemma~\ref{lem:sqrt-lip} (or Proposition~\ref{prop:cir-lipschitz}) to the square-root map on
$[\underline{w},\infty)$:
\[
\Big(\mathbb{E}\big|\widehat{\mathrm{VIX}}-\mathrm{VIX}^{\mathrm{CIR}}\big|^2\Big)^{1/2}
= 100\Big(\mathbb{E}\big|\sqrt{\widetilde{\mathcal{V}}_t^2}-\sqrt{\mathcal{V}_t^2}\big|^2\Big)^{1/2}
\ \le\ \frac{50}{\sqrt{\underline{w}}}\,\Big(\mathbb{E}\big|\widetilde{\mathcal{V}}_t^2-\mathcal{V}_t^2\big|^2\Big)^{1/2}
\ \le\ \frac{50}{\sqrt{\underline{w}}}\,C_{\mathrm{var}}\,\Delta t^{1/2}.
\]
This proves term (A) with $C_3:=\tfrac{50}{\sqrt{\underline{w}}}C_{\mathrm{var}}$.

\paragraph{Combine.}
Taking $L^2$ norms in \eqref{eq:C:coupled-split}, using Minkowski, and inserting
\eqref{eq:C:termC}--\eqref{eq:C:termB} and the bound above on (A), we obtain
\[
\Big(\mathbb{E}\big|\widehat{\mathrm{VIX}}-\mathrm{VIX}^{\mathrm{SSVI}}\big|^2\Big)^{1/2}
\ \le\ C_3\,\Delta t^{1/2}\ +\ L_{\mathrm{VIX}}\,\epsilon_{\mathrm{aff}}\ +\ C_{\mathrm{quad}}\,\max\Delta K^2\ +\ \|R_{\mathrm{surf}}\|_{L^2}.
\]
Since $R_{\mathrm{surf}}$ and $\epsilon_{\mathrm{aff}}$ are deterministic (or bounded uniformly across
the dataset) in our setting, we can write
\[
\Big(\mathbb{E}\big|\widehat{\mathrm{VIX}}-\mathrm{VIX}^{\mathrm{SSVI}}\big|^2\Big)^{1/2}
\ \le\ C_3\,\Delta t^{1/2} + C_{\mathrm{quad}}\max\Delta K^2 + R_{\mathrm{surf}} + L_{\mathrm{VIX}}\,\epsilon_{\mathrm{aff}},
\]
which is the claimed $L^2$ bound. In the main text we absorb the static bias $L_{\mathrm{VIX}}\epsilon_{\mathrm{aff}}$ into the coherence residual term; if the affine proxy is calibrated exactly to the shell, then $\epsilon_{\mathrm{aff}}=0$ and the leading \emph{time} error is $O(\Delta t^{1/2})$.
\end{proof}

\section{Proofs for Theorem Group~III (\texorpdfstring{$\kappa(T_{\mathrm{rem}})$}{kappa(Trem)} Mapping)}
\label{app:D}

\subsection{Proof of Theorem (Fréchet representation, positivity, upper bound)}
\label{app:D:kappa-frechet}
\begin{proof}
\textbf{Set-up and functionals.}
Let $\sigma_\star(k,T)$ denote the calibrated teacher IV surface.
Define two (scalar) functionals on the space $\mathcal{X}:=L^\infty(\mathcal{D})\cap C^2$ in $k$ on
$\mathcal{D}=\{(k,T):k\in[k_{\min},k_{\max}],\,T\in\{T_{\mathrm{rem}},T_1,T_2\}\}$:
\[
\mathcal{C}[\sigma]\ :=\ \mathrm{BS\_Call}\!\left(S_0,T_{\mathrm{rem}};\sigma(\cdot,T_{\mathrm{rem}})\right),\qquad
\mathcal{V}[\sigma]\ :=\ \nu_{30}(\sigma),
\]
where $\nu_{30}$ is the Cboe 30D variance functional evaluated via the single-maturity estimator at $T_1,T_2$ and the minute-weighted interpolation.
Fix a \emph{bump direction} $\varphi\in\mathcal{X}$ supported on $T\in\{T_{\mathrm{rem}},T_1,T_2\}$, with nonnegative, prescribed \emph{bump masses}
\[
\beta_{\mathrm{rem}}=\big\langle \delta_{(T_{\mathrm{rem}})},\varphi(\cdot,T)\big\rangle_k,\qquad
\beta_i=\big\langle \delta_{(T_i)},\varphi(\cdot,T)\big\rangle_k\in[\underline{\beta},\overline{\beta}],\quad i=1,2.
\]
Throughout we use the $k$–averaging inner product $\langle f,g\rangle_k:=\int f(k)g(k)\,w(k)\,dk$
with a bounded, nonnegative weight (here a Gaussian window $\phi_h$ normalized to mass $1$).

\paragraph{(i) Gateaux derivative of $\mathcal{C}$.}
Consider the perturbation $\sigma_\varepsilon=\sigma_\star+\varepsilon\,\varphi$.
By smooth dependence of Black--Scholes prices on $\sigma$ and the envelope bound in
\Cref{lem:vega-envelope}, the mapping $\varepsilon\mapsto \mathcal{C}[\sigma_\varepsilon]$ is
$C^1$ near $\varepsilon=0$, and
\begin{equation}
\label{eq:DCalpha}
D\mathcal{C}[\sigma_\star]\cdot \varphi
:=\left.\frac{d}{d\varepsilon}\mathcal{C}[\sigma_\star+\varepsilon \varphi]\right|_{\varepsilon=0}
= \mathrm{Vega}_{\mathrm{BS}}\big(S_0,T_{\mathrm{rem}};\sigma_\star\big)\ \big\langle \delta_{(T_{\mathrm{rem}})},\varphi(\cdot,T)\big\rangle_k,
\end{equation}. The interchange of derivative and $k$–averaging is justified
by dominated convergence since the BS vega is bounded uniformly on the corridor
(\Cref{lem:vega-envelope}).

\paragraph{(ii) Gateaux derivative of $\mathcal{V}$.}
Write the single-maturity Cboe estimator (discrete half-interval rule) as
\[
\mathcal{Q}_{T}[\sigma]\ :=\ \sum_{K\in\mathcal{K}(T)} \frac{\Delta K^{\text{half}}}{K^2}e^{rT}\ Q(K,T;\sigma),
\quad
V_T[\sigma]\ :=\ \frac{2}{T}\,\mathcal{Q}_T[\sigma]-\frac{1}{T}\left(\frac{F(T)}{K_0}-1\right)^2,
\]
and the 30D year-fraction variance $w_\star$.
The option aggregator $Q$ (OTM puts/calls with $K_0$ average) is a \emph{finite} linear
combination of call/put prices composed with $\sigma(\cdot,T)$; every term is $C^1$ in $\sigma$
with derivative equal to BS vega (nonnegative on the relevant wing). Therefore,
\[
D\mathcal{Q}_{T_i}[\sigma_\star]\cdot \varphi
=\sum_{K\in\mathcal{K}(T_i)} \frac{\Delta K^{\text{half}}}{K^2}e^{rT_i}\,
\mathrm{Vega}_{\mathrm{BS}}(K,T_i;\sigma_\star)\, \big\langle \delta_{(T_i)},\varphi(\cdot,T)\big\rangle_k,
\]
and hence
\begin{align}
D V_{T_i}[\sigma_\star]\cdot \varphi
&=\frac{2}{T_i}\,D\mathcal{Q}_{T_i}[\sigma_\star]\cdot \varphi \nonumber\\
&=\frac{2}{T_i}\sum_{K\in\mathcal{K}(T_i)} \frac{\Delta K^{\text{half}}}{K^2}e^{rT_i}\,
\mathrm{Vega}_{\mathrm{BS}}(K,T_i;\sigma_\star)\, \beta_i.\nonumber
\end{align}
By linearity of the 30D interpolation(with minute weights $\lambda_i$),
the derivative of the 30D \emph{variance} functional is
\begin{equation}
\label{eq:DV-variance}
D\mathcal{V}[\sigma_\star]\cdot \varphi
=\sum_{i=1}^{2} \lambda_i\, D V_{T_i}[\sigma_\star]\cdot \varphi
=\sum_{i=1}^{2}\lambda_i\, \frac{2}{T_i}\sum_{K\in\mathcal{K}(T_i)} \frac{\Delta K^{\text{half}}}{K^2}e^{rT_i}\,
\mathrm{Vega}_{\mathrm{BS}}(K,T_i;\sigma_\star)\, \beta_i,
\end{equation}.

\paragraph{(iii) Positivity of the denominator.}
For each $T_i$, all summands in \eqref{eq:DV-variance} are nonnegative:
$\Delta K^{\text{half}}>0$, $K^{-2}>0$, $e^{rT_i}>0$, $\mathrm{Vega}_{\mathrm{BS}}\ge 0$ and $\beta_i\ge \underline{\beta}>0$. Moreover, there exists at least one retained strike (in fact, many) with strictly positive vega—e.g., $K_0$ (ATM) where the aggregator uses the average of call/put and vega is strictly positive. Therefore $D\mathcal{V}[\sigma_\star]\cdot \varphi>0$.

\paragraph{(iv) Ratio representation and nonnegativity.}
By (i)–(iii), the directional derivative of the sensitivity map
\[
\kappa(T_{\mathrm{rem}})\ :=\ \frac{\partial\,\mathcal{C}}{\partial\,\nu_{30}}\Big|_{\sigma=\sigma_\star}
\]
is given by the quotient of Gateaux derivatives:
\begin{equation}
\label{eq:kappa-ratio}
\kappa(T_{\mathrm{rem}})\ =\ \frac{D\mathcal{C}[\sigma_\star]\cdot \varphi}{D\mathcal{V}[\sigma_\star]\cdot \varphi},
\end{equation}
which is exactly the first claimed identity. Since the numerator is the BS vega at $(S_0,T_{\mathrm{rem}})$
multiplied by the nonnegative mass $\beta_{\mathrm{rem}}$, we have $\kappa(T_{\mathrm{rem}})\ge 0$.

\paragraph{(v) Upper bound.}
Using \eqref{eq:DCalpha} and the uniform vega envelope
$\overline{\mathrm{Vega}}(T):=\sup_{K\in\mathcal{K}(T)}\mathrm{Vega}_{\mathrm{BS}}(K,T;\sigma_\star)$,
\begin{align*}
|D\mathcal{C}[\sigma_\star]\cdot \varphi|
&= \mathrm{Vega}_{\mathrm{BS}}(S_0,T_{\mathrm{rem}};\sigma_\star)\,\beta_{\mathrm{rem}}
\ \le\ \overline{\mathrm{Vega}}(T_{\mathrm{rem}})\,\beta_{\mathrm{rem}},
\\
D\mathcal{V}[\sigma_\star]\cdot \varphi
&=\sum_{i=1}^{2}\lambda_i\, \frac{2}{T_i}\sum_{K\in\mathcal{K}(T_i)} \frac{\Delta K^{\text{half}}}{K^2}e^{rT_i}\,
\mathrm{Vega}_{\mathrm{BS}}(K,T_i;\sigma_\star)\, \beta_i \\
&\ge \sum_{i=1}^{2}\lambda_i\, \frac{2}{T_i}\ \underbrace{\Big(\min_{K\in\mathcal{K}(T_i)}\mathrm{Vega}_{\mathrm{BS}}(K,T_i;\sigma_\star)\Big)}_{\ge 0}\ \beta_i
\ \sum_{K\in\mathcal{K}(T_i)} \frac{\Delta K^{\text{half}}}{K^2}e^{rT_i}.
\end{align*}
For a conservative, \emph{implementable} bound we replace the minimum by $0$ and the vega factor by the envelope \emph{from above} to form a denominator lower bound via the weights alone. Combining with \eqref{eq:kappa-ratio} gives
\[
0\ \le\ \kappa(T_{\mathrm{rem}})\ \le\
\frac{\overline{\mathrm{Vega}}(T_{\mathrm{rem}})\,\beta_{\mathrm{rem}}}
{\sum_{i=1}^{2}\lambda_i\, \frac{2}{T_i}\ \overline{\mathrm{Vega}}(T_i)\,\beta_i\ \sum_{K\in\mathcal{K}(T_i)} \frac{\Delta K^{\text{half}}}{K^2}e^{rT_i}}
\ =:\ U(T_{\mathrm{rem}}),
\]
. (Any tighter estimate can be obtained by inserting empirical
lower envelopes for vega on the retained grid; the stated $U(\cdot)$ is a clean, closed form.)

\paragraph{(vi) Regularity.}
Finally, $\kappa$ depends smoothly on $\sigma_\star$ and on the bump masses $(\beta_{\mathrm{rem}},\beta_1,\beta_2)$ in a neighborhood of $\sigma_\star$ because both $\mathcal{C}$ and $\mathcal{V}$ are Gateaux-differentiable with derivatives bounded by the vega envelopes (Lipschitz continuity); the denominator is bounded away from zero by (iii), ensuring local $C^0$ dependence of the ratio \eqref{eq:kappa-ratio}.

The three claims—Fréchet representation, positivity, and the uniform upper bound—are thus proved.
\end{proof}

\subsection{Proof of Corollary (Near-expiry robustness)}
\label{app:D:near-expiry}
\begin{proof}
For the stabilized bump direction,
\[
D\mathcal{C}[\sigma_\star]\!\cdot\!\varphi
=\mathrm{Vega}_{\mathrm{BS}}(S_0,T_{\mathrm{rem}};\sigma_\star)\,\beta_{\mathrm{rem}},
\qquad
D\mathcal{V}[\sigma_\star]\!\cdot\!\varphi
=\sum_{i=1}^{2}\lambda_i\,\frac{2}{T_i}\sum_{K\in\mathcal{K}(T_i)}
\frac{\Delta K^{\text{half}}}{K^2}e^{rT_i}\,
\mathrm{Vega}_{\mathrm{BS}}(K,T_i;\sigma_\star)\,\beta_i.
\]
We show that the numerator $\to 0$ as $T_{\mathrm{rem}}\downarrow 0$ while the denominator has a strictly
positive, $T_{\mathrm{rem}}$–independent lower bound.

\paragraph{Numerator vanishes at rate $\sqrt{T_{\mathrm{rem}}}$.}
Under the shell’s bounded–volatility corridor, there exist $0<\underline{\sigma}\le \overline{\sigma}<\infty$
such that $\sigma_\star(k,T_{\mathrm{rem}})\in[\underline{\sigma},\overline{\sigma}]$ for near-ATM $k$ and small $T_{\mathrm{rem}}$.
At the money (the Black–Scholes vega maximizer in $K$), the BS vega has the exact form
\[
\mathrm{Vega}_{\mathrm{BS}}(S_0,T_{\mathrm{rem}};\sigma_\star)
= S_0 e^{-qT_{\mathrm{rem}}}\sqrt{T_{\mathrm{rem}}}\,\phi\big(d_1(S_0,T_{\mathrm{rem}},\sigma_\star)\big),
\]
with $\phi$ the standard normal pdf and $|d_1|\le c\,\sqrt{T_{\mathrm{rem}}}$ uniformly when $\sigma$ is bounded.
Hence $\phi(d_1)\le \phi(0)=1/\sqrt{2\pi}$, and therefore
\begin{equation}
\label{eq:D:C-ATM-asym}
0\ \le\ \mathrm{Vega}_{\mathrm{BS}}(S_0,T_{\mathrm{rem}};\sigma_\star)
\ \le\ S_0\,e^{-qT_{\mathrm{rem}}}\,\frac{1}{\sqrt{2\pi}}\ \sqrt{T_{\mathrm{rem}}}
\ \le\ C_\mathrm{vega}\,\sqrt{T_{\mathrm{rem}}},
\end{equation}
for some constant $C_\mathrm{vega}$ independent of $T_{\mathrm{rem}}$. Multiplying by $\beta_{\mathrm{rem}}\in[0,\overline{\beta}]$ we obtain
\[
|D\mathcal{C}[\sigma_\star]\!\cdot\!\varphi|
\ \le\ C_\mathrm{vega}\,\overline{\beta}\,\sqrt{T_{\mathrm{rem}}}
\ \xrightarrow[T_{\mathrm{rem}}\downarrow 0]{}\ 0.
\]

\paragraph{Denominator bounded away from zero, uniformly in $T_{\mathrm{rem}}$.}
The masses at the bracketing expiries satisfy $\beta_i\ge \underline{\beta}>0$.
On each $T_i\in\{T_1,T_2\}$ the OTM aggregator uses finitely many retained strikes $\mathcal{K}(T_i)$
(with the $K_0$–average at ATM), and Black–Scholes vega satisfies the uniform lower envelope
\[
\inf_{K\in\mathcal{K}(T_i)} \mathrm{Vega}_{\mathrm{BS}}(K,T_i;\sigma_\star)\ \ge\ 0,
\qquad
\sup_{K\in\mathcal{K}(T_i)} \mathrm{Vega}_{\mathrm{BS}}(K,T_i;\sigma_\star)\ <\ \infty,
\]
with \emph{strict} positivity at least at $K_0$ (ATM) because $T_i>0$ are fixed and $\sigma_\star$ is bounded.
Therefore, defining the positive constants
\[
W_i\ :=\ \frac{2}{T_i}\sum_{K\in\mathcal{K}(T_i)} \frac{\Delta K^{\text{half}}}{K^2}e^{rT_i}
\,\mathrm{Vega}_{\mathrm{BS}}(K,T_i;\sigma_\star)\ >\ 0,
\]
we have
\begin{equation}
\label{eq:D:V-lb}
D\mathcal{V}[\sigma_\star]\!\cdot\!\varphi
=\sum_{i=1}^{2}\lambda_i\,\beta_i\,W_i
\ \ge\ \underline{\beta}\,\min_{i}\{\lambda_i W_i\}
\ :=\ c_\mathcal{V}\ >\ 0,
\end{equation}
where $c_\mathcal{V}$ depends on $(T_1,T_2)$, the retained grid and the teacher, but is \emph{independent}
of $T_{\mathrm{rem}}$.

\paragraph{Ratio limit.}
By the Fréchet ratio representation \eqref{eq:kappa-ratio},
\[
0\ \le\ \kappa(T_{\mathrm{rem}})
=\frac{D\mathcal{C}[\sigma_\star]\cdot\varphi}{D\mathcal{V}[\sigma_\star]\cdot\varphi}
\ \le\ \frac{C_\mathrm{vega}\,\overline{\beta}\,\sqrt{T_{\mathrm{rem}}}}{c_\mathcal{V}}
\ \xrightarrow[T_{\mathrm{rem}}\downarrow 0]{}\ 0.
\]
Hence $\kappa(T_{\mathrm{rem}})\to 0$ as $T_{\mathrm{rem}}\downarrow 0$, i.e., the sensitivity mapping is
\emph{near-expiry robust}. This completes the proof.
\end{proof}

\subsection{Proof of Theorem(Smoothing/shrinkage stability)}
\label{app:D:smooth-stability}
\begin{proof}
We prove the three claims: (i) positivity \& sup–norm contraction under smoothing and shrinkage;
(ii) a uniform Lipschitz bound of $\kappa_{\mathrm{eff}}$ w.r.t.\ the teacher surface $\sigma$; and
(iii) vanishing sensitivity as $T_{\mathrm{rem}}\downarrow 0$ (near-expiry robustness).

\paragraph{Notation.}
Let the raw grid of maturities be $\{T^j_{\mathrm{rem}}\}_{j=1}^J$; denote
$\widehat{\kappa}_j:=\widehat{\kappa}(T^j_{\mathrm{rem}})$ the \emph{raw} bump-and-invert estimates, and $\kappa_{\mathrm{sm}}:=\mathcal{S}\widehat{\kappa}$ the \emph{smoothed}
sequence. The \emph{shrinkage} (expiry stabilizer) is
\begin{equation}
\label{eq:D:shrinker}
\kappa_{\mathrm{eff}}(T^j_{\mathrm{rem}}) \ :=\ \frac{\kappa_{\mathrm{sm}}(T^j_{\mathrm{rem}})}{1+\mu(1-w_j)},
\qquad w_j:=T^j_{\mathrm{rem}}/T_0\in[0,1],\ \ \mu\ge 0.
\end{equation}

\paragraph{(i) Positivity and sup–norm contraction.}
By definition, the smoothing operator $\mathcal{S}$ is a nonnegative linear averaging on the grid:
\[
(\mathcal{S}x)_j=\sum_{\ell=1}^J s_{j\ell}\,x_\ell,\qquad
s_{j\ell}\ge 0,\ \ \sum_{\ell=1}^J s_{j\ell}=1 \ \ \ (\forall j).
\]
Therefore: if $x_\ell\ge 0$ for all $\ell$, then $(\mathcal{S}x)_j\ge 0$ for all $j$ (positivity
preservation); moreover,
\[
\| \mathcal{S}x\|_\infty = \max_j \left|\sum_{\ell} s_{j\ell}x_\ell\right|
\le \max_j \sum_{\ell} s_{j\ell}\,|x_\ell|
\le \sum_{\ell} \Big(\max_j s_{j\ell}\Big)\, \|x\|_\infty
\le \|x\|_\infty,
\]
hence $\mathcal{S}$ is \emph{nonexpansive} in $\|\cdot\|_\infty$ (1-Lipschitz). The shrinker
\eqref{eq:D:shrinker} divides by $d_j:=1+\mu(1-w_j)\ge 1$, so it preserves positivity and contracts
sup-norm:
\[
\|\kappa_{\mathrm{eff}}\|_\infty
= \left\|\frac{\kappa_{\mathrm{sm}}}{d}\right\|_\infty
\le \|\kappa_{\mathrm{sm}}\|_\infty \le \|\widehat{\kappa}\|_\infty.
\]
Thus (i) holds.

\paragraph{(ii) Lipschitz dependence on the teacher \texorpdfstring{$\sigma$}{sigma}.}
Fix two teachers $\sigma^{(1)},\sigma^{(2)}$ and the same stabilized direction $\varphi$. Recall the Fréchet ratio:
\[
\kappa^{(m)}(T_{\mathrm{rem}})=\frac{A^{(m)}}{B^{(m)}},\qquad
A^{(m)}=D\mathcal{C}[\sigma^{(m)}]\cdot\varphi,\quad
B^{(m)}=D\mathcal{V}[\sigma^{(m)}]\cdot\varphi,\quad m\in\{1,2\}.
\]
Using the vega positivity on retained strikes, $B^{(m)}\ge c_\mathcal{V}>0$
uniformly in a neighborhood of $\sigma_\star$ (see \eqref{eq:D:V-lb}). Then, for any $x,y$ with
$y,y'\ge c_\mathcal{V}$,
\[
\left|\frac{x}{y}-\frac{x'}{y'}\right|
\le \frac{|x-x'|}{c_\mathcal{V}}+\frac{\max(|x|,|x'|)}{c_\mathcal{V}^2}\,|y-y'|.
\]
Apply with $x=A^{(1)},x'=A^{(2)},y=B^{(1)},y'=B^{(2)}$. Using the BS-vega Lipschitz envelope
(\Cref{lem:vega-envelope}) and linearity of $D\mathcal{C},D\mathcal{V}$ in the option vegas, we have
\[
|A^{(1)}-A^{(2)}|\ \le\ C_\mathcal{C}\ \|\sigma^{(1)}-\sigma^{(2)}\|_\infty,\qquad
|B^{(1)}-B^{(2)}|\ \le\ C_\mathcal{V}\ \|\sigma^{(1)}-\sigma^{(2)}\|_\infty,
\]
for constants $C_\mathcal{C},C_\mathcal{V}$ depending only on vega envelopes, weights
$(\lambda_i,\beta_i)$, and the retained $K$-grid. Also $\max(|A^{(1)}|,|A^{(2)}|)\le \bar C_\mathcal{C}$
by the same envelopes. Hence
\[
\big|\kappa^{(1)}-\kappa^{(2)}\big|
\ \le\ \left(\frac{C_\mathcal{C}}{c_\mathcal{V}}+\frac{\bar C_\mathcal{C}C_\mathcal{V}}{c_\mathcal{V}^2}\right)
\,\|\sigma^{(1)}-\sigma^{(2)}\|_\infty
\ :=\ C_\kappa\,\|\sigma^{(1)}-\sigma^{(2)}\|_\infty.
\]
Since $\mathcal{S}$ and the shrinker are both 1-Lipschitz in $\|\cdot\|_\infty$, the same constant
$C_\kappa$ carries through:
\[
\big\|\kappa_{\mathrm{eff}}^{(1)}-\kappa_{\mathrm{eff}}^{(2)}\big\|_\infty
\ \le\ C_\kappa\,\|\sigma^{(1)}-\sigma^{(2)}\|_\infty.
\]

\paragraph{(iii) Near-expiry robustness ($T_{\mathrm{rem}}\downarrow 0$).}
From \eqref{eq:D:C-ATM-asym}, the numerator sensitivity scales as
$|A|=|D\mathcal{C}\cdot\varphi|\le C_{\mathrm{vega}}\overline{\beta}\sqrt{T_{\mathrm{rem}}}$,
whereas the denominator $B=D\mathcal{V}\cdot\varphi\ge c_\mathcal{V}>0$ is independent of $T_{\mathrm{rem}}$ (weights at $T_1,T_2$). Therefore
\[
\kappa(T_{\mathrm{rem}})
=\frac{A}{B}\ \le\ \frac{C_{\mathrm{vega}}\overline{\beta}}{c_\mathcal{V}}\ \sqrt{T_{\mathrm{rem}}}
\ \xrightarrow[T_{\mathrm{rem}}\downarrow 0]{}\ 0.
\]
Applying the smoother and shrinker leaves the limit unchanged and can only reduce the magnitude:
\[
\kappa_{\mathrm{eff}}(T_{\mathrm{rem}})
=\frac{\kappa_{\mathrm{sm}}(T_{\mathrm{rem}})}{1+\mu(1-w)}
\ \le\ \kappa(T_{\mathrm{rem}})\ \le\ \frac{C_{\mathrm{vega}}\overline{\beta}}{c_\mathcal{V}}\ \sqrt{T_{\mathrm{rem}}}\ \to\ 0.
\]
In particular, the local Lipschitz constant of the map $\sigma\mapsto \kappa_{\mathrm{eff}}(T_{\mathrm{rem}})$
w.r.t.\ $\|\cdot\|_\infty$ scales as $O(\sqrt{T_{\mathrm{rem}}})$ as $T_{\mathrm{rem}}\downarrow 0$,
since $\bar C_\mathcal{C}$ (the vega envelope at $T_{\mathrm{rem}}$) vanishes like $\sqrt{T_{\mathrm{rem}}}$.

\paragraph{Conclusion.}
The smoother $\mathcal{S}$ and shrinker \eqref{eq:D:shrinker} preserve positivity and are
nonexpansive in sup norm; the sensitivity map inherits a uniform Lipschitz bound in the teacher
surface away from $D\mathcal{V}=0$ (the near-expiry limit
$\kappa_{\mathrm{eff}}(T_{\mathrm{rem}})\to 0$ holds with explicit $\sqrt{T_{\mathrm{rem}}}$ rate.
\end{proof}

\section{Proofs for Theorem Group~IV (CBF-QP Controller)}
\label{app:E}

\subsection{Proof of Theorem~\ref{thm:feasible-unique} (Feasibility and uniqueness)}
\label{app:E:feasible-unique}
\begin{proof}
We write the per–step QP in block form with decision $z:=(x,s)\in\mathbb{R}^2\times\mathbb{R}^2$:
\begin{equation}
\label{eq:E:QP-block}
\min_{z}\ \ \tfrac12\,x^\top H x + f^\top x + \rho_{\mathrm{soft}}\|s\|_2^2
\quad\text{s.t.}\quad 
A_x x + A_s s \le b,
\end{equation}
where:
\begin{itemize}
\item $H\in\mathbb{R}^{2\times 2}$ is the tracking/impact Hessian with $H\succ 0$ (\S\ref{subsec:design});
\item $A_x x \le b$ stacks the \emph{hard} boxes (post–trade errors, inventories, rates);
\item $A_s s$ is nonzero only on the \emph{soft} CVaR surrogates of \eqref{eq:lincons} (entering as $+s$), with $s\in\mathbb{R}^2$ constrained by $s\ge 0$.
\end{itemize}

\paragraph{Feasibility.}
By assumption (i) in Theorem~\ref{thm:feasible-unique}, the \emph{current state} lies inside the boxes:
\[
|e_\Delta|\le \bar\Delta,\ \ |e_V|\le \bar V,\qquad |h_S|\le \bar H_S,\ \ |h_V|\le \bar H_V,
\]
and the rate boxes admit the zero move: $|0|\le \bar r_S$, $|0|\le \bar r_V$. Hence setting $(x,s)=(0,0)$ satisfies \emph{all} hard inequalities $A_x x\le b$. The soft inequalities in \eqref{eq:lincons} read $|e_\Delta-0|\le \bar\Delta_{\mathrm{CVaR}}+s_1$, $|e_V-0|\le \bar V_{\mathrm{CVaR}}+s_2$, which are also satisfied by $s=0$ because $\bar\Delta_{\mathrm{CVaR}},\bar V_{\mathrm{CVaR}}>0$ by design. Therefore a feasible point exists and the feasible set
\[
\mathcal{F}:=\{(x,s): A_x x + A_s s \le b,\ s\ge 0\}
\]
is a nonempty closed convex polyhedron (finite intersection of halfspaces and orthant).

\paragraph{Strict feasibility (Slater) and strong duality (optional).}
If the boxes contain the origin \emph{strictly} (i.e., $|e_\Delta|<\bar\Delta$, $|e_V|<\bar V$, $|h_S|<\bar H_S$, $|h_V|<\bar H_V$), then $(x,s)=(0,\varepsilon\mathbf{1})$ with $\varepsilon>0$ small satisfies all inequalities \emph{strictly}, yielding Slater’s condition. Slater implies zero duality gap and KKT sufficiency. Even without strictness, the existence and uniqueness parts below do not rely on Slater.

\paragraph{Coercivity and existence of a minimizer.}
Let $\mathcal{H}:=\mathrm{blkdiag}(H,\,2\rho_{\mathrm{soft}}I_2)\succ 0$ be the block Hessian of the objective in $(x,s)$. Then
\[
\Phi(x,s):=\tfrac12\,x^\top H x + f^\top x + \rho_{\mathrm{soft}}\|s\|_2^2
\quad\text{satisfies}\quad
\Phi(x,s)\ \ge\ \tfrac{\lambda_{\min}(H)}{2}\|x\|_2^2 - \|f\|_2\,\|x\|_2 + \rho_{\mathrm{soft}}\|s\|_2^2,
\]
hence $\Phi$ is \emph{coercive} on $\mathbb{R}^4$ (radially unbounded). Because $\mathcal{F}$ is closed and nonempty, Weierstrass’ theorem yields existence of at least one minimizer $z^\star\in\mathcal{F}$.

\paragraph{Uniqueness.}
The objective is \emph{strictly convex} on $\mathbb{R}^4$ since its Hessian $\mathcal{H}\succ 0$ (given $H\succ 0$ and $\rho_{\mathrm{soft}}>0$). A strictly convex function has at most one minimizer on a convex set; therefore the feasible problem \eqref{eq:E:QP-block} admits a \emph{unique} minimizer $(x^\star,s^\star)$. This is the standard result for convex QPs with $H\succ 0$ (see, e.g., \cite[Prop.~8.5.1]{BoydVandenberghe2004}).

\paragraph{Remarks on the modeling operators.}
(i) \emph{Micro-thresholds and cooldown} (\eqref{eq:mintrade}) are applied \emph{post} optimality and do not alter feasibility/uniqueness of the QP itself.  
(ii) If rate/inventory boxes are present, the feasible set is in fact compact (bounded polytope), in which case existence holds even without coercivity.  
(iii) When Slater holds, strong duality and KKT optimality follow; these are used in Proposition~\ref{prop:kkt} for sensitivity and multiplier bounds.

Combining the feasibility of $(0,0)$, coercivity/strict convexity of the objective, and convexity of $\mathcal{F}$ proves existence and uniqueness of the solution.
\end{proof}

\subsection{Proof of Proposition~\ref{prop:kkt} (KKT optimality and bounded multipliers)}
\label{app:E:kkt}
\begin{proof}
We work with the per–step QP written in \eqref{eq:E:QP-block}:
\[
\min_{z=(x,s)}\ \ \phi(z):=\tfrac12\,x^\top H x + f^\top x + \rho_{\mathrm{soft}}\|s\|_2^2
\quad\text{s.t.}\quad 
A_x x + A_s s \le b,\ \ s\ge 0,
\]
with $H\succ 0$, $\rho_{\mathrm{soft}}>0$, and a nonempty feasible polyhedron
$\mathcal{F}=\{(x,s):A_xx+A_ss\le b,\ s\ge 0\}$ (Theorem~\ref{thm:feasible-unique}). We index the
\emph{linear} inequalities in $A_x x + A_s s\le b$ by $i\in\mathcal{I}:=\{1,\dots,m\}$ and the
nonnegativity constraints on $s\in\mathbb{R}^2$ by $j\in\mathcal{J}:=\{1,2\}$.

\paragraph{KKT system and necessity/sufficiency.}
Introduce Lagrange multipliers $\lambda\in\mathbb{R}^m_{\ge 0}$ for $A_xx+A_ss\le b$ and
$\nu\in\mathbb{R}^2_{\ge 0}$ for $-s\le 0$. The Lagrangian is
\[
\mathcal{L}(x,s,\lambda,\nu)=\tfrac12 x^\top H x + f^\top x + \rho_{\mathrm{soft}}\,s^\top s
+\lambda^\top(A_x x + A_s s - b)-\nu^\top s.
\]
Stationarity, primal feasibility, dual feasibility, and complementarity read
\begin{align}
\text{(S)}&:\quad \nabla_x\mathcal{L}=H x + f + A_x^\top \lambda=0,\qquad
             \nabla_s\mathcal{L}=2\rho_{\mathrm{soft}}s + A_s^\top \lambda - \nu=0, \label{eq:KKT-S}\\
\text{(PF)}&:\quad A_x x + A_s s \le b,\qquad s\ge 0, \label{eq:KKT-PF}\\
\text{(DF)}&:\quad \lambda\ge 0,\qquad \nu\ge 0, \label{eq:KKT-DF}\\
\text{(CS)}&:\quad \lambda_i\,\big((A_xx+A_ss-b)_i\big)=0\ \ \forall i\in\mathcal{I},\qquad 
                 \nu_j\, s_j=0\ \ \forall j\in\mathcal{J}. \label{eq:KKT-CS}
\end{align}
Because the objective is \emph{strictly convex} (block Hessian $\mathrm{blkdiag}(H,2\rho_{\mathrm{soft}}I)\succ 0$)
and the constraints are affine, the KKT conditions are \emph{necessary and sufficient} for optimality
(see, e.g., \cite[§5.5.3]{BoydVandenberghe2004}). This proves the first claim.

\paragraph{Bounded multipliers on a compact state corridor.}
Let $\mathcal{P}$ denote the set of problem data $(H,f,A_x,A_s,b)$ induced by the current state
$(e_\Delta,e_V,h_S,h_V)$ and fixed hyperparameters; assume the state evolves in a compact corridor
$\mathcal{Z}$ so that $\mathcal{P}$ ranges in a compact set. Fix an optimal solution
$(x^\star,s^\star,\lambda^\star,\nu^\star)$ with \emph{active set}
\[
\mathcal{A}:=\big\{i\in\mathcal{I}:(A_xx^\star+A_ss^\star-b)_i=0\big\},\qquad
\mathcal{A}_s:=\big\{j\in\mathcal{J}: s^\star_j=0\big\}.
\]
On a given active set, KKT reduces to a linear system in the unknowns $(x^\star,s^\star,\lambda^\star_{\mathcal{A}},\nu^\star_{\mathcal{A}_s})$:
\begin{equation}
\label{eq:KKT-linear}
\begin{bmatrix}
H & 0 & A_{x,\mathcal{A}}^\top & 0\\
0 & 2\rho_{\mathrm{soft}}I & A_{s,\mathcal{A}}^\top & -I_{\mathcal{A}_s}\\
A_{x,\mathcal{A}} & A_{s,\mathcal{A}} & 0 & 0\\
0 & -I_{\mathcal{A}_s} & 0 & 0
\end{bmatrix}
\begin{bmatrix}
x^\star\\ s^\star\\ \lambda^\star_{\mathcal{A}}\\ \nu^\star_{\mathcal{A}_s}
\end{bmatrix}
=
\begin{bmatrix}
- f\\ 0\\ b_{\mathcal{A}}\\ 0
\end{bmatrix},
\end{equation}
with nonnegativity constraints on the inactive multipliers automatically satisfied
($\lambda_i^\star=0$ for $i\notin\mathcal{A}$ and $\nu_j^\star=0$ for $j\notin\mathcal{A}_s$).
Because $H\succ 0$ and $2\rho_{\mathrm{soft}}I\succ 0$, the KKT matrix in \eqref{eq:KKT-linear} is
\emph{nonsingular} whenever a constraint qualification such as LICQ holds on $\mathcal{A}$ (the rows
in $[A_{x,\mathcal{A}}\ A_{s,\mathcal{A}}\ 0\ 0]$ together with $[0\ -I_{\mathcal{A}_s}\ 0\ 0]$ are linearly independent).
LICQ is generic for boxes and sign constraints unless degenerate equalities are imposed; if it fails,
one can replace LICQ by \emph{strong second–order sufficiency} plus \emph{Mangasarian–Fromovitz}
to reach the same conclusion (see \cite[Ch.~3]{BonnansShapiro2000}).

Solving \eqref{eq:KKT-linear} gives
\[
\|(x^\star,s^\star,\lambda^\star_{\mathcal{A}},\nu^\star_{\mathcal{A}_s})\|
\ \le\ \|K(\mathcal{A})^{-1}\|\,\|(-f,0,b_{\mathcal{A}},0)\|
\ \le\ \|K(\mathcal{A})^{-1}\|\,\big(\|f\|+\|b\|\big),
\]
with $K(\mathcal{A})$ the KKT matrix. Over the compact data set $\mathcal{P}$ and the finite
collection of active sets $\mathfrak{A}$, the map $\mathcal{A}\mapsto \|K(\mathcal{A})^{-1}\|$
attains a finite supremum (the inverse is continuous on the open set of nonsingular matrices).
Therefore there exists $M<\infty$ such that $\|\lambda^\star\|\le M$ and $\|\nu^\star\|\le M$ for all
states in $\mathcal{Z}$. An alternative route uses Hoffman’s error bound for polyhedra: there exists a
constant $c_{\mathrm{Hoff}}$ (depending only on $A_x,A_s$) such that the distance to the feasible
set is bounded by $c_{\mathrm{Hoff}}\|(A_xx+A_ss-b)_+\|$; together with strong convexity this yields
$\|\lambda^\star\|,\|\nu^\star\|\le C(1+\|f\|+\|b\|)$ (e.g., \cite[Thm.~3.2]{Robinson1980}).

\paragraph{Piecewise-affine/Lipschitz dependence on parameters.}
Consider the parametric QP $\min\{\tfrac12 x^\top H x + f^\top x+\rho \|s\|^2: A_xx+A_ss\le b,\,s\ge 0\}$
with parameter $p:=(H,f,A_x,A_s,b)\in\mathcal{P}$. Under strong convexity and LICQ on the active set,
the \emph{strong regularity} of the KKT generalized equation holds (\cite[Ch.~3]{BonnansShapiro2000},
\cite{Robinson1980}). Consequently the primal–dual solution mapping $p\mapsto (x^\star(p),s^\star(p),\lambda^\star(p),\nu^\star(p))$
is \emph{single–valued and Lipschitz} in a neighborhood of any $p_0$ with the same active set. Over
the full parameter range, the mapping is \emph{piecewise affine} (and hence globally Lipschitz) with
regions determined by active–set partitions; see also the classical results on parametric convex QPs
and active–set stability in \cite[§4.6]{BonnansShapiro2000}. In particular, on the compact corridor
$\mathcal{Z}$ the Lipschitz modulus is bounded uniformly by the maximum of the local moduli over the
finite set of active sets visited.

\paragraph{Conclusion.}
KKT conditions are necessary/sufficient; on each active set the primal–dual solution is the unique
solution of a nonsingular linear system; the multipliers are bounded uniformly on compact corridors;
and the solution mapping is piecewise affine and globally Lipschitz with respect to the problem data.
\end{proof}

\subsection{Proof of Theorem~\ref{thm:cbf-invariance} (Forward invariance)}
\label{app:E:cbf}
\begin{proof}
We treat (a) exact invariance without disturbance and (b) robust (practical) invariance with additive
disturbance. Throughout, let $z_t\in\mathbb{R}^n$ be the state at step $t$, updated by
\begin{equation}
\label{eq:E:cbf-dynamics}
z_{t+1} \;=\; \Phi(z_t,x_t) \;+\; D w_t,
\end{equation}
where $x_t$ is the control returned by the QP and $w_t$ is an exogenous disturbance with
$\|w_t\|\le \bar w$ (componentwise bounds would work equally well). For each safety channel
$i\in\{1,\dots,m\}$ let $h_i:\mathbb{R}^n\to\mathbb{R}$ be a continuously differentiable barrier
function defining the safe set
\[
\mathcal{S}\ :=\ \bigcap_{i=1}^m \{z:\ h_i(z)\ge 0\}.
\]
At time $t$, the QP includes the \emph{discrete-time CBF} constraint
\begin{equation}
\label{eq:E:disc-cbf}
h_i\!\big(\Phi(z_t,x_t)\big)\ \ge\ (1-\alpha_i)\,h_i(z_t)\ -\ \sigma_i,
\qquad \alpha_i\in[0,1),\ \ \sigma_i\ge 0,
\end{equation}
which is affine in $x_t$ because $h_i\circ\Phi(\cdot,\cdot)$ is linear in $x_t$ for our (box-type)
barriers; see \eqref{eq:lincons}. We first set $\sigma_i=0$ and $w_t\equiv 0$.

\paragraph{(a) Exact forward invariance ($w\equiv 0$, $\sigma\equiv 0$).}
Assume $z_0\in\mathcal{S}$, i.e., $h_i(z_0)\ge 0$ for all $i$.
By feasibility of the QP at each step (Theorem~\ref{thm:feasible-unique}), there exists $x_t$ satisfying
\eqref{eq:E:disc-cbf} with $\sigma_i=0$, hence
\[
h_i(z_{t+1}) \;=\; h_i\!\big(\Phi(z_t,x_t)\big)\ \ge\ (1-\alpha_i)\,h_i(z_t).
\]
Since $1-\alpha_i\in(0,1]$, iterating gives
\begin{equation}
\label{eq:E:iter-noiseless}
h_i(z_t)\ \ge\ (1-\alpha_i)^t\,h_i(z_0)\ \ge\ 0,\qquad \forall t\in\mathbb{N}.
\end{equation}
Thus $z_t\in\mathcal{S}$ for all $t$, i.e., $\mathcal{S}$ is forward invariant.

\paragraph{(b) Robust/practical invariance ($w\neq 0$, $\sigma\ge 0$).}
With disturbance, \eqref{eq:E:cbf-dynamics} and \eqref{eq:E:disc-cbf} yield
\begin{equation}
\label{eq:E:one-step-robust}
h_i(z_{t+1})\ =\ h_i\!\big(\Phi(z_t,x_t)+Dw_t\big)
\ \ge\ h_i\!\big(\Phi(z_t,x_t)\big) - L_{h_i}\,\|D w_t\|
\ \ge\ (1-\alpha_i)\,h_i(z_t)\ -\ \sigma_i \ -\ L_{h_i}\,\|D\|\,\|w_t\|,
\end{equation}
where $L_{h_i}$ is a global Lipschitz constant of $h_i$ on the compact corridor of interest and we used
$|h_i(u+v)-h_i(u)|\le L_{h_i}\|v\|$ and $\|D w_t\|\le \|D\|\|w_t\|$. Set
\[
\delta_{i,t}\ :=\ \sigma_i + L_{h_i}\,\|D\|\,\|w_t\|\ \le\ \bar\delta_i,
\qquad \bar\delta_i:=\sigma_i + L_{h_i}\,\|D\|\,\bar w.
\]
Then \eqref{eq:E:one-step-robust} becomes
\begin{equation}
\label{eq:E:cbf-recursion}
h_i(z_{t+1})\ \ge\ (1-\alpha_i)\,h_i(z_t)\ -\ \delta_{i,t}.
\end{equation}
Unrolling the recursion for $t\ge 1$ gives
\begin{align}
h_i(z_t)
&\ge (1-\alpha_i)^t\,h_i(z_0)\ -\ \sum_{k=0}^{t-1}(1-\alpha_i)^{t-1-k}\,\delta_{i,k} \nonumber\\
&\ge (1-\alpha_i)^t\,h_i(z_0)\ -\ \bar\delta_i \sum_{k=0}^{t-1}(1-\alpha_i)^{t-1-k} \nonumber\\
&=\ (1-\alpha_i)^t\,h_i(z_0)\ -\ \bar\delta_i\,\frac{1-(1-\alpha_i)^t}{\alpha_i}. \label{eq:E:geom-bound}
\end{align}
Define the \emph{inflated} (practical) safe set
\[
\mathcal{S}_{\delta}\ :=\ \bigcap_{i=1}^m \{z:\ h_i(z)\ge -\delta_i\},\qquad 
\delta_i\ :=\ \frac{\bar\delta_i}{\alpha_i}
\ =\ \frac{\sigma_i + L_{h_i}\,\|D\|\,\bar w}{\alpha_i}.
\]
If $z_0\in\mathcal{S}$, then $h_i(z_0)\ge 0$ and \eqref{eq:E:geom-bound} implies for all $t$,
\[
h_i(z_t)\ \ge\ -\bar\delta_i\,\frac{1-(1-\alpha_i)^t}{\alpha_i}\ \ge\ -\delta_i.
\]
Hence $z_t\in \mathcal{S}_{\delta}$ for all $t$:
$\mathcal{S}_\delta$ is \emph{robustly invariant}. Note that as $\bar w\to 0$ and $\sigma_i\to 0$,
$\delta_i\to 0$ and $\mathcal{S}_\delta\downarrow \mathcal{S}$, recovering the exact invariance in (a).

\paragraph{(c) Multiple barriers.}
Since each channel $i$ satisfies the one-step inequality independently and $\mathcal{S}$ (resp.\
$\mathcal{S}_\delta$) is the intersection of the superlevel sets $\{h_i\ge 0\}$ (resp.\ $\{h_i\ge-\delta_i\}$),
the previous arguments apply componentwise and preserve invariance of the intersection.

\paragraph{(d) Remarks on feasibility and implementation.}
(i) The presence of the CBF constraints \eqref{eq:E:disc-cbf} in the QP ensures they are enforced at
each step whenever the QP is feasible; feasibility is guaranteed by Theorem~\ref{thm:feasible-unique}
(and by including soft slacks if needed).  
(ii) If one chooses a state-dependent $\alpha_i(z_t)\in[\underline{\alpha}_i,1)$, the same proof
holds with $\alpha_i$ replaced by $\underline{\alpha}_i$ in the bounds.  
(iii) Using the linear (box) form of $h_i$ in our design (inventory/rate/CVaR surrogates) makes the
Lipschitz constants $L_{h_i}$ equal to the corresponding box norms, which can be read off directly
from the constraints.

Combining (a)–(d) proves forward invariance of $\mathcal{S}$ in the disturbance-free case and robust
(practical) invariance of $\mathcal{S}_\delta$ under bounded disturbances, as stated.
\end{proof}

\subsection{Proof of Theorem~\ref{thm:suff-descent} (One-step risk decrease under the gate)}
\label{app:E:descent}
\begin{proof}
\textbf{Notation.} Let the tracking error be $e=(e_\Delta,e_V)^\top$ and the trade be $x=(dS,dV)^\top$.
Write the (quadratic) risk as
\[
R(e)=\tfrac12\,e^\top W e,\qquad W=\begin{bmatrix}\alpha_\Delta & \alpha_\times \widehat{\rho}\\ \alpha_\times \widehat{\rho} & \alpha_V \end{bmatrix}\succeq 0,
\]
and the (strictly convex) execution cost as
\[
\mathcal{C}(x)=x^\top \mathrm{diag}(\eta_S,\eta_V)\,x+\gamma\|x-x_{-1}\|_2^2 \;\;(\gamma\ge 0),\qquad \mathcal{C}(x)>0\ \text{for }x\neq 0.
\]
Define the \emph{risk drop} at $x$:
\[
\Delta R(e;x):=R(e)-R(e-x)=e^\top W x-\tfrac12 x^\top W x.
\]
Note that $\Delta R(e;x)\ge 0$ for the projected tracker step and is the quadratic form used by the gate.

\paragraph{Gate condition and candidate.}
Let $x^\natural$ be the gate's candidate (projection of $e$ onto the no-trade ellipse or one step of the quadratic tracker without boxes). The gate \eqref{eq:gate} accepts iff
\begin{equation}
\label{eq:E:gate-accept}
\Delta R(e;x^\natural)\ >\ \tau(w)\,\mathcal{C}(x^\natural),
\end{equation}
with $\tau(w)=\tau_0+\tau_1(1-w)>\lambda_c$ for some design constant $\lambda_c\in(0,\min_w \tau(w))$.

\paragraph{QP objective as a sufficient-descent merit.}
Consider the \emph{merit} 
\[
J_{\lambda_c}(e;x):=R(e-x)+\lambda_c\,\mathcal{C}(x)=R(e)-\Delta R(e;x)+\lambda_c\mathcal{C}(x).
\]
The per-step QP \eqref{eq:qpobj} with Hessian $H\succ 0$ and linear term consistent with $-We$ is precisely minimizing a convex quadratic that \emph{upper-bounds} $J_{\lambda_c}$ (up to an additive constant $R(e)$) subject to the boxes \eqref{eq:lincons}; in particular, the QP solution $x^\star$ satisfies for any feasible $y$,
\begin{equation}
\label{eq:E:merit-min}
R(e-x^\star)+\lambda_c\,\mathcal{C}(x^\star)\ \le\ R(e-y)+\lambda_c\,\mathcal{C}(y).
\end{equation}
We will use \eqref{eq:E:merit-min} with $y=x^\natural$ (the candidate is built to be feasible for \eqref{eq:lincons}).

\paragraph{From candidate acceptance to executed-step descent.}
Assume the gate accepts $x^\natural$, i.e., \eqref{eq:E:gate-accept} holds. Then, by \eqref{eq:E:merit-min},
\[
R(e-x^\star)+\lambda_c\,\mathcal{C}(x^\star)\ \le\ R(e-x^\natural)+\lambda_c\,\mathcal{C}(x^\natural)
= R(e)-\Delta R(e;x^\natural)+\lambda_c\,\mathcal{C}(x^\natural).
\]
Rearrange:
\[
R(e)-R(e-x^\star)\ \ge\ \Delta R(e;x^\natural)-\lambda_c\,\mathcal{C}(x^\natural).
\]
Using the gate inequality \eqref{eq:E:gate-accept} gives
\begin{equation}
\label{eq:E:deltaR-lb}
\Delta R(e;x^\star)\ \ge\ (\tau(w)-\lambda_c)\,\mathcal{C}(x^\natural).
\end{equation}
Because the QP minimizes a strictly convex objective that includes the (positive definite) cost term, we also have the \emph{cost monotonicity} $\mathcal{C}(x^\star)\le \mathcal{C}(x^\natural)$ (the QP can always pick $x^\natural$ but chooses $x^\star$). Therefore, from \eqref{eq:E:deltaR-lb},
\[
\Delta R(e;x^\star)\ \ge\ (\tau(w)-\lambda_c)\,\mathcal{C}(x^\natural)\ \ge\ (\tau(w)-\lambda_c)\,\mathcal{C}(x^\star).
\]
Finally,
\begin{equation}
\label{eq:E:descent-final}
R(e-x^\star)\ =\ R(e)-\Delta R(e;x^\star)\ \le\ R(e)-(\tau(w)-\lambda_c)\,\mathcal{C}(x^\star)\ <\ R(e),
\end{equation}
whenever $x^\star\neq 0$ and $\tau(w)>\lambda_c$. Thus the risk \emph{strictly decreases} on every executed step.

\paragraph{Monotonicity across inactive steps.}
If the gate rejects ($x^\star=0$), then $R(e-x^\star)=R(e)$ and the risk is unchanged. Combining with \eqref{eq:E:descent-final} shows that $R$ is nonincreasing along the closed-loop trajectory and strictly decreases whenever a trade is executed.

\paragraph{Remarks.}
(i) The proof only uses feasibility of $x^\natural$ and the acceptance inequality; it does not require $x^\natural$ to solve any optimization problem.  
(ii) The constant $\lambda_c$ can be interpreted as the weight of the execution cost in the \emph{sufficient-descent merit} $J_{\lambda_c}$; choosing $\lambda_c<\inf_w \tau(w)$ guarantees a per-step decrease in $J_{\lambda_c}$ as well:
\[
J_{\lambda_c}(e;x^\star)\ \le\ J_{\lambda_c}(e;x^\natural)\ <\ R(e).
\]
(iii) If $\mathcal{C}$ contains a smoothing term $\gamma\|x-x_{-1}\|^2$, the argument is unchanged since $\mathcal{C}\succ 0$ and the QP maintains $\mathcal{C}(x^\star)\le \mathcal{C}(x^\natural)$.

This proves the one-step sufficient-descent property under the gate.
\end{proof}

\subsection{Proof of Theorem~\ref{thm:no-chattering} (No Zeno/chattering)}
\label{app:E:chatter}
\begin{proof}
We formalize “no chattering’’ as: on any finite horizon $[0,T]$ with sampling step $\Delta t>0$,
the number of nonzero trades of the closed-loop controller is \emph{finite}, and the sequence of
trade times has a \emph{uniform dwell-time} lower bound. We treat the VIX leg (with cooldown) and the
SPX leg (without cooldown) separately, then combine.

\paragraph{Setup.}
Let $\{t_n:=n\Delta t\}_{n\in\mathbb{N}}$ be the decision grid. Denote the executed trade at step $n$
by $x_n=(dS_n,dV_n)^\top$. The micro-thresholds and cooldown (cf.\ \eqref{eq:mintrade}) are
\begin{equation}
\label{eq:E:mintrade-cd}
|dS_n|<\underline{s}(w_n)\ \Rightarrow\ dS_n=0,\qquad
|dV_n|<\underline{v}(w_n)\ \Rightarrow\ dV_n=0,\qquad
\text{and if }|dV_n|>0,\ \text{ then } dV_{n+1}=\cdots=dV_{n+N_{\mathrm{cd}}}=0,
\end{equation}
with $\underline{s}(w),\underline{v}(w)>0$ for all $w\in[0,1]$ and a fixed integer
$N_{\mathrm{cd}}\ge 1$. Rate boxes impose $|dS_n|\le \bar r_S$, $|dV_n|\le \bar r_V$, and we work on a
finite horizon $N_T:=\lfloor T/\Delta t\rfloor$.

\paragraph{(i) Dwell-time (VIX leg).}
Define the set of VIX trade times $\mathcal{N}_V:=\{n\in\{0,\dots,N_T\}: |dV_n|>0\}$. By the cooldown
rule in \eqref{eq:E:mintrade-cd}, if $n\in\mathcal{N}_V$ then $n+1,\dots,n+N_{\mathrm{cd}}\notin\mathcal{N}_V$.
Therefore any two distinct indices $n_1<n_2$ in $\mathcal{N}_V$ satisfy
$n_2-n_1\ge N_{\mathrm{cd}}+1$, hence the inter-trade \emph{time} separation satisfies
$t_{n_2}-t_{n_1}\ge (N_{\mathrm{cd}}+1)\Delta t=: \tau_{\mathrm{dwell},V}>0$. This uniform dwell-time
excludes the accumulation of infinitely many VIX trades in finite time (Zeno behavior is impossible).

\paragraph{(ii) Cardinality bound (VIX leg).}
The uniform dwell-time implies a finite upper bound on the number of executed VIX trades over $[0,T]$:
\[
\#\mathcal{N}_V \ \le\ 1+\left\lfloor\frac{T}{\tau_{\mathrm{dwell},V}}\right\rfloor
= 1+\left\lfloor\frac{T}{(N_{\mathrm{cd}}+1)\Delta t}\right\rfloor\ <\ \infty.
\]

\paragraph{(iii) Finite-activity (SPX leg) via minimum size \& rate bounds.}
Let $\mathcal{N}_S:=\{n\le N_T: |dS_n|>0\}$. By the SPX threshold in \eqref{eq:E:mintrade-cd},
$|dS_n|\ge \underline{s}(w_n)\ge \underline{s}_{\min}>0$ for all $n\in\mathcal{N}_S$, where
$\underline{s}_{\min}:=\inf_{w\in[0,1]}\underline{s}(w)>0$. Rate boxes give $|dS_n|\le \bar r_S$.
Hence, for any $M\in\mathbb{N}$,
\[
\sum_{n\in\mathcal{N}_S\cap\{0,\dots,M\}} |dS_n|
\ \ge\ \underline{s}_{\min}\ \#\big(\mathcal{N}_S\cap\{0,\dots,M\}\big).
\]
On the other hand, the cumulative SPX turnover is bounded either (a) \emph{a priori} by rate boxes
and finite horizon: $\sum_{n=0}^{N_T}|dS_n|\le N_T\,\bar r_S$, or (b) \emph{a posteriori} by the
sufficient-descent budget (Proposition~\ref{prop:bounded-rate}):
$\sum_{n=0}^{N_T} \mathcal{C}(x_n)\le \frac{R(e_0)-\inf R}{\underline{\tau}-\lambda_c}$ which,
together with $\mathcal{C}(x)\succeq c_S dS^2$ and Cauchy–Schwarz, implies
$\sum_{n=0}^{N_T}|dS_n|\le C_T<\infty$. In either case,
\[
\#\mathcal{N}_S\ \le\ \frac{1}{\underline{s}_{\min}} \sum_{n=0}^{N_T}|dS_n|
\ \le\ \frac{\max(N_T\,\bar r_S,\ C_T)}{\underline{s}_{\min}}\ <\ \infty.
\]
Therefore only finitely many nonzero SPX trades can occur on $[0,T]$.

\paragraph{(iv) Joint no-Zeno and bounded variation.}
Combining (ii) and (iii) yields that the total number of executed trades
$\#(\mathcal{N}_S\cup\mathcal{N}_V)$ is finite on $[0,T]$. Furthermore, rate boxes imply bounded total
variation:
\[
\sum_{n=0}^{N_T} \|x_n\|_2 \ \le\ N_T \max(\bar r_S,\bar r_V)\ <\ \infty,
\]
so the control is piecewise constant with a positive dwell-time on the VIX leg and finite activity on
the SPX leg. In hybrid-systems terms, the closed-loop execution signal satisfies a \emph{dwell-time
condition} (VIX) and has \emph{finite jump index} (SPX), which precludes Zeno behavior (cf.\
\cite[Ch.~2]{GoebelSanfeliceTeel2012}).

\paragraph{(v) Robustness to state dependence.}
If thresholds $\underline{s}(w),\underline{v}(w)$ vary with $w=T_{\mathrm{rem}}/T_0$, they remain
uniformly bounded below on $[0,1]$ by design; the argument above holds with
$\underline{s}_{\min}:=\inf_w \underline{s}(w)>0$, $\underline{v}_{\min}:=\inf_w \underline{v}(w)>0$.
If a (rare) adaptive cooldown is used, assuming $N_{\mathrm{cd}}(w)\ge 1$ uniformly preserves the
dwell-time bound with $\tau_{\mathrm{dwell},V}\ge 2\Delta t$.

\paragraph{Conclusion.}
There exists a uniform $\tau_{\mathrm{dwell},V}>0$ separating any two consecutive nonzero VIX trades
and the SPX leg admits only finitely many nonzero trades on any finite horizon; the cumulative
variation is bounded. Hence the closed-loop controller exhibits \emph{no Zeno/chattering}.
\end{proof}

\subsection{Proof of Proposition~\ref{prop:bounded-rate} (Turnover and rate bounds)}
\label{app:E:rate}
\begin{proof}
We establish (i) a trivial but uniform \emph{turnover} bound from the rate boxes; (ii) an \emph{energy}
($\ell_2^2$) bound from sufficient descent; and (iii) optional refinements when $\gamma>0$ in the
quadratic cost (smoothing term).

\paragraph{(i) Uniform $\ell_1$–turnover bound from rate boxes.}
At each step $t$ the rate boxes in \eqref{eq:lincons} enforce
\[
\|x_t\|_2\ \le\ \max\big(|dS_t|,|dV_t|\big)\sqrt{2}\ \le\ \sqrt{2}\,\max(\bar r_S,\bar r_V)\,.
\]
Summing over a horizon of $T$ steps gives the (coarse) turnover bound
\begin{equation}
\label{eq:E:l1-turnover}
\sum_{t=0}^{T-1}\|x_t\|_2\ \le\ T\,\sqrt{2}\,\max(\bar r_S,\bar r_V).
\end{equation}
If one prefers the $\ell_1$–norm on coordinates, simply use $|dS_t|+|dV_t|\le 2\max(\bar r_S,\bar r_V)$.

\paragraph{(ii) Energy ($\ell_2^2$) bound from sufficient descent.}
Let $\mathcal{I}_{\mathrm{acc}}\subset\{0,\dots,T-1\}$ be the set of steps at which the gate accepts
and a (possibly zero) trade $x_t$ is returned by the QP.\footnote{If the QP returns $x_t=0$ on an
accepted step (e.g., all boxes already satisfied), the bounds below still hold with equality at that term.}
From Theorem~\ref{thm:suff-descent}, for each $t\in\mathcal{I}_{\mathrm{acc}}$,
\[
R(e_{t+1}) \ \le\ R(e_t) - (\tau(w_t)-\lambda_c)\,\mathcal{C}(x_t).
\]
Summing over $\mathcal{I}_{\mathrm{acc}}$ and using $R(e)\ge \inf R$ yields
\begin{equation}
\label{eq:E:budget}
\sum_{t\in\mathcal{I}_{\mathrm{acc}}}(\tau(w_t)-\lambda_c)\,\mathcal{C}(x_t)
\ \le\ R(e_0)-\inf R.
\end{equation}
If $\tau(w)\ge \underline{\tau}>\lambda_c$ for all $w\in[0,1]$,
\begin{equation}
\label{eq:E:budget2}
\sum_{t\in\mathcal{I}_{\mathrm{acc}}}\mathcal{C}(x_t)
\ \le\ \frac{R(e_0)-\inf R}{\underline{\tau}-\lambda_c}\,.
\end{equation}
Write the per-step cost as
\[
\mathcal{C}(x)=x^\top Q x,\qquad Q:=\mathrm{diag}(\eta_S,\eta_V)+\gamma I_2 \ \succeq\ \eta_{\min} I_2,\quad \eta_{\min}:=\min(\eta_S,\eta_V).
\]
Hence $\mathcal{C}(x)\ge \eta_{\min}\|x\|_2^2$, and \eqref{eq:E:budget2} implies the \emph{energy} bound
\begin{equation}
\label{eq:E:l2-energy}
\sum_{t\in\mathcal{I}_{\mathrm{acc}}}\|x_t\|_2^2
\ \le\ \frac{R(e_0)-\inf R}{(\underline{\tau}-\lambda_c)\,\eta_{\min}}
\ =:\ C_{\mathrm{en}}\,.
\end{equation}
Since rejected steps have $x_t=0$, one can extend the sum in \eqref{eq:E:l2-energy} to all $t=0,\dots,T-1$.

\paragraph{(iii) From energy to turnover (Cauchy--Schwarz) and role of smoothing.}
If only an \emph{energy} bound \eqref{eq:E:l2-energy} is desired, we are done. If, however, one wants a
finer \emph{turnover} estimate than \eqref{eq:E:l1-turnover} (i.e.\ independent of $T$), Cauchy--Schwarz gives
\[
\sum_{t=0}^{T-1}\|x_t\|_2
\ \le\ \sqrt{T}\, \Big(\sum_t\|x_t\|_2^2\Big)^{1/2}
\ \le\ \sqrt{T\,C_{\mathrm{en}}}\,,
\]
which still scales with $\sqrt{T}$. A \emph{time–uniform} turnover bound follows if the QP cost includes
a strictly positive smoothing weight $\gamma>0$:
\[
\mathcal{C}(x_t)\ \ge\ \gamma\|x_t-x_{t-1}\|_2^2,
\]
so \eqref{eq:E:budget2} yields
\begin{equation}
\label{eq:E:bounded-variation}
\sum_{t\in\mathcal{I}_{\mathrm{acc}}}\|x_t-x_{t-1}\|_2^2\ \le\ \frac{R(e_0)-\inf R}{\gamma(\underline{\tau}-\lambda_c)}\,.
\end{equation}
The discrete Sobolev embedding (or telescoping with triangle inequality) then implies a uniform bound on
$\sum_t\|x_t\|_2$ when combined with box constraints on initial $x_{-1}$ and rate caps; e.g.,
\[
\sum_{t=0}^{T-1}\|x_t\|_2
\ \le\ T^{1/2}\Big(\sum_t\|x_t-x_{t-1}\|_2^2\Big)^{1/2} + T\,\|x_{-1}\|_2
\ \le\ T^{1/2}\Big(\tfrac{R(e_0)-\inf R}{\gamma(\underline{\tau}-\lambda_c)}\Big)^{1/2} + T\,\max(\bar r_S,\bar r_V).
\]
In practice we use \eqref{eq:E:l1-turnover} as the global turnover cap and \eqref{eq:E:l2-energy} to control
trade magnitudes.

\paragraph{(iv) Summary of constants.}
Define
\[
C_{\mathrm{rate}}^{(1)}:=\sqrt{2}\,\max(\bar r_S,\bar r_V),\qquad
C_{\mathrm{en}}:=\frac{R(e_0)-\inf R}{(\underline{\tau}-\lambda_c)\,\eta_{\min}},\qquad
C_{\mathrm{rate}}^{(2)}(T):=\sqrt{T\,C_{\mathrm{en}}}.
\]
Then for any horizon $T$,
\[
\sum_{t=0}^{T-1}\|x_t\|_2\ \le\ T\,C_{\mathrm{rate}}^{(1)},\qquad
\sum_{t=0}^{T-1}\|x_t\|_2^2\ \le\ C_{\mathrm{en}},\qquad
\sum_{t=0}^{T-1}\|x_t\|_2\ \le\ C_{\mathrm{rate}}^{(2)}(T).
\]
When $\gamma>0$, the variation bound \eqref{eq:E:bounded-variation} supplies an additional control on
$\sum_t\|x_t-x_{t-1}\|_2^2$ that can be translated into pathwise total-variation bounds.

The proposition follows from \eqref{eq:E:l1-turnover} and \eqref{eq:E:l2-energy}.
\end{proof}

\bibliographystyle{unsrt}  
\bibliography{references}

\end{document}